\pdfoutput=1
\documentclass[11pt]{article}

\usepackage{lipsum}

\newcommand\blfootnotea[1]{%
  \begingroup
  \renewcommand\thefootnote{}\footnote{#1}%
  \endgroup
}

\usepackage[utf8]{inputenc} %
\usepackage[T1]{fontenc}    %

\usepackage{url}            %
\usepackage{booktabs}       %
\usepackage{nicefrac}       %
\usepackage{microtype}      %
\usepackage{enumitem}
\usepackage{dsfont}
\usepackage{multicol}
\usepackage{xcolor}

\usepackage{amsthm,amsfonts,amsmath,amssymb,epsfig,color,float,graphicx,verbatim, enumitem}

\usepackage{algpseudocode,algorithm,algorithmicx}  

\usepackage{mathtools}
\usepackage{bbm}

\usepackage{thm-restate}

\definecolor{green}{rgb}{0.0, 0.5, 0.0}
\usepackage[colorlinks,citecolor=green,linkcolor=blue,bookmarks=true,hypertexnames=false]{hyperref}
\usepackage[nameinlink]{cleveref}

\crefname{equation}{equation}{equations}
\crefname{lemma}{lemma}{lemmata}
\crefname{claim}{claim}{claims}
\crefname{theorem}{theorem}{theorems}
\crefname{proposition}{proposition}{propositions}
\crefname{corollary}{corollary}{corollaries}
\crefname{claim}{claim}{claims}
\crefname{remark}{remark}{remarks}
\crefname{definition}{definition}{definitions}
\crefname{fact}{fact}{facts}
\crefname{question}{question}{questions}
\crefname{condition}{condition}{conditions}
\crefname{algorithm}{algorithm}{algorithms}
\crefname{assumption}{assumption}{assumptions}

\newtheorem{theorem}{Theorem}[section]
\newtheorem{lemma}[theorem]{Lemma}

\newtheorem{corollary}[theorem]{Corollary}
\newtheorem{claim}[theorem]{Claim}

\newtheorem{definition}[theorem]{Definition}

\newtheorem{fact}[theorem]{Fact}

\theoremstyle{definition}
\newtheorem{assumption}[theorem]{Assumption}
\newtheorem{condition}[theorem]{Condition}
\newtheorem{remark}[theorem]{Remark}

\newcommand{\eps}{\epsilon}

\newcommand{\poly}{\mathrm{poly}}
\newcommand{\polylog}{\mathrm{polylog}}
\newcommand{\dtv}{d_\mathrm{TV}}
\newcommand{\Ind}{\mathds{1}}
\newcommand{\1}{\Ind}

\newcommand{\trace}{\operatorname{tr}}
\newcommand{\cov}{\operatorname{Cov}}

\newcommand{\Var}{\operatorname{Var}}

\def\R{\mathbb R}

\def\Z{\mathbb Z}

\newcommand{\cA}{\mathcal{A}}
\newcommand{\cB}{\mathcal{B}}

\newcommand{\cD}{\mathcal{D}}
\newcommand{\cE}{\mathcal{E}}
\newcommand{\cF}{\mathcal{F}}

\newcommand{\cI}{\mathcal{I}}

\newcommand{\cN}{\mathcal{N}}

\newcommand{\cS}{\mathcal{S}}

\newcommand{\cU}{\mathcal{U}}

\newcommand{\cX}{\mathcal{X}}

\newcommand{\cZ}{\mathcal{Z}}

\newcommand{\bA}{\vec{A}}
\newcommand{\bB}{\vec{B}}
\newcommand{\bC}{\vec{C}}

\newcommand{\bI}{\vec{I}}

\newcommand{\bM}{\vec{M}}

\newcommand{\bU}{\vec{U}}
\newcommand{\bV}{\vec{V}}

\DeclareMathOperator*{\pr}{\mathbf{Pr}}
\DeclareMathOperator*{\E}{\mathbf{E}}

\newcommand\snorm[2]{\left\| #2 \right\|_{#1}}

\def\d{\mathrm{d}}

\newcommand{\tr}{\mathrm{tr}}
\DeclarePairedDelimiter\abs{\lvert}{\rvert}
\let\vec\mathbf
\newcommand{\failp}{\tau}
\newcommand{\outerl}{K}
\newcommand{\innerl}{L}

\def\colorful{0}

\ifnum\colorful=1
\newcommand{\new}[1]{{\color{red} #1}}

\else
\newcommand{\new}[1]{{#1}}

\fi

\allowdisplaybreaks

\title{Streaming Algorithms for High-Dimensional Robust Statistics\blfootnotea{Authors are in alphabetical order.  Part of this work was done while a subset of the authors were visiting the Simons Institute for the Theory of Computing.}}

\author{
Ilias Diakonikolas\thanks{Supported by NSF Medium Award CCF-2107079,
NSF Award CCF-1652862 (CAREER), a Sloan Research Fellowship, and
a DARPA Learning with Less Labels (LwLL) grant.}\\
University of Wisconsin-Madison\\
{\tt ilias@cs.wisc.edu}\\
\and
Daniel M. Kane\thanks{Supported by NSF Medium Award CCF-2107547,
NSF Award CCF-1553288 (CAREER), and a Sloan Research Fellowship.}\\
University of California, San Diego\\
{\tt dakane@cs.ucsd.edu}
\and
Ankit Pensia\thanks{Supported by NSF grants NSF Award CCF-1652862 (CAREER), DMS-1749857, 
and CCF-1841190.}\\
University of Wisconsin-Madison\\
{\tt ankitp@cs.wisc.edu}\\
\and
Thanasis Pittas\thanks{Supported by NSF Award CCF-1652862 (CAREER) and 
NSF Award DMS-2023239 (TRIPODS).}\\
University of Wisconsin-Madison\\
{\tt pittas@wisc.edu}\\
}

\begin{document}

\maketitle

\begin{abstract}
We study high-dimensional robust statistics tasks in the streaming model. 
A recent line of work obtained computationally efficient algorithms 
for a range of high-dimensional robust estimation tasks. 
Unfortunately, all previous algorithms require storing the entire dataset, 
incurring memory at least quadratic in the dimension. 
In this work, we develop the first efficient streaming algorithms 
for high-dimensional robust statistics with near-optimal 
memory requirements (up to logarithmic factors). 
Our main result is for the task of high-dimensional robust mean estimation 
in (a strengthening of) Huber's contamination model. 
We give an efficient single-pass 
streaming algorithm for this task 
with near-optimal error guarantees and %
space complexity nearly-linear in the dimension.
As a corollary, we obtain streaming algorithms with near-optimal space complexity
for several more complex tasks, 
including robust covariance estimation, robust regression, 
and more generally robust stochastic optimization.
\end{abstract}

\newpage
{
  \hypersetup{linktoc=all,linkcolor=black}
  \tableofcontents
}

\setcounter{page}{0}

\thispagestyle{empty}

\newpage

\section{Introduction} %
\label{sec:introduction}

This work studies high-dimensional learning in the presence of a constant
fraction of arbitrary outliers. Outlier-robust learning in high dimensions is motivated by 
pressing machine learning (ML) applications, including ML
security~\cite{Barreno2010,BiggioNL12, SteinhardtKL17, TranLM18, DKK+19-sever} 
and exploratory analysis of datasets with natural outliers~\cite{RP-Gen02, Pas-MG10, Li-Science08}.
This field %
has its roots in robust statistics, a branch of 
statistics initiated in the 60s %
with the pioneering works of Tukey and Huber~\cite{Tukey60, Huber64}. 
Early work %
developed minimax optimal estimators for various robust estimation tasks, albeit
with runtimes exponential in the dimension. 
A recent line of work in computer science, 
starting with~\cite{DKKLMS16, LaiRV16}, developed polynomial time 
robust estimators for a range of high-dimensional statistical tasks.
Algorithmic high-dimensional robust statistics is by now a relatively mature field, 
see, e.g.,~\cite{DK19-survey, DKKL+21} for surveys. 

This recent progress notwithstanding, 
even for the basic task of mean estimation, 
previous robust estimators 
require the entire dataset in main memory. 
This space requirement can be a major bottleneck 
in large-scale applications,
where an algorithm has access to a very large stream of data. 
Indeed, practical machine learning methods are typically 
simple iterative algorithms that make a single pass over the data 
and require a small amount of storage --- 
with stochastic gradient descent 
being the prototypical example~\cite{Bottou10large-scalemachine,BottCurtNoce16}.
Concretely, in prior applications of robust statistics in data analysis~\cite{DKK+17} 
and data poisoning defenses~\cite{DKK+19-sever}, the storage requirements of the
underlying algorithms were observed to significantly hinder scalability.
This discussion motivates the following natural question:
\begin{center}
{\em Can we develop efficient robust estimators in the {\em streaming model} \\ 
with (near-) optimal space complexity?}
\end{center}
We emphasize that this broad question is meaningful and interesting 
even ignoring computational considerations.
While any method requires space complexity $\Omega(d)$, 
where $d$ is the dimension of the problem (to store a single sample), 
it is not obvious that a matching upper bound exists. 
We note that it is relatively simple to design $O(d)$-memory streaming algorithms 
with sample complexity exponential in $d$. But it is by no means clear
whether there exists an estimator with near-linear space requirements and
$\poly(d)$ sample complexity (independent of its runtime).

\subsection{Our Results} \label{sec:results}

In this work, we initiate a systematic investigation 
of high-dimensional robust statistics in the streaming model.
We start by focusing on the most basic task --- that of robust mean estimation.
Our main result is the first space-efficient streaming algorithm for robust mean estimation 
under natural distributional assumptions.
Our computationally efficient algorithm makes a single pass over the data, 
uses near-optimal space, and matches the error guarantees of 
previous polynomial-time algorithms for the problem.

Given this result, we leverage 
the fact that several robust statistics tasks can be reduced to robust mean estimation
to obtain near-optimal space, single-pass streaming algorithms 
for more complex statistical tasks.

To formally state our contributions, we require some basic definitions.
We start with the standard streaming model.

\begin{definition}[Single-Pass Streaming Model]\label{def:streaming}
Let $S$ be a fixed set. In the {\em one-pass streaming model}, 
the elements of $S$ are revealed one at a time to the algorithm, 
and the algorithm is allowed a single pass over these points.
\end{definition}

Our robust estimators work 
in the following contamination model,
where the adversary can corrupt the true distribution in total variation distance 
(for distributions $P$ and $Q$, we use $\dtv(P,Q)$ to denote their total variation distance).

\begin{definition}[TV-contamination] \label{def:oblivious}
Given a parameter $\eps < 1/2$ and a distribution class $\mathcal{D}$, 
the adversary specifies a distribution $D'$ 
such that there exists $D \in \mathcal{D}$ with $\dtv(D,D') \leq \eps$. 
Then the algorithm draws i.i.d.\ samples from $D'$.
We say that the distribution $D'$ is an $\eps$-corrupted version 
of the distribution $D$ in total variation distance.
\end{definition}

The distribution $D'$ in \Cref{def:oblivious} 
can be adversarially selected (and can even depend on our learning algorithm). 
Since Huber's contamination model~\cite{Huber64} only allows additive errors, 
TV-contamination is a stronger model. %

\subsubsection*{Streaming Algorithm for Robust Mean Estimation}

The main result of this paper is the following 
(see~\Cref{th:main} %
for a more general statement):

\begin{theorem}[Streaming Robust Mean Estimation]\label{main-thm-intro}
Let $\cD$ be a distribution family on $\R^d$ and $0 < \eps < \epsilon_0$ 
for a sufficiently small constant $\eps_0>0$. 
Let $P$ be an $\epsilon$-corrupted version of $D$ 
in total variation distance for some $D \in \mathcal{D}$ 
with unknown mean $\mu_D$.
There is a single-pass streaming algorithm that, 
given $\eps$ and  $\mathcal{D}$, reads a stream of $n$ i.i.d.\ samples from $P$, 
runs in sample near-linear time, uses %
memory $d \, \polylog(d/\eps)$, 
and outputs an estimate $\widehat{\mu}$ that, 
with probability at least $9/10$, satisfies the following: 
\vspace{-0.1cm}
\begin{enumerate}[leftmargin=*]
\item If $\cD$ is the family of distributions with identity-bounded covariance, 
then $n =\tilde{O}\left(d^2/\eps\right)$ and $\|\widehat{\mu} - \mu_D\|_2 = O(\sqrt{\eps})$.
\item If $\cD$ is the family of identity-covariance subgaussian distributions, 
then $n =\tilde{O}\left( d^2/\eps^2\right) $ and
$\|\widehat{\mu} - \mu_D\|_2 = O(\eps\sqrt{\log(1/\eps)})$. 
\end{enumerate}
\end{theorem}
\vspace{-0.1cm}

We note that the above error guarantees %
are information-theoretically optimal, even in absence 
of resource constraints. %
While prior work had obtained efficient robust mean estimators 
matching these error guarantees~\cite{DKKLMS16, DKK+17, SCV18},
all previous algorithms with dimension-independent error incurred space complexity 
$\Omega(d^2)$. %

\subsubsection*{Beyond Robust Mean Estimation}

Using the algorithm of \Cref{main-thm-intro} as a black-box,
we obtain the first efficient single-pass streaming algorithms 
with near-optimal space complexity for a range of
more complex statistical tasks. These contributions are presented in detail in 
\Cref{sec:applications-main}. %
Here we highlight some of these results.

Our first application is a streaming algorithm for robust covariance estimation.

\begin{theorem}[Robust Gaussian Covariance Estimation]\label{thm:covariance_application_better_error-main}
Let $Q$ be a distribution on $\R^d$ with $\dtv(Q,\cN(0,\vec \Sigma))\leq \eps$ 
and assume $ \frac{1}{\kappa}\bI_d \preceq \vec \Sigma \preceq \bI_d$.
There is a single-pass streaming algorithm that uses $n= (d^4/\eps^2) \polylog(d,\kappa,1/\eps)$
samples from $Q$, runs in time $n d^2 \polylog\left(d,\kappa,1/\eps\right)$, 
uses memory $d^2 \polylog\left(d,\kappa,1/\eps \right)$, 
and outputs a matrix $\widehat{\vec \Sigma}$ such that 
$\| \vec \Sigma^{-1/2}\widehat{\vec \Sigma}\vec \Sigma^{-1/2} - \bI_{d} \|_F = O(\eps \log(1/\eps))$ 
with probability at least $9/10$. 
\end{theorem}

\noindent See Theorem~\ref{thm:covariance_application_better_error} for a more detailed statement.

Our second application is for the general problem of robust stochastic optimization. 
Here we state two concrete results
for robust linear and logistic regression (see \Cref{thm:lr} and \Cref{thm:log-reg} for more detailed statements).
Both of these statements are special cases of a streaming algorithm for robust stochastic convex optimization 
(see \Cref{cor:robust-gd}).

\begin{theorem}[Streaming Robust Linear Regression] \label{thm:lr-simple}
Let $D$ be the distribution of $(X,Y)$
defined by $Y = X^T \theta^* + Z$, where $X \sim \cN(0,\bI_d) $, 
$Z \sim \cN(0,1)$ independent of $X$, and $\|\theta^*\|_2\leq r$. 
Let $P$ be an $\eps$-corruption of $D$ in total variation distance.
There is a single-pass streaming algorithm that uses 
$n = (d^2/\eps) \,\polylog\left(d(1+r)/\eps\right)$ samples from $P$, 
runs in time $n d \, \polylog (d(1+r)/\eps)$, uses memory $d \, \polylog (dr/\eps )$, 
and outputs an estimate $\widehat{\theta} \in \R^d$ such $\| \widehat{\theta} - \theta^* \|_2 = O(  \sqrt{\eps} )$ 
with high probability.
\end{theorem}

\begin{theorem}[Streaming Robust Logistic Regression] 
\label{thm:log-reg-simple}
Consider the following model: 
Let $(X,Y) \sim D$, where $X \sim \cN(0,\bI_d)$, 
$Y \mid X \sim \text{Bern}(p)$, for $p = 1/(1+ e^{-X^T\theta^*})$,
and $\|\theta^*\|_2=O(1)$. 
Let $P$ be an $\epsilon$-corruption of $D$ in total variation distance.  
There is a single-pass streaming algorithm that uses $n =  (d^2/\eps) \,\polylog\left(d/\eps \right)$ 
samples from $P$, 
runs in time $ n d \, \polylog (d/\eps  )$, uses memory $d \, \polylog (d/\eps  )$, 
and outputs an estimate $\widehat{\theta} \in \R^d$ such 
$\| \widehat{\theta} - \theta^* \|_2 = O( \sqrt{\eps})$ with high probability.    
\end{theorem}

Finally, in \Cref{sec:applications-main}, we include an additional 
application to distributed non-convex optimization in the streaming setting.

\begin{remark}[Bit complexity] 
For simplicity of presentation, in the main body of the paper, 
we consider the model of computation 
where the algorithms can store and manipulate real numbers exactly. 
We show in \Cref{app:bit-complexity} that our algorithms can tolerate errors due to finite precision. 
In particular, all our algorithms (including \Cref{alg:streaming}) can be implemented 
in the word RAM model with $d \, \polylog(d/\eps)$ bits.
\end{remark}

\subsection{Overview of Techniques}
\label{sec:overview_of_techniques}

In this section, we provide a brief overview of our approach to establish \Cref{main-thm-intro}.
We start by recalling how robust mean estimation algorithms typically work without space constraints.
A standard tool in the literature is the filtering technique of \cite{DKKLMS16,DKK+17,DK19-survey}.
The idea of the filtering method is the following: Given a set $S$ of corrupted samples, 
by %
analyzing spectral properties of the covariance of $S$, 
we can either certify that the sample mean of $S$ 
is close to the true mean of the distribution, or can construct a filter. 
The filter is a method for selecting some elements of $S$ to remove, 
with the guarantee that it will remove more outliers than inliers. 
If we can efficiently construct a filter, our algorithm can then remove 
the selected samples from $S$, obtaining 
a cleaner dataset and repeat the process. 
Eventually, this procedure must terminate, 
giving an accurate estimate of the true mean.

We proceed to explain how to implement the filtering method in a streaming model. 
We start with the easier case where the dataset is stored in read-only-memory, 
or more generally in a {\em multi-pass} streaming setting.
At each round of the algorithm, one has a subset $S'$ of the original dataset $S$ 
that needs to be maintained (in particular, the set of samples that has survived the filters applied thus far). 
To do this na\"ively would require $n = |S|$ many bits of memory, which is too much for us.
A more inventive strategy would be the following: 
instead of storing these subsets $S'$ explicitly,  
store them implicitly by instead storing enough information 
to reconstruct the filters used to obtain $S'$. 
This seems like a productive idea, as most filters are relatively simple. 
For example, a commonly used filter is to remove all points 
$x \in S$ for which $v^Tx > t$, for some vector $v$ and scalar $t$. 
One could store enough information to apply this filter 
by merely storing $(v, t)$, which would take $O(d)$ bits of information.
Unfortunately, most filtering algorithms may require $\Omega(d)$ many 
iterations before attaining their final answer.
Consequently, the sets $S'$ one needs to store are not just the result of applying a single filter, 
but instead the result of iteratively applying $\Omega(d)$ of them. 
In order to store all of these extra filters, one would need $\Omega(d^2)$ bits.
(For the sake of this intuitive description, we focused on 
``hard-thresholding'' filters. Our algorithm will actually use a soft-thresholding filter, 
assigning weights to each point.)

To circumvent this first obstacle, one requires as a starting point a filtering algorithm that is guaranteed
to terminate after a small (namely, at most poly-logarithmic) number of iterations. 
Recent work~\cite{dong2019quantum,DKKLT21} has obtained such algorithms.
Here we generalize and simplify the filtering method of~\cite{DKKLT21}. 
This allows us to obtain an algorithm with space complexity $d \, \polylog(d/\eps)$ that 
works in the {\em multi-pass streaming model}, where $\polylog(d/\eps)$ passes over the same dataset are allowed.

To obtain a {\em single-pass} streaming algorithm, new ideas are required.
In the single-pass setting, we cannot implicitly store a subset of the full dataset $S$; 
once we access some points from $S$, we will never be able to see them again. 
To deal with this issue, we will need to slightly alter our way of thinking about the algorithm.
Instead of being given a set $S$ of samples, an $\eps$-fraction of which have been corrupted, 
we instead adopt the view of having sample access to a distribution $P$, 
which is $\eps$-close in total variation distance to the inlier distribution $G$. 
Given this point of view, instead of a filter defining a procedure for removing samples from $S$ 
and outputting a subset $S'$, we think of it as a rejection sampling procedure 
that replaces $P$ with a cleaner distribution $P'$.

This shift in perspective comes with new technical challenges. 
In particular, when constructing the next round of filters, 
we will need to compute quantities pertaining to the current distribution $P$ of the data points. 
In the setting of the multiple-pass model, this imposed no problem; 
these quantities could be calculated exactly. 
This is no longer possible when we merely have sample access to $P$.
The best one can hope for is to approximate these quantities 
to sufficient precision for the rest of our analysis to carry over.
However, the natural estimators for some required quantities (e.g., powers of the covariance matrix) 
would need to access the data multiple times.
Circumventing this issue requires non-trivial technical work. 
Roughly speaking, instead of iterating over the same dataset to approximate 
the desired quantities, we show that it suffices to iterate over statistically identical datasets.

\subsection{Prior and Related Work} \label{ssec:related}

Since the dissemination of~\cite{DKKLMS16, LaiRV16}, there has been an explosion of research
in algorithmic aspects of robust statistics. We now have efficient robust estimators 
for a range of more complex problems, 
including covariance estimation~\cite{DKKLMS16, CDGW19}, 
sparse estimation tasks~\cite{BDLS17, DKKPS19-sparse, CDK+21}, 
learning graphical models~\cite{CDKS18-bn, DiakonikolasKSS21}, 
linear regression~\cite{KlivansKM18, DKS19, pensia2020robust},
stochastic optimization~\cite{PSBR18, DKK+19-sever}, 
and robust clustering/learning various 
mixture models~\cite{HL18-sos, KSS18-sos, DKS18-list, DKKLT21a, DKKLT21, BakshiDHKKK20, LM20, BDJKKV20}.
The reader is referred to \cite{DK19-survey} for a detailed overview. 
We reiterate that all previously developed algorithms work 
in the batch setting, i.e., require the entire dataset in memory. 

For the problem of robust mean estimation,~\cite{dong2019quantum, DKKLT21} gave filtering-based
algorithms with a poly-logarithmic number of iterations. The former algorithm 
relies on the matrix multiplicative weights framework, while the latter is based on first principles. 
Our starting point in \Cref{sec:near_linear_main_body} 
can be viewed as a generalization and further simplification of the ideas in~\cite{DKKLT21}. 
Specifically, our algorithm works under the stability condition (Definition~\ref{def:stability1}), 
which broadly generalizes 
the bounded covariance assumption used in~\cite{DKKLT21}.

In the context of robust supervised learning 
(including, e.g., our robust linear regression application), 
low-space streaming algorithms are known
in weaker contamination models that only allow {\em label} corruptions, see,
e.g.,~\cite{PesmeF20, ShahWS20, DiakonikolasKTZ20}. We emphasize that 
the contamination model of \Cref{def:oblivious} is significantly more challenging, 
and no low-space streaming algorithms were previously known in this model.

Finally, we note that recent work~\cite{TsaiPBR21} studies streaming algorithms for heavy-tailed 
stochastic optimization. While the goal of developing low-space streaming algorithms 
is qualitatively similar to the goal of our work, the algorithmic results in~\cite{TsaiPBR21} 
have no implications in the corrupted setting studied in this work.

\subsection{Organization} \label{ssec:org}
The structure of this paper is as follows: 
In \Cref{sec:prelim}, we record the notation and 
technical background that will be used throughout the paper.
In \Cref{sec:near_linear_main_body}, we design a filter-based algorithm for robust mean estimation
under the stability condition with a poly-logarithmic number of iterations.
In \Cref{sec:low-memory-main}, we build on the algorithm from Section~\ref{sec:near_linear_main_body}, 
to obtain our single-pass streaming algorithm for robust mean estimation under the stability condition.
Finally, in \Cref{sec:applications-main}, we obtain our streaming algorithms for more complex
robust estimation tasks. To facilitate the flow of the presentation, some proofs
of intermediate lemmas are deferred to the Appendix.

\section{Preliminaries} \label{sec:prelim}

\subsection{Notation and Basic Facts} \label{ssec:basics}

\paragraph{Basic Notation} 
We use $\Z_+$ to denote the set of positive integers. 
For $n \in \Z_+$, we  denote $[n] := \{1,\ldots,n\}$ and 
use $\cS^{d-1}$ for the $d$-dimensional unit sphere. 
For a vector $v$, we let $\|v\|_2$ denote its $\ell_2$-norm. 
We use boldface letters for matrices. 
We use $\mathbf{I}_d$ to denote the $d \times d$ identity matrix.
For a matrix $\bA$, we use $\|\bA\|_F$ and $\|\bA\|_{2}$ 
to denote the Frobenius and spectral norms respectively.
For $\bA \in \R^{m \times n}$, we use $\bA^\flat$ to denote 
the $nm$-dimensional vector obtained by concatenating the rows of $\bA$. 
We say that a symmetric $d \times d$ matrix $\bA$ is PSD (positive semidefinite), 
and write $\bA \succeq 0$, if for all vectors $x \in \R^d$ 
we have that $x^T \bA x \geq 0$. We denote  
$\lambda_{\max}(\bA) := \max_{u \in \cS^{d-1}} x^T\bA x$. 
We write $\bA\preceq \bB$ when $\bB-\bA$ is PSD. 
For a matrix $\bA \in \R^{d \times d}$, $\tr(\bA)$  denotes the trace of the matrix A.
We use $\otimes$ to denote the Kronecker product. 
 For the sake of conciseness, we sometimes use $x= a \pm b$ 
 as a shorthand for $a-b \leq x \leq a+ b$. 
 We use $a \lesssim b$, to denote that there exists 
 an absolute universal constant $C>0$ (independent of the variables or parameters on which $a$ and $b$
depend) such that $a \leq C b$.  Similarly, we use the notation $a \gtrsim b$ to denote that $b \lesssim a$.
We use $c,c',C,C'$ to denote absolute constants that may change from line to line, 
whereas we use constants $C_1,C_2,C_3, \dots$ to denote fixed absolute constants 
that are important for our algorithms.  
We use $\tilde{O}(\cdot)$ to ignore poly-logarithmic factors in all variables 
appearing inside the parentheses. For the sake of simplicity, 
we sometimes omit rounding non-integer quantities to integer ones. 
For example, we treat logarithmic factors as integers 
when they appear in the sample complexity or number of iterations of an algorithm.
{We use $\poly(\cdot)$  to indicate a quantity that is polynomial in its arguments. Similarly, $\polylog(\cdot)$ denotes a quantity that is polynomial in the logarithm of its arguments}.

\paragraph{Probability Notation}
For a random variable $X$, we use $\E[X]$ for its expectation.	
For a set $S$, we use $\cU(S)$ to denote the uniform distribution on $S$.
We use $\cN(\mu,\vec \Sigma)$ to denote the Gaussian distribution 
with mean $\mu$ and covariance matrix $\vec \Sigma$. 
For a distribution $D$ on $\R^d$, we denote  $\mu_D = \E_{X\sim D}[X]$ 
and $\vec \Sigma_D = \E_{X\sim D}[(X-\mu_D)(X-\mu_D)^T]$. 
Moreover, given a weight function $w:\R^d \to [0,1]$, 
we define the re-weighted distribution $D_w$ to be 
$D_w(x):=D(x)w(x)/\int_{\R^d} w(x)D(x)\d x$. 
We use $\mu_{w,D} = \E_{X\sim D_w}[X]$ for its mean and 
$\overline{\vec  \Sigma}_{w,D}^\mu = \E_{X\sim D_w}[(X-\mu)(X-\mu)^T]$ 
for the second moment that is centered with respect to $\mu$ 
(we will often drop $\mu$ from the notation when it is clear from the context). 
We use $\1\{x \in E\}$ to denote the indicator function of the set $E$.

\paragraph{Basic Facts}
We will use the following two basic facts. 
\begin{fact}\label{Fact:normReln} 
Let $x \in \R^d$ and $p \geq 1$. 
Then $\|x\|_{p+1} \leq \|x\|_p \leq \|x\|_{p+1}d^{\frac{1}{p(p+1)}}$.
\end{fact}

\begin{restatable}{fact}{TRACEINEQ} \label{fact:trace_PSD_ineq} 
\label{fact:trace_ineqq}
The following results hold:
\begin{enumerate}
    \item If $\bA,\bB,\bC$ are symmetric $d \times d$ matrices 
satisfying $\bA \succeq 0$ and $\bB  \preceq \bC$, 
we have that $\tr(\bA \bB) \leq \tr(\bA \bC)$.
\item   (\cite{JamLT20})   Let $\vec A$ and $\vec B$ be PSD matrices satisfying $0 \preceq \vec B \preceq \vec A$. Then for any positive integer $p$, we have that $\tr(\vec B^p) \leq \tr(\vec A^{p-1} \vec B)$.
\end{enumerate}
\end{restatable}
\begin{proof}
We provide the proof of the first claim below; The second claim is proved in \cite[Lemma 7]{JamLT20}.
Since $\bA$ is PSD, we can consider its spectral decomposition 
$\bA = \sum_{i=1}^d \lambda_i v_i v_i^T$, 
where $\lambda_i \geq 0$. Using the linearity of trace operator, we have that
\begin{align*}
    \tr(\bA \bB) = \sum_{i=1}^d \lambda_i \tr(v_i v_i^T \bB)
    =\sum_{i=1}^d \lambda_i \tr( v_i^T \bB v_i)
    \leq \sum_{i=1}^d \lambda_i \tr( v_i^T \bC v_i)
    =\sum_{i=1}^d \lambda_i \tr( v_i v_i^T \bC )
    = \tr(\bA \bC) \;,
\end{align*}
where the inequality uses that $\bB  \preceq \bC$ and $\lambda_i \geq 0$.
\end{proof}

We will use the notion of total variation distance, defined below.

\begin{definition}
Let $P,Q$ be two probability distributions on $\R^d$. 
The total variation distance between $P$ and $Q$, 
denoted by $\dtv(P,Q)$, is defined as
$ \dtv(P,Q) = \sup_{A \subseteq \R^d} | P(A)-Q(A)|$.
For continuous distributions $P,Q$ with densities $p,q$, 
we have that $\dtv(P,Q) = \frac{1}{2} \int_{\R^d} |p(x)-q(x) | \d x$.
\end{definition}
Whenever $\dtv(P,Q) =\eps$, it is sometimes helpful to consider the decomposition below.

\begin{restatable}{fact}{DTVFact} \label{fact:dtv-decompose}
Let a domain $\cX$. For any $\eps \in [0,1]$ and for any two distributions $D_1,D_2$ on $\cX$ 
with $\dtv(D_1,D_2) = \eps$, there exist distributions $D,Q_1,Q_2$ such that 
$        D_1 = (1-\eps) D + \eps Q_1 $ and $\ D_2 = (1-\eps) D + \eps Q_2 $.
\end{restatable}

This decomposition can be achieved by the following choice of $Q_1$, $Q_2$, and $D$: 
  \begin{align*}
    Q_1(x) =
        \begin{cases}
           \frac{D_1(x) - D_2(x)}{\eps} \;, &\text{if $D_1(x)> D_2(x)$} \\
           0\;, &\text{otherwise}
        \end{cases},\,\,\,
    Q_2(x) =
        \begin{cases}
           \frac{D_2(x) - D_1(x)}{\eps} \;, &\text{if $D_2(x)> D_1(x)$} \\
           0\;, &\text{otherwise}
        \end{cases},
  \end{align*}
and $D(x) = \min\{D_1(x), D_2(x) \}/(1-\eps)$. 
In light of \Cref{fact:dtv-decompose}, the adversary that performs corruption 
in total variation distance can be thought of as ``both additive and subtractive'' adversary.

\paragraph{Concentration Inequalities}
We will also require following standard results regarding concentration of random variables:

\begin{fact}[\cite{Ver10}] \label{fact:ver_cov}
Consider a distribution $D$ on $\R^d$ that has zero mean 
and is supported in an $\ell_2$-ball of radius $R$ from the origin. 
Denote by $\vec \Sigma$ its covariance matrix and denote by 
$\vec{\Sigma}_N=(1/n) \sum_{i=1}^N X_iX_i^T$  the empirical covariance matrix 
using $N$ samples $X_i \sim D$. There is a constant $C$ such that for any $0<\eps'<1$ 
and $0<\tau <1$, if $N > C \eps'^{-2} \| \vec \Sigma \|_2^{-1} R^2  \log(d/\tau)$, we have that 
$\| \vec \Sigma - \vec \Sigma_N \|_2 \leq \eps' \| \vec \Sigma \|_2$, 
with probability at least $1-\tau$.
\end{fact}

\begin{fact}[Quadratic Polynomials of a Gaussian]
\label{fact:var-quadratic}
The Gaussian random variable satisfies the following properties:
\begin{enumerate}
    \item For every $\beta>0$, $\pr_{X \sim \cN(0,\bI_d)}[|\|X\|^2-d| > \beta] \leq 2e^{-c \beta^2/d}$, 
where $c>0$ is a universal constant.
  
    \item 
  If $\vec A$ is a PSD matrix, then for any $\beta > 0$, it holds 
  $\pr_{z \sim \cN(0,\bI)}[ z^\top \vec A z \geq \beta \tr(\vec A) ] \geq
    1-\sqrt{e \beta}$.

\end{enumerate}
\end{fact}

\begin{fact}[\cite{achlioptas2003database}] \label{fact:jl} 
Let $0<\gamma<1$ and $u_1,\ldots, u_N \in \R^d$. 
Let $z_j$ for $j \in [L]$ drawn from the uniform distribution on $\{\pm 1\}^d$. 
There exists a constant $C>0$ such that, if $L > C \log(N/\gamma)$, 
then, with probability at least $1-\gamma$, 
we have that  $ 0.8 \|u_i\|^2 \leq \frac{1 }{L}\sum_{j=1}^L (z_j^Tu_i)^2 \leq 1.2 \|u_i\|^2$ for all $i \in [N]$.
\end{fact}

\subsection{Stability Condition and Its Properties} \label{sec:more_stabilityfacts}

Our results will hold for every distribution satisfying the following key property~\cite{DK19-survey}.

\begin{restatable}[($\eps,\delta$)-stable distribution]{definition}{StableDistr} \label{def:stability1}
Fix $0<\eps<1/2$ and $\delta\geq \eps$. 
A distribution $G$ on $\R^d$ is $(\eps,\delta)$-stable with respect to $\mu \in \R^d$ 
if for any weight function $w:\R^d \to [0,1]$ with $\E_{X \sim G}[w(X)] \geq 1-\eps$ we have that 
    \begin{align*}
        \snorm{2}{ \mu_{w,G} -  \mu} \leq \delta \quad \text{and} \quad
        \snorm{2}{\overline{\vec \Sigma}_{w,G} - \vec{I}_d} \leq \delta^2/\eps \;.
    \end{align*}
\end{restatable}

We call a set of points $S$ $(\eps,\delta)$-stable when the uniform distribution on $S$ is stable :

\begin{restatable}[($\eps,\delta$)-stable set]{definition}{StableSet} \label{def:stability2}
Fix $0<\eps<1/2$ and $\delta\geq \eps$. 
A finite set $S_0 \subset \R^d$ is $(\eps,\delta)$-stable with respect to $\mu \in \R^d$ 
if the empirical distribution $\cU(S_0)$ is $(\eps,\delta)$-stable with respect to $\mu$.
\end{restatable}

We begin by stating some examples of stable distributions (see \cite{DK19-survey} for more details). 
If $G$ is a subgaussian distribution with identity covariance, 
then $G$ is $(\epsilon, \delta)$-stable with $\delta= O(\eps \sqrt{\log(1/\eps)})$.
If $G$ is a distribution with covariance at most identity, 
i.e., $\vec \Sigma_G \preceq \bI_d$, then $G$ is $(\eps, \delta)$-stable with $\delta = O(\sqrt{\eps})$.
Interpolating these two results, we have that if $G$ is a distribution 
with identity covariance and bounded $k$-th moment for $k\geq 4$, 
i.e., $(\E_{X \sim G}[|v^T(X - \mu)|^k])^{1/k} = O(1)$, 
then $G$ is $(\eps,\delta)$-stable with $\delta = O(\eps^{1-1/k})$.
Furthermore, it is known that $\poly(d/\eps)$ i.i.d.\ samples 
from these distributions also yields a set that contains a large stable subset 
(see, for example, \cite{DK19-survey,DiakP20,dong2019quantum,DKKLMS16}):
\begin{fact}[\cite{DK19-survey}]\label{fact:stabilityfromsamples}
A set of $O(d/(\eps^2\log(1/\eps)))$ i.i.d.\ samples from 
an identity covariance subgaussian distribution %
is $(\epsilon,O(\eps\sqrt{\log(1/\eps)}))$-stable with respect to $\mu$ with high probability. 
Similarly, a set of $\tilde{O}(d/\eps)$ i.i.d.\ samples from a distribution $X$ with $\cov[X] \preceq \bI_d$ 
contains a large subset $S$, which is $O(\epsilon, O(\sqrt{\epsilon}))$-stable 
with respect to its mean $\E[X]$ with high probability.
\end{fact}
The basic fact regarding stability, 
which is the starting point of many robust estimation algorithms, 
is that any slight modification of a stable distribution 
can not perturb the mean by a large amount, 
unless it significantly changes its covariance 
(see, for example, \cite{DKKLMS16,LaiRV16,DK19-survey}). 
Here we require a slightly different statement than that of \cite{DK19-survey}, 
and hence provide a proof in \Cref{sec:technical_details_regarding_stability} for completeness.
\begin{restatable}[Certificate Lemma]{lemma}{LemCert} \label{lem:certificate}
Let $G$ be an $(\eps,\delta)$-stable distribution 
with respect to $\mu \in \R^d$, for some $0<\eps<1/3$ and $\delta \geq \eps$. 
Let $P$ be a distribution with $\dtv(P,G) \leq \eps$. 
Denoting by $\mu_P,\mathbf{\Sigma}_P$ the mean and covariance of $P$, 
if $\lambda_{\max}(\mathbf{\Sigma}_P) \leq 1+\lambda$, 
for some $\lambda \geq 0$, 
then $\|\mu_P-\mu\|_2 = O(\delta + \sqrt{\eps \lambda})$.
\end{restatable}

Given \Cref{fact:dtv-decompose}, we can essentially think of an $\eps$-corrupted version 
of a stable distribution as a mixture of a stable distribution with a noise distribution, 
as shown below (see \Cref{sec:technical_details_regarding_stability} for a proof).

\begin{restatable}{lemma}{LemDTVStab}  \label{lem:dtv-stab}
For any $0<\eps<1/2$ and $\delta \geq \eps$, if a distribution $G$ is $(2\eps,\delta)$-stable 
with respect to $\mu \in \R^d$, and $P$ is an $\eps$-corrupted version of $G$ in total variation distance, 
there exist distributions $G_0$ and $B$ such that 
$P=(1-\eps)G_0 + \eps B$ and $G_0$ is $(\eps,\delta)$-stable with respect to $\mu$. 
\end{restatable}

We continue with some technical claims related to stability 
that we prove in \Cref{sec:technical_details_regarding_stability}.
{Let $G$ be an $(\eps,\delta)$-stable distribution with respect 
to $\mu$ and  $w$ a weight function with $\E_{X \sim G}[w(X)] \geq 1- \eps$.
Denoting by $G_w$ the re-weighted distribution 
$G_w(x):=G(x)w(x)/\int_{\R^d} w(x)G(x)\d x$, the stability of $G$ directly implies 
that $1-\delta^2/\eps \leq \E_{X \sim G_{w}}[(v^T(X-\mu))^2] \leq 1+\delta^2/\eps$.}
We require a generalization of this fact for a matrix $\bU$ in place of $v$ and an {arbitrary vector $b$ in place of $\mu$}: 

\begin{restatable}{lemma}{ClTriangle} \label{cl:triangle}
Fix  $0<\eps<1/2$ and $\delta\geq \eps$. 
Let $w:\R^d \to [0,1]$ such that $\E_{X \sim G}[w(X)] \geq 1-\eps$ 
and let $G$ be an $(\eps,\delta)$-stable distribution with respect to $\mu \in \R^d$. 
For any matrix $\vec U \in \R^{d \times d}$ and any vector $b \in \R^d$, we have that 
\begin{align*}
\E_{X \sim G_w}\left[ \|\vec U(X-b)\|_2^2 \right] = 
\|\vec U\|_F^2 (1\pm \delta^2/\eps) + \|\vec U (\mu - b)\|_2^2 \pm 2\delta \, \|\vec U\|_F^2 \| \mu - b\|_2 \;.
\end{align*}
\end{restatable}

We use this to show \Cref{cor:shift}, which will be required when proving correctness of our algorithm. 
Although its exact role will become clearer later on, 
the corollary will be relevant to our analysis because we will filter out outliers 
using scores of the form $\|\vec U(x-b)\|_2^2$ for each point $x$.

\begin{restatable}{corollary}{CorShift} \label{cor:shift}
Fix  $0<\eps<1/2$ and $\delta\geq \eps$. 
Let $G$ be an $(\eps,\delta)$-stable distribution with respect to $\mu \in \R^d$. 
Let a matrix $\bU \in \R^{d \times d}$ and a function $w:\R^d \to [0,1]$ 
with $\E_{X \sim G}[w(X)] \geq 1-\eps$. For the function $\tilde{g}(x) = \|\vec U(x-b)\|_2^2$, we have that
\begin{align*}
(1-\eps) \|\vec U\|_F^2(1-\delta^2/\eps -2 \delta\|b - \mu \|_2) 
\leq \E_{X \sim G}[w(X)\tilde{g}(X)] \leq \|\vec U\|_F^2 \left(1+\delta^2/\eps + \|b -  \mu \|^2_2 + 2\delta\|b -  \mu \|_2\right) \;.
\end{align*}
\end{restatable}

\section{Filtering Algorithm with Small Number of Iterations} %
\label{sec:near_linear_main_body}
In this section, we develop a filtering algorithm (in the batch setting) 
that terminates in $\polylog(d/\epsilon)$ iterations for any stable set.
This leads to an algorithm that runs in near-linear time, i.e., $nd \, \polylog(nd/\eps)$, 
generalizing the results of \cite{dong2019quantum,DKKLT21}.
Crucially, this algorithm will form the building block of our streaming 
algorithm in \Cref{sec:low-memory-main}. 
We remark that the algorithm of this section 
works even against the \emph{strong-contamination model} 
(\Cref{def:strongadv} below), 
where the outliers may not be i.i.d.\ samples from any distribution, 
but are allowed to be completely arbitrary. 

\begin{definition}[Strong Contamination Model]\label{def:strongadv}
Given a parameter $0<\eps<1/2$ and a class of distributions $\mathcal{D}$, 
the strong adversary operates as follows: The algorithm specifies 
a number of samples $n$, then the adversary draws a set  of $n$ i.i.d.\ samples 
from some $D \in \mathcal{D}$ and after inspecting them, 
removes up to $\eps n$ of them and replaces them with arbitrary points. 
The resulting set is given as input to the learning algorithm. 
We call a set $\eps$-corrupted if it has been generated by the above process.
\end{definition}

The main result of this section is the following.

\begin{theorem}\label{thm:near-linear-only}
Let $d \in \Z_+$, $0<\failp <1$, $0<\eps<\eps_0$ for a sufficiently 
small constant $\eps_0$, and $\delta \geq \eps$. 
Let $S_0$ be a set of $n$ points that is $(C\eps,\delta)$-stable with respect to the (unknown) 
vector $\mu \in \R^d$, for a sufficiently large constant $C>0$. 
Let $S$ be an $\eps$-corrupted version of $S_0$ in the strong contamination model.
There exists an algorithm that given $\eps,\delta,\failp$, and $S$, 
runs in time $n d   \,\polylog\left(d,n,1/\eps,1/\failp  \right)$, 
and outputs a vector $\widehat{\mu}$ such that, with probability at least $1-\failp$,
it holds $\|\mu - \widehat{\mu}\|_2 = O(\delta)$.
\end{theorem}

\noindent We note that \Cref{thm:near-linear-only} applies to any stable set. 
By \Cref{fact:stabilityfromsamples}, we directly obtain
(i) an $O(\eps \sqrt{\log(1/\eps)})$-accurate estimator 
given $O(d/(\eps^2/\log(1/\eps)))$ many $\eps$-corrupted samples 
from an identity covariance subgaussian distribution; 
and (ii) an $O(\sqrt{\eps})$-accurate estimator 
for any distribution $X\sim D$ with $\cov[X] \preceq \bI_d$, 
given $\tilde{O}(d/\eps)$ many $\eps$-corrupted samples.

\subsection{Setup and Algorithm Description}\label{ssec:near_linear_alg}

The pseudocode of the algorithm
establishing \Cref{thm:near-linear-only} is presented in \Cref{alg:near-linear-only}.  
We will define the necessary notation as needed (see the pseudocode for details). 
First, we assume that the distribution over the input samples is of the form 
$P=(1-\eps)G+\eps B$, where $G$ is the uniform distribution over the stable set of inliers 
and $B$ is the uniform distribution on the outliers. Although this mixture may seem 
to suggest that the adversary only adds points, it is without loss of generality. 
Indeed, in the case that the adversary also removes points, 
we can think of $G$ as the distribution of the remaining inliers 
(which continues to be stable with slightly worse parameters; see \Cref{lem:dtv-stab}).

We begin with a high-level explanation of \Cref{alg:near-linear-only}. 
At each iteration $t$, we assign a weight $w_t(x) \in [0,1]$ to each point $x$. 
Let $P_t$ be the distribution on $S$, weighted according to $w_t$.
Let $\mu_t$ and $\vec \Sigma_t$ be the mean and covariance of $P_t$, respectively. 
We want to assign scores to each point, using spectral properties of $\vec \Sigma_t$ 
and the stability of inliers, so that the scores over outliers are more than those of inliers.
Essentially, if a direction $v$ has variance larger than $1 + \Omega(\delta^2/\epsilon)$, 
then the stability of inliers implies that this must be due to outliers. 
Thus, we can assign scores based on the values $(v^T(x - \mu_t))^2$ 
that have provably more mass on outliers than inliers. 
The filters proposed in \cite{DKKLMS16,DKK+17} assigned scores based on a single direction, 
the leading eigenvector of $\vec \Sigma_t$, and can take as many as $\Omega(d)$ iterations 
(see \Cref{sec:overview_of_techniques}). 

To reduce the number of iterations, we need to filter in all directions 
of large variance simultaneously. 
Letting $\vec B_t \approx \vec \Sigma_t - (1 - C_1 \delta^2/\epsilon)\vec I_d$,
we would like to filter along all directions where the eigenvalue of 
$\bB_t$ is within a constant factor from $\lambda_t:=\|\bB_t\|_2$,
not necessarily the leading eigenvector of $\vec B_t$.
As we show in \Cref{sec:correctness}, this can be approximately 
achieved by assigning scores for each point $x$ based on 
$g_t(x):=\|\vec M_t(x - \mu_t)\|_2^2$, where $\vec M_t = \vec B_t^{\log d}$.
At a high level, this happens because the spectrum of $\vec M_t$ 
is distributed across along {\it all} large eigenvectors of $\vec B_t$.

\begin{algorithm}[h]  
    \caption{Robust Mean Estimation in polylog iterations}\label{alg:near-linear-only}
    \begin{algorithmic}[1]   
    \State \textbf{Input:} $S=\{x_i\}_{i \in [n]},\delta, \eps$
    \State Let $C_1 \geq 22$, $C$ be a sufficiently large constant, $C_2=100C$ and $C_3 = 0.1$.
      \State Let $R=\sqrt{(d/\eps)(1+\delta^2/\eps)}$.
      \State Let $P=(1-\eps)G+\eps B$ be the empirical distribution on the points from $S$. \footnotemark
      \State Let {$\outerl = C \log d \log \left(dR/\eps \right)$, $\innerl =  C \log((n+d) \outerl/\failp)$.}
      
      \State \label{line:naive_approx}
      Obtain a na\"ive estimate $\widehat{\mu}$ of $\mu$ with $\|\widehat{\mu} - \mu\|_2 \leq 4R$.%
      \State \label{line:roughfilter} Initialize $w_1(x) \gets \1\{\|x - \widehat{\mu}\|_2 \leq 5R\}$ for all $x \in S$. 
      \For{ $t \in [\outerl]$}
        \State Let $P_{t}$ be the distribution of $P$ weighted by $w_t$.
        \State Let $\mu_t, \vec{\Sigma}_t$ be the mean and covariance of $P_t$.
        \State Let $\vec B_t = (\E_{X\sim P} [w_t(X)])^2 \vec \Sigma_t - \left(1 -  C_1 \frac{ \delta^2}{\epsilon} \right)\vec I_d$ 
        \State Let $\vec M_t = \vec B_t^{\log d}$.  \Comment{{$\vec M_t$ does not need to be explicitly calculated.}}
        \State Let $\lambda_t =\|\vec B_t\|_2$ 
        \State Find \ $\widehat{\lambda}_t \in [0.8 \lambda_t,1.2\lambda_t]$ by power iteration. \Comment{{See \Cref{remark:multiplication_impl} for efficient implementation.}} \label{line:poweriter} %
        \If{$ \widehat{\lambda}_t > C_2 \delta^2/\epsilon$}  \label{line:lambda}
        \For {$j \in [\innerl]$}   \label{line:jl}   %
              \State $z_{t,j} \sim \cU (\{ \pm 1\}^d)$, 
              \State $v_{t,j} \gets {\vec{M}}_t z_{t,j}$. \Comment{{See \Cref{remark:multiplication_impl} for efficient implementation.}}\label{line:calc_v}
        \EndFor
        \State Denote by $\bU_t$ the matrix having the vectors $\frac{1}{\sqrt{\innerl}}v_{t,j}$ for  $j \in [\innerl]$ as rows.
        \State Let $\tilde{g}_t(x)=\|\vec U_t(x-\mu_{t})\|^2_2$ and  $\tilde{\tau}_{t}(x)=\tilde g_{t}(x) \1\{\tilde{g}_{t}(x) > C_3\|\vec U_t\|_F^2 \widehat{\lambda}_t/\eps\}$  \label{line:thresholding},  %
        \State $\ell_{\max} \gets ({dR}/{\eps})^{C \log d}$, $T \gets 0.01 \widehat{\lambda}_t \|\vec U_t\|_F^2$.  \label{line:set_ell_max}
        \State $w_{t+1} \gets  \mathrm{DownweightingFilter}(P,w_t,\tilde \tau_t, R,T, \ell_{\max})$ \Comment{\Cref{alg:reweighting}}
        \EndIf
        \EndFor
      \State \textbf{return} $\mu_t$.
      \label{line:filter_end}
    \end{algorithmic}  
  \end{algorithm} 
 \footnotetext{Without loss of generality, outliers are within $O(R)$ from $\mu$ in $\ell_2$-norm. 
 This is ensured in  \Cref{line:roughfilter}, which removes only $\eps$-fraction of inliers (\Cref{cl:radius}). } \label{footnote:1}

\begin{algorithm}[h]  
    \caption{Downweighting Filter} 
      \label{alg:reweighting}
    \begin{algorithmic}[1]
      \State \textbf{Input}: 
      $P,w,\tilde{\tau}, R, T,\ell_{\max}$
      \State $r \gets C  dR^{2+4\log d}$.  %
      \State Let $w_\ell(x)= w(x)(1 - \tilde{\tau}(x)/r)^\ell$.
      \State Find the smallest $\ell \in  \{1,\ldots, \ell_{\max} \}$ satisfying $\E_{X \sim P}\left[ w_\ell(X) \tilde{\tau}(X)   \right] \leq 2T$ 
       using binary search. \label{line:filter-cond}
        \State \textbf{return} $w_\ell$.
    \end{algorithmic}  
  \end{algorithm}

Even though assigning scores based on $\vec M_t$, i.e., $g_t(x)$, 
reduces the number of iterations, computing $g_t(x)$ for all $x \in S$ is slow.
We thus use a Johnson-Lindenstrauss (JL) sketch of $\vec M_t$, 
denoted by $\vec U_t$. We denote by $\tilde{g}_t(x):= \|\vec U_t (x - \mu_t)\|_2$ 
the resulting scores.
We claim that the set $\{\tilde{g}_t(x)\}_{x \in S}$ 
can be calculated in time $\tilde{O}(nd)$ 
such that for each $x \in S$, $\tilde{g}_t(x) \approx g_t(x)$.
First, we will show (see \Cref{cl:jl-cons}) that
$\vec U_t$ can be as small as $\innerl \times d$, 
where {$\innerl$ is some $\polylog(nd/(\eps\tau))$}, 
which follows from the classical JL lemma 
(stating that $n$ points can be linearly embedded into a $\log n$-dimensional space).
Also, each row of $\vec U_t$ can be computed 
by repeatedly multiplying a vector  $\log d$ times by $\vec B_t$ (\Cref{line:calc_v}). By the following remark, all rows of $\bU_t$ can be computed in time $\tilde{O}(nd)$ and thus, each iteration of \Cref{alg:near-linear-only} 
runs in near-linear time.

{
\begin{remark}(Efficient Implementation)\label{remark:multiplication_impl}
Note that for any $v \in \R^d$, the vector $\vec B_t v$ can be calculated in  $O(nd)$ time. This is because $\vec \Sigma_t v =  \sum_{x \in S} w_t(x) (v^T(x-\mu_t))(x-\mu_t)/(\sum_{x \in S} w_t(x))$ which means that the result  can be computed 
in $O(nd)$ time by calculating $\mu_t$ and $v^T(x-\mu_t)$ first. Regarding \Cref{line:poweriter}, an approximate large eigenvector can be computed via power iteration, 
i.e., starting from a random Gaussian vector and multiplying by $\bB_t$ 
iteratively  $\log d$ many times (see, for example, \cite{BlumEtAl2020}). As mentioned above, each of these multiplications can be done in $O(nd)$ time.
\end{remark}
}

For the proof of correctness, we require that the JL 
and spectral approximations used by the algorithm are sufficiently accurate. 
We prove that the following event occurs with high probability.

\begin{restatable}[Deterministic Conditions For \Cref{alg:near-linear-only}]{condition}{CondNearLinear} %
\label{cond:near-linear}
For all $t \in [\outerl]$, the following hold:
\begin{enumerate}
    \item \emph{Spectral norm of $\bB_t$}:  $\widehat{\lambda}_t \in  [0.8 \lambda_t,  1.2 \lambda_t]$. \label{as:lambda_near_linear}
    \item \emph{Frobenius norm}:  $\| \bU_t \|_F^2 \in \left[0.8 \| \bM_t\|_F^2, 1.2 \| \bM_t\|_F^2\right]$.\label{as:frobenious_near_linear}
    \item \emph{Scores}: For all $x \in S$, $\tilde{g}_t(x) \in [0.8 g_t(x),  1.2 g_t(x)]$.\label{as:scores_approx}
\end{enumerate}%
\end{restatable}

\subsection{Establishing the Deterministic Conditions} \label{sec:deterministic_cond}

In this section, we establish that \Cref{cond:near-linear} holds with high probability. Regarding {
\Cref{as:lambda_near_linear} of this condition}, 
an approximate large eigenvector can be computed via power iteration as described in \Cref{remark:multiplication_impl}.
This gives us an algorithm that runs in time {$O(nd\log d\log(\outerl/\failp))$} 
and satisfies \Cref{as:lambda_near_linear} with probability $1-\failp$.

We now move to the other two conditions. 
The claim is that instead of using the matrix $\bM_t$ to calculate the scores, 
it suffices to store and use only a small set of random  projections $\{\bM_t z_{t,j}\}_{j\in [\innerl]}$. 
This is exactly the Johnson-Lindenstrauss sketch that is computed 
in \Cref{line:jl} of \Cref{alg:near-linear-only}. 
Using \Cref{fact:jl}, we get the following guarantees (see \Cref{sec:omitted_JL} for the proof).

\begin{restatable}{lemma}{clJLCons} \label{cl:jl-cons} 
Fix a set of $n$ points  $x_1,\ldots, x_n \in \R^d$. 
For $t\in [\outerl]$, define ${g}_t(x):= \|{\vec M}_t (x-\mu_{t})\|_2^2$ and let 
$\tilde{g}_t(x),v_{t,j}$ as in \Cref{alg:near-linear-only}. 
If $C$ is a sufficiently large constant and $\innerl = C \log((n+d) \outerl/\failp)$, with probability at least $1-\failp$,  
{for every $t \in [\outerl]$ we have the following:}
\begin{enumerate}
        \item $0.8{g}_t(x_i) \leq \tilde{g_t}(x_i) \leq 1.2 {g}_t(x_i)$ for every $i \in [n]$,
        \item $0.8 \|\vec M_t\|_F^2 \leq \left(\frac{1}{\innerl}\sum_{j=1}^{\innerl}\|v_{t,j}\|_2^2\right) \leq 1.2 \|\vec M_t\|_F^2$.
    \end{enumerate}
\end{restatable}

This concludes the proof {that \Cref{cond:near-linear} is satisfied with high probability.}

\subsection{Downweighting Filter}
\label{sec:downweighting_filter}
We use the following re-weighting procedure also used in \cite{dong2019quantum}.
Recall  {that $P$ denotes} the empirical distribution on the samples, 
which we write as $P=(1-\eps)G + \eps B$, 
where $G$ and $B$ are the {contributions} from the good and bad samples respectively. 
Roughly speaking, our filter guarantees two things when going from the weights $w(x)$ to $w'(x)$:
  \begin{enumerate}
      \item \emph{The weight removed from the outliers is greater than the weight removed from the inliers}.
      \item $\E_{X \sim P}[w'(X)\tilde{\tau}(X)] \leq 2 \E_{X \sim G}[w(X)\tilde{\tau}(X)]$, 
      i.e., \emph{the weighted mean of scores after filtering over both inliers 
      and outliers is at most twice the weighted mean of scores of inliers \emph{before} filtering}.
  \end{enumerate}
Regarding the first guarantee, since the fraction of outliers is at most $\eps$, 
this ensures that the filtered distribution $P_t$ will never be more than $O(\eps)$-far 
in  total variation distance from the initial (corrupted) distribution, 
and thus the condition of the certificate lemma that $\dtv(P,P_t)\leq O(\eps)$ will always be satisfied. 
The second guarantee ensures that the filtering step reduces the average score significantly. 
We prove the following in \Cref{sec:omitted_filter}.

\begin{restatable}{lemma}{Lemfilterguarantee}  \label{lem:filterguarantee}
Let $P=(1-\eps)G + \eps B$ be the empirical distribution on $n$ samples, 
as in \Cref{alg:near-linear-only}. 
If $(1-\eps)\E_{X \sim G}[w(X)\tilde{\tau}(X)] \leq  T$, $\|\tilde{\tau}\|_\infty \leq r$, and $\ell_{\max}>r/T$, 
then \Cref{alg:reweighting} modifies the weight function $w$ to $w'$ such that 
\begin{enumerate}
    \item $ (1-\eps) \E_{X \sim G}[w(X) - w'(X)] < \eps \E_{X \sim B}[w(X) - w'(X)]$,
    \item $\E_{X \sim P}\left[ w'(X) \tilde{\tau}(X)   \right] \leq 2T$,
\end{enumerate}
and the algorithm {terminates after} $O(\log(\ell_{\max}))$ iterations, 
each of which takes $O(n)$ time. 
\end{restatable}
We note that the two conditions $\|\tilde{\tau}\|_\infty \leq r$, $\ell_{\max}>r/T$ 
of \Cref{lem:filterguarantee} hold by our choice of $\ell_{\max}$ 
and $r$ inside \Cref{alg:near-linear-only} and \Cref{alg:reweighting} as follows. 
For $\|\tilde{\tau}\|_\infty$, we have the following upper bound
\begin{align} \label{eq:bound_on_tau}
    \tilde{\tau}_t(x) &\leq \tilde{g}_t(x) \leq \| \bU_t(x-\mu_t)\|_2^2
    \lesssim R^2 \|\bU_t\|_2^2 \lesssim R^2 \|\bM_t\|_F^2
    \lesssim R^2 \| \vec \Sigma_t\|_2^{2\log d} = O(d R^{2+4 \log d} )\;,
\end{align}
where we used the guarantee of our JL sketch that 
$\|\bU_t\|_2^2 \leq 1.2 \|\bM_t\|_F^2$ (\Cref{cl:jl-cons}).
A crude upper bound on $r/T$ follows from the following inequalities:
\begin{align*}
      \frac{r}{T} \lesssim \frac{ d R^{2+4 \log d}}{\widehat{\lambda}_t \| \bM_t\|_F^2}
    \lesssim \frac{ d R^{2+4 \log d} }{{\lambda}_t \| \bM_t\|_F^2} 
    \lesssim \left( \frac{dR}{\delta^2/\eps}  \right)^{O(\log d)} \;,
\end{align*}
where the first inequality uses the values {of $r$ and $T$ as set in the algorithm}, 
the second inequality uses  \Cref{as:lambda_near_linear,as:frobenious_near_linear} 
of the deterministic conditions of \Cref{cond:near-linear}, 
and  the last inequality uses the fact that $\|\vec M_t \|_F^2 \geq \|\vec M_t \|_2^2$ and $\|\vec M_t \|_2^2$ 
cannot be smaller than $ ({C_2}\delta^2/\eps)^{O(\log d)}$ 
(otherwise \Cref{line:lambda1} terminates the algorithm).

We now {use the guarantees of} \Cref{alg:reweighting} as follows. 
We first show that the weighted mean of the inliers' scores is small.
\begin{restatable}{lemma}{lemGoodpartsmall} \label{lem:goodpartsmall}
Under the setting of \Cref{alg:near-linear-only} and the deterministic \Cref{cond:near-linear}, 
we have that $\E_{X \sim G}[w_{t}(X)\tau_t(X)]$ and $ \E_{X \sim G}[w_{t}(X)\tilde{\tau}_t(X)]$ 
are {bounded from above by} $c \lambda_t \|\vec M_t\|_F^2 $
for some constant $c$ of the form $c=C/C_2$, 
where $C_2$ is the constant used in \Cref{line:lambda} and $C$ is some absolute constant.
\end{restatable}
The proof is based on stability arguments from \Cref{sec:more_stabilityfacts} 
and can be found in \Cref{sec:omitted_proof_of_lemma}.

\begin{remark}\label{rem:value_of_c}
In our analysis, it will be important that the constant $c$ in  \Cref{lem:goodpartsmall} 
can be made sufficiently smaller than $1$, for example, $c < 0.01$.
 This can be achieved by choosing $C_2$ to be a large enough constant. 
\end{remark}

Using \Cref{alg:reweighting}, we get that the weighted sum of scores after filtering is also small.

\begin{restatable}{lemma}{AfterFiltering} \label{lem:AfterFiltering}
Under the setting of \Cref{alg:near-linear-only} and the deterministic \Cref{cond:near-linear}, 
we have that $\eps \E_{X \sim B}[w_{t+1}(X)\tilde{\tau}_t(X)] < c \lambda_t \|\vec M_t\|_F^2 $, 
with $c$ being of the form $c=C/C_2$, where $C_2$  is the constant used in  \Cref{line:lambda} 
and $C$ is some absolute constant.
Furthermore, $\eps \E_{X \sim B}[w_{t+1}(X)\tau_t(X)] < c \lambda_t \|\vec M_t\|_F^2$.
\end{restatable}
\begin{proof}
The first claim follows by the stopping condition of the algorithm, 
\Cref{lem:goodpartsmall}, and the fact that $w_{t+1}\leq w_t$.
We now prove the second conclusion {by relating $\tau$ to $\tilde{\tau}$:} 
Recall that we denote by $S$ the $\eps$-corrupted version of the original set of samples $S_0$.
Since $\tilde{g}_t(x)$ {is within a constant factor of} $g_t(x)$ (\Cref{cond:near-linear}) 
for all $x \in S$, the scores $\tilde{\tau}_t(x)$ and $\tau_t(x)$ 
are comparable (up to an additive term) as shown below.

\begin{restatable}{claim}{ClRelationTauFirst} \label{cl:relation_tau_1} 
In the setting of \Cref{alg:near-linear-only} and under the \Cref{cond:near-linear}, 
if $x \in S$, we have that $\tau_t(x) \leq  1.25 \tilde{\tau}_t(x) + 3 C_3 (\lambda_t/\eps) \tr(\bM_t^2)$, 
where $C_3$ is the constant used in \Cref{alg:near-linear-only}.
\end{restatable}

We prove \Cref{cl:relation_tau_1} in~\Cref{sec:omitted_correctness}. 
Using \Cref{cl:relation_tau_1}, we have the following  set of inequalities:
\begin{align*}
\eps \cdot \E_{X \sim B}[w_{t+1}(X)\tau_t(X)] 
&= \frac{1}{n}\sum_{x_i \in S\setminus S_0} w_{t+1}(x_i)\tau_t(x_i)   \\
&\leq 3 C_3\lambda_t\  \| \bM_t\|_F^2 +  \frac{1}{n} 1.25\sum_{i \in S\setminus S_0} w_{t+1}(x_i)\tilde{\tau}_t(x_i) \tag{using \Cref{cl:relation_tau_1}}\\
&= 3 C_3\lambda_t  \| \bM_t\|_F^2 + 1.25 \eps\E_{X \sim B}[w_{t+1}(X)\tilde{\tau}_t(X)]   \\
&\leq 3 C_3\lambda_t  \| \bM_t\|_F^2 + 1.25 c\lambda_t \| \bM_t\|_F^2 \tag{using the first part of \Cref{lem:AfterFiltering}}\\
&< 10 (C/C_2) \lambda_t \| \bM_t\|_F^2  \tag{using the value of $C_3$}\;.
\end{align*}

The last inequality above uses the fact that the constant $C_3$ 
is chosen to be $C_3=C/C_2$ in \Cref{alg:near-linear-only}, 
where $C$ is a sufficiently large constant.
\end{proof}

\subsection{Correctness of \Cref{alg:near-linear-only}: Proof of~\Cref{thm:near-linear-only}}  \label{sec:correctness}

The rest of this section is dedicated to proving  \Cref{thm:near-linear-only}. 
We first state the correctness of the na\"ive approximation step of  \Cref{line:naive_approx}, 
then record the invariants of the algorithm in \Cref{sec:invariants}, 
and finally show in \Cref{sec:mainproof} that it suffices for the number of iterations {$\outerl$ to be bounded by some  $\polylog(d,R,1/\eps)$}. %

The na\"ive approximation step of  \Cref{line:naive_approx} is based 
on the following folklore fact (see \Cref{sec:omitted_naiveprune} for more details).

\begin{restatable}{claim}{NAIVEPRUNE}\label{claim:naivepruning}
Let the fraction of outliers be $\eps<1/10$ and a parameter $0<\failp<1$. 
Let the distribution $P=(1-\eps)G+\eps B$. 
Let $R>0,\mu \in \R^d$ be such that $\pr_{X \sim G}[\|X-\mu \|_2>R]\leq \eps$. 
There is an estimator $\widehat{\mu}$ on $k=O(\log(1/\failp))$ samples from $P$ 
such that $\| \widehat{\mu} -\mu \|_2 \leq 4R$ with probability at least $1-\failp$. 
Furthermore, $\widehat{\mu}$ can be computed in time $O(k^2 d)$ and memory $O(kd)$.
\end{restatable}

\Cref{cl:radius} below gives a valid upper bound on $R$ using the $(\eps,\delta)$-stability of the good distribution.
\begin{restatable}{claim}{CLRAD}\label{cl:radius}
If $R= \sqrt{\frac{d}{\eps}(1+\delta^2/\eps)}$, then  $\pr_{X \sim G}[\|X-\mu \|_2 > R]\leq \eps$.
\end{restatable}
\begin{proof}
    By Markov's inequality, we have that
    \begin{align*}
        \pr_{X \sim G} \left[\|X - \mu\|^2_2 \geq \frac{d}{\eps}\left( 1+\delta^2/\eps \right)  \right] 
        \leq \eps \frac{\E_{X \sim G}[\|X - \mu\|^2_2   ]}{d( 1+\delta^2/\eps )} \leq \eps \;.
    \end{align*}

\end{proof}

\subsubsection{Invariants {of \Cref{alg:near-linear-only}}}
\label{sec:invariants}
Recall that the end goal is to obtain a filtered version, $P_t$, of $P$ 
that is not too far from $P$ in total variation distance {$\dtv(P,P_t) = O(\eps)$, 
and satisfies that $\| \bB_t \|_2 = O(\delta^2/\eps)$.} For the first condition to be satisfied, 
we ensure that the Downweighting filter removes more weight from $G$ than $B$ (\Cref{lem:filterguarantee}).  
Using this, we show that $\E_{X \sim G}[w_t(X)] \geq 1-O(\eps)$, 
which implies the bound on the total variation distance (\Cref{cl:total-var-small}). 
The proofs are deferred to \Cref{sec:omitted_correctness}.

\begin{restatable}{claim}{clinvariantweight}\label{cl:invariant-weight}
Under \Cref{cond:near-linear}, \Cref{alg:near-linear-only}  maintains the following invariant: 
    $\E_{X \sim G}[w_t(X)] \geq 1-3\eps$. In particular,     if $\eps \leq 1/8$, then $\dtv(P_t,P) \leq 9\eps$. 
\label{cl:total-var-small}
\end{restatable}

The following properties of $\bB_t$ as PSD operator will also be useful later on.

\begin{restatable}{claim}{ClPSD} \label{cl:PSD} 
{Under \Cref{cond:near-linear},} if $C_1 \geq 22$,
    $\vec B_t  \succeq (0.5C_1\delta^2/\eps) \vec{I}_d$ for every $t \in [\outerl]$.
\end{restatable}
The proof of \Cref{cl:PSD} is provided in \Cref{sec:omitted_correctness}.
Although just showing that $\bB_t \succeq 0$ would suffice for this section, 
the slightly stronger bound of the above claim will be useful in \Cref{sec:low-memory-main}. 
\Cref{cl:PSD} follows from \Cref{cl:invariant-weight} and the stability of $G$. 
In particular, the stability of $G$ implies that $\overline{\vec \Sigma}_G \succeq (1 -\delta^2/ \epsilon) \bI_d$.
We now prove the following claim, which is the reason for having 
the multiplicative factor of $\E_{X \sim P}[w_t(X)]^2$ in the definition of $\bB_t$.  
\begin{claim} \label{cl:decreasing}
We have that $\vec B_{t+1} \preceq \vec B_t$ for every $t \in [\outerl]$.
\end{claim}
\begin{proof}
    We use the alternative definition of the covariance matrix:
    Let $X , Y$ be i.i.d.\ from $P$, then
    \begin{align*}
        \vec{\Sigma}_t = \frac{1}{2 (\E_{X \sim P}[w_t(X)])^2} \E_{X,Y \sim P}[w_t(X)w_t(Y)(X-Y)(X-Y)^T]\;.
    \end{align*}
    Since $w_{t+1}(x) \leq w_t(x)$ for all $x$, this completes the proof.
\end{proof}

\subsubsection{Reducing the Potential Function} \label{sec:mainproof}

Recall that each iteration of \Cref{alg:near-linear-only} can be implemented in near-linear time. 
Thus, it remains to show that the choice {$\outerl=C\log d \log(dR/\eps)$ }
suffices to guarantee correctness of our algorithm.
We now sketch the proof using a potential function argument.
Let $\Lambda_t$ be the vector in $\R^d$ containing the eigenvalues of $\bB_t$. 
Recall that our goal is to show that $\|\vec B_t\|_{2} = \|\Lambda_t\|_\infty = O(\delta^2/\epsilon)$ 
in $\polylog$ many iterations. 
Let $p := 2 \log d$. Since $\|x\|_p = \Theta(\|x\|_\infty)$ for any $x \in \R^d$, 
we are motivated to use the  potential function $\phi_t := \|\Lambda_t\|_p^p$. 
We now focus on showing that $\phi_t$ decreases rapidly.
Observe that for any $i \in \Z_+$, $\trace(\bB_t^i) = \|\Lambda_t\|_i^i$. 
We start with the following inequalities (and explain them directly below):
\begin{align}
    \phi_{t+1} &= \|\Lambda_{t+1}\|_{p}^p \leq \left(d^{\frac{1}{p(p+1)}} \|\Lambda_{t+1}\|_{p+1}\right)^p  \notag \\
    &= d^{\frac{1}{p+1}}\left(\|\Lambda_{t+1}\|_{p+1}^{p+1}\right)^{\frac{p}{p+1}}  \notag \\
                &= d^{\frac{1}{p+1}} (\trace( \vec B_{t+1}^{p+1}))^{\frac{p}{p+1}} \notag \\
                &\leq d^{\frac{1}{p}} \left(\trace(\vec M_t \vec B_{t+1}\vec M_t)\right)^{\frac{p}{p+1}} \label{eq:potential_ineq_main} \;,
    \end{align}
where the first line uses \Cref{Fact:normReln}, 
the third one uses $\trace(\bB_{t+1}^i) = \|\Lambda_{t+1}\|_i^i$, 
and the last line uses \Cref{fact:trace_ineqq} along with the fact that $\bB_{t+1} \preceq \bB_t$, 
which holds because removing points can only make their covariance smaller; 
see \Cref{sec:correctness} for more details.

Then the goal becomes to bound from above the term $\trace(\vec M_t \vec B_{t+1}\vec M_t)$. 
The claim is that $\trace(\vec M_t \vec B_{t+1}\vec M_t)$ is related to  
$\E_{X \sim P}[w_{t+1}(X) \tau_t(X)]$, and thus can be bounded by $c \lambda_t \|\bM_t\|_F^2$.
Using the guarantees of the Downweighting filter (\Cref{lem:AfterFiltering}), 
we  prove the following result:
\begin{restatable}{lemma}{CLProgress} \label{cl:progress}  
Consider the setting of \Cref{alg:near-linear-only} and assume that \Cref{cond:near-linear} holds. 
Then $\trace(\vec M_t \vec B_{t+1}\vec M_t) \leq  c \lambda_t \|\vec M_t\|_F^2$ 
for some $c$ of the form $C/C_2$, where $C_2$ 
is the constant used in  \Cref{line:lambda} 
and $C$ is some absolute constant.
\end{restatable}

Before giving the details regarding \Cref{cl:progress}, 
we first show that it suffices to prove that the potential function 
decreases by a multiplicative factor. In the rest of the proof, 
we will assume that $c< 0.1$, which can be guaranteed by taking $C_2$ 
to be a sufficiently large constant (cf. \Cref{rem:value_of_c}). 
We continue with \Cref{eq:potential_ineq_main} as follows: 
\begin{align*}
    \phi_{t+1} &\leq d^{\frac{1}{p}} \left(\trace(\vec M_t \vec B_{t+1}\vec M_t)\right)^{\frac{p}{p+1}} \\
    &\leq d^{\frac{1}{p}} \left(c \|\Lambda_t\|_{\infty} \|\Lambda_t\|_p^p\right)^{\frac{p}{p+1}} \tag*{(using \Cref{cl:progress})}  \\
    &\leq d^{\frac{1}{p}} c^{\frac{p}{p+1}} \left(\|\Lambda_t\|_{p} \|\Lambda_t\|_p^p\right)^{\frac{p}{p+1}} \tag*{ (using $\|\Lambda_t\|_\infty \leq \|\Lambda_t\|_i$ for $i\geq 1$)}\\
    &= d^{\frac{1}{p}} c^{\frac{p}{p+1}} \|\Lambda_t\|_{p}^p\\
    &\leq 3 \sqrt{c} \|\Lambda_t\|_{p}^p \leq 0.9999 \phi_t \;,
\end{align*}
    where the last line uses that $d^{1/p} = \exp(\frac{\log d}{2 \log d}) \leq 3$, $p/(p+1) \geq 0.5$, and $c < 1$.
    We thus get the desired convergence.

The final step is to bound the number of iterations needed 
for \Cref{lem:certificate} to ensure that $\|\mu_t - \mu\|_2 = O(\delta)$. 
Concretely, due to our na\"ive pruning, at the beginning of the algorithm 
we have the upper bound $\phi_1 \leq d R^{O(\log d)}$. 
After $\outerl$ iterations, we have that $\phi_{\outerl} \leq 0.99^{\outerl}d R^{O(\log d)}$. 
Setting $\outerl $ as 
    \begin{align} \label{eq:choose\outerl}  %
        \outerl = C \log d \log \left(\frac{dR}{\delta^2/\eps} \right) 
    \end{align}
suffices to have that 
$\|\bB_{\outerl}\|_2  \leq (d O(\delta^2/\eps)^{2 \log d})^{\frac{1}{2 \log d}}  = O(\delta^2/\eps)$.
This implies that 
    \begin{align} \label{eq:bound_sigma} 
        \| \vec{\Sigma}_{\outerl} \|_2 \leq \frac{1}{\E_{X\sim P} [w_{\outerl}(X)]^2}\left( \|\bB_{\outerl}\|_2 + 1 \right) \leq 1 + \left( \frac{1}{(1-3\eps)^2}-1 \right) + O\left( \frac{\delta^2}{\eps} \right) \leq  1 + O\left( \frac{\delta^2}{\eps} \right) \;,
    \end{align}
    where we used that $\E_{X\sim P} [w_{\outerl}(X)])^2 \geq 1- 3\eps$, 
    $\delta \geq \eps$, and $\epsilon \leq \epsilon_0$. 
    An application of \Cref{lem:certificate} shows that the estimate 
    has error at most $\|\mu_t - \mu\|_2 = O(\delta)$. 
    This completes the proof of \Cref{thm:near-linear-only}. 
    The rest of the section focuses on proving \Cref{cl:progress}.

\paragraph{Proof Sketch of \Cref{cl:progress}} %
\label{par:details_regarding_cl:progress}

Before giving the full proof of \Cref{cl:progress}, we provide a brief proof sketch.
By the definition of $\bB_{t+1}$, we have the following (the full proof is deferred to the end of this section)
\begin{align*}
\trace(\vec M_t \vec B_{t+1}\vec M_t) 
\leq \E_{X \sim P}[w_{t+1}(X)]\E_{X \sim P}[w_{t+1}(X)g_t(X)] - \left(1-C_1 \frac{\delta^2}{\eps}\right)\|\vec M_t\|_F^2 \;.
\end{align*}
In order to bound from above $\E_{X \sim P}[w_{t+1}(X)g_t(X)]$, 
one can consider the contribution due to inliers (distribution $G$) 
and contribution due to outliers (distribution $B$). Using the stability of inliers and \Cref{cor:shift}, 
we have that $\E_{X \sim G}[w_{}(X)g_t(X)] \leq (1+c\lambda_t)\|\bM_t\|_F^2$, 
for any weight function $w$ satisfying the conditions of \Cref{cor:shift}. 
We know that $w_{t+1}$ satisfies them because of our invariant in \Cref{cl:invariant-weight}. 
Turning to the contribution of outliers, we want to bound 
$\eps \cdot \E_{X \sim B}[w_{t+1}(X)g_t(X)]$. 
By definition, we have that $g_t(x) \leq \tau_t(x) + C_3 \lambda_t \|\bM_t\|_F^2/\eps$, 
and thus we get that the desired expression is bounded from above 
by $\eps \E_{X \sim B}[w_{t+1}(X)\tau_t(X)] + C_3 \lambda_t \|\bM_t\|_F^2$. 
The first expression was bounded from above in \Cref{lem:AfterFiltering} 
by using the downweighting filter, and the second is small because of how $C_3$ is set in our algorithm.
     
This completes the proof sketch of \Cref{cl:progress}. We now provide the complete proof.

\begin{proof}[Proof of \Cref{cl:progress}]
We will bound the contribution of inliers and outliers to the quantity 
$\E_{X \sim P}[w_{t+1}(X)g_t(X)]$ from above. Recall from our notation 
that the decomposition into inliers and outliers is $P=(1-\eps)G+\eps B$.
    For the inliers, we use \Cref{cor:shift} with $\vec U = \bM_t$ and $b=\mu_t$ to obtain the following:
    \begin{align}
        \E_{X \sim G}[w_{t+1}(X)g_t(X)] &\leq \| \bM_t\|_F^2\left(1+\frac{\delta^2}{\eps} 
        + \|\mu_t-\mu\|_2^2 + 2 \delta \|\mu_t-\mu\|_2 \right)        
        \leq \| \bM_t\|_F^2\left(1+ c\lambda_t \right)\;, \label{eq:good_contr}
    \end{align}
    where the last inequality uses that, by the certificate lemma (\Cref{lem:certificate}), 
    every term except the first in the previous expression is less than a sufficiently small  fraction of $\lambda_t$.

Regarding the outliers, we decompose their contribution to $\E_{X \sim P}[w_{t+1}(X)g_t(X)]$ into two sets: 
(i) the set of points with  projection greater than the threshold $C_3 \| \bM_t\|_F^2{\lambda}_t/\eps$ 
used in  \Cref{line:thresholding} of the algorithm, and 
(ii) the set of points with smaller projection. 
Concretely, letting $L_t := \{ x \in \R^d \; : \; g_t(x) >C_3 \| \bM_t\|_F^2{\lambda}_t/\eps \}$, we have that
    \begin{align}
        \eps \E_{X \sim B}[w_{t+1}(X)g_t(X)] 
    &   =\eps \E_{X \sim B}[w_{t+1}(X) \tau_t(X) ] + \eps \E_{X \sim B}[w_{t+1}(X)g_t(X) \1\{x \not\in L_t \}] \notag \\
  &\leq c\lambda_t \| \bM_t\|_F^2 + \eps C_3 \| \bM_t\|_F^2{\lambda}_t/\eps \leq c'  \| \bM_t\|_F^2{\lambda}_t \;, \label{eq:bad_contr}
    \end{align}
    where the first inequality follows from \Cref{lem:AfterFiltering} 
    and the second inequality follows from the choice of $C_3$ in \Cref{alg:near-linear-only}. 

    We also use the following relation on $\vec \Sigma_{t+1}$:
    \begin{align}
        \vec \Sigma_{t+1} &= \E_{X \sim P_{t+1}}\left[(X-\mu_{t+1})(X-\mu_{t+1})^T\right] \nonumber\\
        \nonumber    &\preceq \E_{X \sim P_{t+1}}\left[(X-\mu_{t})(X-\mu_{t})^T\right] \\
        \label{eq:uppBdSigma}    &= \frac{1}{\E_{X \sim P}[w_{t+1}(X)]}\E_{X \sim P}[w_{t+1}(X)(X-\mu_{t})(X-\mu_{t})^T] \;.
    \end{align} 
    {Recalling the definition $\vec B_{t+1} = (\E_{X\sim P} [w_{t+1}(X)])^2 \vec \Sigma_{t+1} - \left(1 -  C_1 \frac{ \delta^2}{\epsilon} \right)\vec I_d$, 
    \Cref{eq:uppBdSigma} implies that $\vec B_{t+1} \preceq  \vec F_{t+1}$, where
    $\vec F_{t+1} := (\E_{X \sim P}[w_{t+1}(X)])\E_{X \sim P}[w_{t+1}(X)(X-\mu_{t})(X-\mu_{t})^T] - \left(1 -  C_1 \frac{ \delta^2}{\epsilon} \right)\vec I_d$.
    Using \Cref{fact:trace_PSD_ineq} along with the fact that $\vec{B}_{t+1} \succeq 0$ (\Cref{cl:PSD}), we get the following}:
    \begin{align*}
         \tr(\vec M_t \vec B_{t+1}\vec M_t) &= \tr(\vec M_t^2 \vec B_{t+1})  
         \leq  \tr(\vec M_t^2 \vec F_{t+1})  \tag*{(using $\tr(\vec{A}\vec{B}\vec{C}) = \tr(\vec{C}\vec{A}\vec{B})$)} \\
        &= \tr\left( \vec M_t \left(\E_{X \sim P}[w_{t+1}(X)] \E_{X \sim P}[w_{t+1}(X)(X-\mu_{t})(X-\mu_{t})^T] 
        - \left(1 -  C_1 \frac{ \delta^2}{\epsilon} \right)\vec I_d  \right) \vec M_t \right) \\
	&=\E_{X\sim P} [w_{t+1}(X)]  \E_{X \sim P}[w_{t+1}(X)\tr((X-\mu_{t})^T\vec M_t^2 (X-\mu_{t}))] 
	- \left(1 -  C_1 \frac{ \delta^2}{\epsilon} \right)\| \bM_t\|_F^2 \\
        &=   \E_{X\sim P} [w_{t+1}(X)]  \E_{X \sim P}[w_{t+1}(X)g_t(X)] - \left(1 -  C_1 \frac{ \delta^2}{\epsilon} \right)\| \bM_t\|_F^2\\
        &\leq \left( 1+c{\lambda}_t + c'{\lambda}_t  - (1 -  C_1 \delta^2/\eps ) \right)\| \bM_t\|_F^2 \tag*{(using \Cref{eq:good_contr,eq:bad_contr})}\\
        &\leq c'' \lambda_t \| \bM_t\|_F^2 \;.
    \end{align*}
    This concludes the proof.
\end{proof}

\section{Efficient Streaming Algorithm for Robust Mean Estimation} \label{sec:low-memory-main}

We now turn to the main focus of this paper and present a low-memory algorithm
for robust mean estimation.
Our algorithm works in two setups: 
(i) the single-pass streaming setting, where a set of i.i.d.\ samples 
from an $\epsilon$-corrupted distribution 
in total variation distance (\Cref{def:oblivious}) arrive \emph{one at a time} (\Cref{def:streaming}), 
and (ii) the strong-contamination model (\Cref{def:strongadv}), 
where the algorithm is allowed poly-logarithmically 
many passes over the input stream (defined below).  

\begin{definition}[Streaming Model in $k$ Passes]\label{def:streaming-multiple}
For a fixed set $S$, the elements of $S$ are revealed to the algorithm one at a time. 
This  process is repeated $k$ times. The sequence of elements in $S$ within each pass can be arbitrary.
\end{definition}

Our main result is the following theorem for the single-pass streaming model, 
which is a generalized version of \Cref{main-thm-intro}:
\begin{restatable}[Robust Mean Estimation in {Single-Pass Streaming Model}]{theorem}{MainThm}\label{th:main}
Let $d \in \Z_+$, $0<\failp<1$, $0<\eps<\eps_0$ for a sufficiently small constant $\eps_0$, and $\delta \geq \eps$.
Let $D$ be a distribution which is $(C\eps,\delta)$-stable with respect to the 
(unknown) vector $\mu \in \R^d$, for a sufficiently large constant $C>0$. 
Let $R$ be any radius such that $\pr_{X \sim D}[\|X -\mu\|_2 > R] \leq \eps$.
Let $P$ be a distribution with $\dtv(P,D)\leq \eps$. 
There exists an algorithm that given $\eps,\delta,\failp$ and  
\begin{align}\label{eq:sample_complex_in_theorem}
{n = O \left( R^2 \max\left(d,\frac{\eps}{\delta^2}, \frac{(1+\delta^2/\eps)d}{\delta^2 R^2}, \frac{\eps^2d}{\delta^4},\frac{R^2\eps^2}{\delta^2}, \frac{R^2\eps^4}{\delta^6} \right)  \polylog\left(d,\frac{1}{\eps},\frac{1}{\failp},R \right)  \right) }
\end{align} 
i.i.d.\ samples from $P$ in a stream according to the model of \Cref{def:streaming}, 
runs in time $n d   \,\polylog\left(d,1/\eps,1/\failp,R  \right)$, 
uses memory $d \, \polylog\left(d,1/\eps,1/\failp, R  \right)$, 
and returns a vector $\widehat{\mu}$ such that, with probability at least $1-\failp$, 
it holds that $\|\mu - \widehat{\mu}\|_2 = O(\delta)$.
\end{restatable}

Note that \Cref{main-thm-intro} in \Cref{sec:results} is a special case of \Cref{th:main} 
for the two important families of distributions: 
(i) subgaussian distributions with identity covariance, and 
(ii) distributions with bounded covariance. 
\begin{enumerate}
\item For subgaussian distributions with identity covariance, 
we have that $R=\Theta(\sqrt{d\log(1/\eps)})$, $\delta = O(\eps \sqrt{\log(1/\eps)}) $, 
and thus  ${n} = \tilde{O}\left( d^2/\eps^2  \right)$.
\item For distributions with covariance at most identity, we have that 
$R = \Theta(\sqrt{d/\eps})$, $\delta = O(\sqrt{\eps})$, and thus  
${n} = \tilde{O}\left( d^2/\epsilon \right)$.
\end{enumerate}

{In order to obtain a low-memory algorithm for the robust mean estimation problem, 
we start with an obstacle that one faces when trying to modify the existing \Cref{alg:near-linear-only} to that setting.}
The issue is that, since $n$ can be much larger than $d$, we cannot even store the weight function $w_t$.
Fortunately, this can be handled by freshly computing the scores $w_t(x)$ for any given $x$, 
whenever we need them. This  requires us to store only $\{ (\bU_{t}, \ell_t): t \in [\outerl]\}$, 
where $\bU_{t}$ is the Johnson-Lindenstrauss sketch at the iteration $t$, 
and $\ell_t$ is the corresponding count from the downweighting filter. 
This can be achieved with additional poly-logarithmic memory. 
Thus, \Cref{alg:near-linear-only} can  be readily extended to setting (ii), giving us  \Cref{thm:polylog_passes}.

\begin{corollary}[Robust Mean Estimation in Multiple Passes Streaming Model]\label{thm:polylog_passes}
Let $d \in \Z_+$, $0<\failp  <1$ and $0<\eps<\eps_0$ for a sufficiently small constant $\eps_0$, and $\delta \geq \eps$. 
Let $S$ be an $\eps$-corrupted version of a set that is $(C\eps,\delta)$-stable with respect 
to the (unknown) vector $\mu \in \R^d$, for a sufficiently large constant $C$. Denote by $n$ the cardinality of $S$.  
There exists an algorithm that {operates in the streaming model of \Cref{def:streaming-multiple} 
with $k=\polylog\left(d,1/\eps,1/\failp \right)$ and}, given $\eps,\delta,\failp$ and $T$, 
runs in time $n d   \,\polylog\left(d,1/\eps,1/\failp  \right)$, uses additional memory $d\, \polylog\left(d,1/\eps,1/\failp  \right)$, 
and finds a vector $\widehat{\mu}$ such that, with probability at least $1-\failp$, 
it holds $\|\mu - \widehat{\mu}\|_2 = O(\delta)$.
\end{corollary}

\noindent In the main body of this section, we prove Theorem~\ref{th:main}.

\subsection{Setup and Algorithm Description}\label{ssec:low-memory-alg}

Moving to the single-pass streaming model and \Cref{th:main} requires a change in perspective:
instead of having a corrupted dataset, we now have sample access 
to a distribution $P$ such that $\dtv(P,D) \leq \epsilon$, where $D$ is a stable distribution.
We will reweight this distribution using weights, $w_t(\cdot)$, 
that are now {\em functions on the whole $\R^d$} instead of a fixed dataset.
Thus, $P_t$ now denotes the reweighting of the (corrupted)  distribution $P$ with the weights $w_t$. 
Similarly $\mu_t, \vec \Sigma_t, \bB_t, \bM_t$ denote the quantities that pertain to the distribution $P_t$.
The goal of our algorithm remains essentially the same: 
obtain $P_t$ such that $\dtv(P_t,P) = O(\epsilon)$ and 
$\|\vec \Sigma_t\|_2 \leq 1 + O(\delta^2/\epsilon)$; 
\Cref{lem:certificate} would then imply that $\|\mu_t - \mu\|_2 = O(\delta)$. 
Before presenting the pseudocode of \Cref{alg:streaming}, 
we identify two problems that arise in generalizing 
our results from \Cref{sec:near_linear_main_body} 
and provide an overview of their solutions:

\paragraph{Calculating Scores} 
Recall that the only place where $\bM_t$ is used in \Cref{alg:near-linear-only} is 
\Cref{line:calc_v}, where $\bM_t$ is multiplied with the vectors $z_{t,j}$.
Let $z$ be an arbitrary vector.
Since $\bM_t = \bB^{\log d}$, in the previous section we were able to compute $\bM_t z$ 
by iteratively multiplying $z$ by $\bB_t$.
Since we now do not have access to $\vec B_t$, but only sample access to $P_t$,
we need a sufficiently fine approximation $\widehat{\bB}_t$ of $\bB_t$ 
(obtained using i.i.d.\ samples). 
The natural approach would then be to multiply $\widehat{\bB}_t$ with $z$ iteratively $\log d$ many times. 
Even though $\widehat{\bB}_t z$ can be computed in a streaming fashion (as outlined in the previous section), 
it is not possible to compute $(\widehat{\bB}_t)^{\log d} z$ without accessing the data $\log d$ times.
To circumvent this issue, we use a fresh sample approximation of $\vec B_t$ in every multiplication step. 
That is, we approximate $\bM_tz$ by $\widehat{\vec M}_t z$, where 
$\widehat{\vec M}_t := \prod_{j=1}^p \widehat{\vec B}_{t,j}$ and each $\widehat{\vec B}_{t,j}$
is computed on a different set of samples.
This approach crucially leverages the fact that 
in the contamination model of \Cref{def:oblivious},
outliers are added in a way that is oblivious to the inliers,
and therefore these datasets are statistically identical
and independent of each other.
We show in \Cref{sec:concentration} that the resulting $\widehat{\vec M}_t$
is a sufficiently accurate approximation of $\bM_t$.
Similarly, we need to modify the Downweighting filter, 
since its implementation using binary search requires performing checks 
of the form $\E_{X \sim P}[w(X)\tilde{\tau}(x)] > 2T$ 
and calculating the weighted mean exactly is no longer possible.
We propose a sample-efficient estimator to approximate 
that expectation (see \Cref{cl:evaluate_cond} in \Cref{sec:remainingparts}) 
and run an ``approximate'' variant of binary search (see \Cref{sec:correctness_streaming}).

\paragraph{Cover Argument} 
We now turn to the more technical issue 
of controlling the size of the JL-sketch, i.e., the number of rows, $\innerl$, 
of the matrix $\vec U_t \in \R^{\innerl \times d}$. 
For simplicity, assume $\widehat{\vec M}_t = \vec M_t$, 
and recall that $\tilde{\tau}(x)$ is the thresholded version of $\|\bU_t(x-\mu_t)\|_2^2$, 
as defined in  \Cref{line:thresholding} and $\tau(x)$ is the same score but using $\bM_t$.
The potential-based analysis in \Cref{sec:near_linear_main_body} requires 
that $\E_{X \sim P}[w_{t+1}(X){\tau}_t(X)]$ is small. 
However, the stopping condition of the Downweighting filter 
 implies only that $\E_{X \sim P}[w_{t+1}(X)\tilde{\tau}_t(X)]$ is small.
In \Cref{sec:near_linear_main_body}, the bound on the former 
was obtained from the bound on the latter 
by using that $\|\vec U_t  (x - \mu_t)\|_2 \approx \|\vec M_t  (x - \mu_t)\|_2$ 
pointwise in the support of $P$ (\Cref{cl:relation_tau_1}).

By the classical JL lemma, the  size of the JL sketch, $\innerl$, 
needs to be at most logarithmic in the size of the set $S$ 
where we require the pointwise approximation to hold. 
Thus, in the previous section, $\innerl$ scaled as $\log|S|=\log n$.
However, in the streaming model
where there is no such dataset, 
it is far from obvious how the analysis should proceed.
A n{\"a}ive approach would be to require the approximation 
to hold on a cover $\tilde{S}$ of the support of $P_t$. 
Since $|\tilde{S}|$ scales exponentially with $d$,  
the required bound on $\innerl$ would be $\log |\tilde{S}| = \Omega(d)$, 
which is too large for our purposes.
Luckily, we can still find a fixed set $S_{\text{cover}}$ such that the following holds: 
(i) $\log |S_{\text{cover}}| = \polylog(d/\epsilon)$, and 
(ii) the expectation of scores over $\cU(S_{\text{cover}})$ 
approximates the expectation of scores over $P$.
That is, as far as the expectations of the scores are concerned, 
$P$ can be approximated by the uniform distribution over $S_{\text{cover}}$.
Arguing as before,  if $\|\vec U_t (x-\mu_t)\| \approx \|\vec M_t  (x-\mu_t)\|$ for each $x \in S_{\text{cover}}$, 
then the downweighting filter also ensures that $\E_{X \sim P}[w_{t+1}(X)\tilde{\tau}_t(X)]$ is small.
Thus, $S_{\text{cover}}$ can serve as a proxy dataset (used only in the analysis) 
to ensure that the size of the JL sketch is sufficiently bounded, 
i.e., that $\log |S_{\text{cover}}| \leq C \polylog(d/\epsilon)$.	

Establishing the desired upper bound on the cardinality of $S_{\text{cover}}$ 
requires a somewhat more sophisticated argument that relies on the VC-dimension 
of a family of functions corresponding to the weight update rule. 
This result is stated in \Cref{sec:cover}. 

We now present the algorithm more formally. We start by clarifying the notation used.

\paragraph{Notation regarding \Cref{alg:streaming}:} 
The quantities involved in the algorithm and its analysis now 
are based on the underlying data distribution $P$ as well as its approximations. 
We note that $P_t, \mu_t, \vec \Sigma_t, \vec B_t, \vec M_t, \lambda_t$ 
are functionals of the distribution $P$ and are primarily used in the analysis. 
The parameters $\widehat{\lambda}_t, \widehat{\vec M_t}$ are approximations 
for $\|\vec B_t\|_2$ and $\vec M_t$ respectively that the algorithm forms 
using samples from $P_t$. Regarding score functions, $g_t(x)=\|\vec M_t(x- \mu_t)\|_2^2$ is as before.
 The computations however  use only the Johnson-Lindenstrauss versions 
 $\tilde{g}_t(x) := \frac{1}{\innerl} \sum_{i =1}^{\innerl} (v_{t,i}^T(x-\widehat{\mu}_t))^2$, 
 which can also be written as $\|\vec U_t(x-\widehat{\mu}_t)\|_2^2$ in matrix form, 
 by defining $\bU_t$ to have the vectors $\frac{1}{\sqrt{\innerl}}v_{t,i}$ as its rows.
Note that $\tilde{g}_t(x) $ is defined using $\widehat{\mu}_t$ instead of $\mu_t$.
Finally, we denote by $\tau_t(x) = g_t(x)\1\{g_t(x) > C_3 \|{\vec M}_t\|_F^2\lambda_t/\eps\}$ 
and $\tilde{\tau}_t(x) = \tilde{g}_t(x)\1\{\tilde{g}_t(x) > C_3 \|\bU_t\|_F^2\widehat{\lambda_t}/\eps\}$.
  
\begin{remark}
Recalling \Cref{lem:dtv-stab}, we may again treat the input distribution as a mixture 
$P=(1-\eps)G+\eps B$, where $G$ is a distribution that is $(C'\eps,\delta)$-stable 
with respect to $\mu$.
\end{remark}

\begin{algorithm}[h!]  
    \caption{Robust Mean Estimation In Single-Pass Streaming Model} 
    \label{alg:streaming}
    \begin{algorithmic}[1] 
      \Function{RobustMeanStreaming}{{$\delta, \eps,\failp$ and sample access to $P$}}
      \State  Let $R$ be  such that $\pr_{X \sim G}[\|X -\mu\|_2 > R] \leq \eps$.
      \State Let $P=(1-\eps)G+\eps B$. Without loss of generality, we assume that the points added by the adversary are within $O(R)$ from $\mu$ in Euclidean norm (see \Cref{footnote:1}).
      \State Let $C$ be a sufficiently large constant.
      \State {Let $\outerl = C \log d \log (d R/\eps)$. }
      \State {Let $\innerl = C \log^3 (d R/\eps) \log^2(1/(\tau\eps))$. }
      \State Let $r = C  dR^{2+4\log d}$.    
      \State Obtain a na\"ive estimation $\widehat{\mu}$ of $\mu$ such that $\|\widehat{\mu} - \mu\|_2 \leq 4R$. 
    \label{line:naive_est_2}
      \State {Let $w: \R^d \to [0,1]$ be the weight function. 
      \State Initialize $w_0(x) \gets \1\{\|x - \widehat{\mu}\|_2 \leq {5R}\}$ \new{for all $x \in \R^d$} and $\ell_1 \gets 0$.}
      \For{ $t \in [\outerl]$}  
      \State Define $w_t(x) = w_{t-1}(x)(1-\tilde{\tau}_t(x)/r)^{\ell_t}$.
        \State Let $P_{t}$ be the distribution of $P$ weighted by $w_t$, i.e., $P_{t}(x) = P(x)w_t(x)/ \E_{X \sim P}[w_t(X)]$. 
        \State Let $\mu_t$ be the mean of $P_t$.
        \State Let $\vec{\Sigma}_t$ be the covariance matrix of $P_t$.
        \State Let $\vec B_t = (\E_{X\sim P} [w_t(X)])^2 \vec \Sigma_t - \left(1 -  C_1 \frac{ \delta^2}{\epsilon} \right)\vec I_d$ and $\vec M_t = \vec B_t^{\log d}$.
        \State Compute an $O(\delta)$-accurate estimator $\widehat{\mu}_t$ of $\mu_t$ (see \Cref{lem:final_step}).
        \State Let $\tilde{n} = C'' {R^2 (\log d)^2  } \max\left(d, \frac{\eps^2d}{\delta^4},\frac{R^2\eps^2}{\delta^2}, \frac{R^2\eps^4}{\delta^6} \right)    \log\left(  \frac{d \outerl \log d}{\failp}\right)$. 
        \State For $k \in [\log d]$, denote by $\widehat{\bB}_{t,k}$  the empirical version of $\bB_t$ over $\tilde{n}$ fresh i.i.d.\ samples (see \Cref{sec:concentration} for more details). \label{line:low-memB1}
        \State Define $\widehat{\bM}_t := \prod_{k=1}^{\log d} \widehat{\bB}_{t,k}$ 
         \Comment{{$\widehat{\bM}_t$ is not stored in memory.}}
                \label{line:low-memB2}
        \State Let $\lambda_t =\|\vec B_t\|_2$ \new{and a constant-factor approximation $\widehat{\lambda}_t$ of it}.  
        \If{$ \widehat{\lambda}_t > C_2 \delta^2/\epsilon$}   \label{line:lambda1}
        \For{$j \in [\innerl]$}   \label{line:jl1}
              \State $z_{t,j} \sim \cU (\{ \pm 1\}^d)$. \label{line:choose_z} 
              \State $v_{t,j} \gets \widehat{\vec{M}}_t z_{t,j}$. \Comment{{See \Cref{remark:implementation} for efficient implementation.}} \label{line:multiply}
              \State {Store $v_{t,j}$ in memory.   }
              \label{line:multiplication}
        \EndFor
        \State Denote by $\bU_t$ the matrix with rows $\frac{1}{\sqrt{\innerl}}v_{t,j}$ for $j \in [\innerl]$.  
        \State \label{line:scoresdef1}Let $\tilde{g}_t(x)=\|\vec U_t(x-\widehat{\mu}_{t})\|^2_2$  and $\tilde{\tau}_{t}(x) = \tilde g_{t}(x) \1\{\tilde{g}_{t}(x) > C_3 \|\vec U_t\|_F^2 \widehat{\lambda}_t/\eps\}$.  \label{line:thresholding1}
        \State $\ell_{\max} \gets \left(\frac{dR}{\delta^2/\eps}\right)^{{C \log d}}$. \label{line:set_ell_max2}
        \State $\ell_t \gets \mathrm{DownweightingFilter}(P,w_t,\tilde \tau_t, R, c \widehat{\lambda}_t \|\vec U_t\|_F^2,\ell_{\max})$. \Comment{\Cref{alg:reweighting2}}
        \State{{Store $\ell_t$ in memory.}}
        \EndIf
        \EndFor
      \State \textbf{return} an $O(\delta)$ approximation $\widehat{\mu}_t$ of the mean $\mu_t$ of the distribution $P_t$ (see \Cref{lem:final_step}).   
      \label{final:estimate} \label{line:return}
    \EndFunction  
    \end{algorithmic}  
  \end{algorithm}

As already mentioned, \Cref{alg:streaming} uses two levels of approximation: 
the first level is approximating the true distributional quantities 
by taking samples, and the second is preserving the latter quantities using the JL sketch.
If both of these approximations are sufficiently accurate, 
the correctness of \Cref{alg:streaming} would follow similarly to \Cref{alg:near-linear-only}.  
Of course, the challenge is to ensure that these approximations hold over the entire distribution, 
while controlling the sample and memory complexity of the algorithm.
As we show in \Cref{sec:cover}, this can be achieved by restricting our attention 
to a finite set (cover) of sufficiently large cardinality.
Thus, the deterministic conditions that we require now also involve the cover set, 
which we denote by $S_{\text{cover}}$. The reader may think of $S_{\text{cover}}$ as a fixed set, 
which will be specified later on (\Cref{cor:ran_cov_specific}).

\begin{condition}[Deterministic Conditions for \Cref{alg:streaming}] %
\label{cond:streaming}
{Let $S_{\text{cover}}$ denote the cover of \Cref{cor:ran_cov_specific} for $\eps' = \poly(d,R,1/\eps)^{\log d}$.} 
Our condition consists of the following event:
\begin{enumerate}
\item \emph{Estimator $\widehat{\mu}_t$}: For all $t \in [\outerl]$, we have that 
$\snorm{2}{\widehat{\mu}_t - \mu_t} \leq \delta/100$. \label{as:mean_est}

\item For every $t \in [\outerl]$, if $\|\bB_t\|_2 \geq (C_1/2)\delta^2/\eps$ and 
$\E_{X \sim P}[w_t(X)]\geq 1-O(\eps)$, we have that:
\begin{enumerate}

   \item \emph{Spectral norm of $\bB_t$}: $\widehat{\lambda}_t \in  [\new{0.1} \lambda_t,  20 \lambda_t]$. \label{as:lambda_samples} 
   
    \item \emph{Frobenius norm}:  $\| \bU_t \|_F^2 \in \left[0.8 \| \bM_t\|_F^2, 1.2 \| \bM_t\|_F^2\right]$. \label{as:frob}

    \item \emph{Scores}: 
             $\tilde{g}_t(x) \geq 0.2 g_t(x) - 0.8(\delta^2/\eps^2)\|\bM_t\|_F^2$, for all $x \in S_{\text{cover}}$. \label{as:scores}
\end{enumerate}
    
 \item \emph{Stopping condition}: Let $T_t:=c\widehat{\lambda}_t\|\bU_t \|_F^2$. 
 {For every $w:\R^d \to [0,1]$, the algorithm has access to an estimator $f(w)$}
 for the quantity $\E_{X \sim P}[w(X)\tilde{\tau}_t(X)]$, such that $\widehat{F}(P)>T_t/2$ 
 whenever $\E_{X \sim P}[w(X)\tilde{\tau}_t(X)]>T_t$. This estimator is accurate 
 when called $O(\log(d)\log(dR/\eps))$ times in every iteration $t \in [\outerl]$.  \label{stopping_cond_as}
 \end{enumerate}%
\end{condition}

We note that the \Cref{stopping_cond_as} above is needed to evaluate 
the stopping condition in the downweighting filter. For every $t \in [\outerl]$, 
the stopping condition is evaluated at most $O(\log(\ell_{\max}))$ times, 
with $\ell_{\max}=O(dR^{2+\log d}/(\widehat{\lambda}_t \| \bU_t \|^2)$  
(using \Cref{lem:filterguarantee} with $r=C_4dR^{2+4 \log d} $ 
and $T:= O(\widehat{\lambda}_t \|\vec U_t\|_F^2)$). 
This means that we require the estimator in \Cref{stopping_cond_as} 
to be accurate on $O(\outerl \log(d)\log(dR/\eps))$ calls.

\subsection{Correctness of \Cref{alg:streaming}} \label{sec:correctness_streaming}

The analysis in this section is along the same lines as that of \Cref{sec:correctness}. 
The na\"ive estimation step of \Cref{line:naive_est_2} is the same 
as that used in \Cref{alg:near-linear-only} (see \Cref{sec:omitted_naiveprune}).

Given \Cref{cond:streaming}, we first show the correctness of \Cref{alg:streaming} 
and leave the task of establishing \Cref{cond:streaming} for \Cref{sec:concentration}. 
The proof of correctness would largely follow by our work done in \Cref{sec:correctness}. 
There are  two adjustments needed in these arguments. 

The first concerns \Cref{lem:filterguarantee}, since the algorithm cannot perform exact binary search. 
Instead, it can use the approximate oracle of \Cref{stopping_cond_as} of \Cref{cond:streaming}, 
resulting in a multiplicative constant in the final guarantee. 
For completeness, we prove correctness for this case in \Cref{sec:appendixnewDownweighting}.

\begin{restatable}{lemma}{Lemfilterguaranteenew}  \label{lem:filterguarantee2}
In the context of \Cref{alg:streaming}, if $(1-\eps)\E_{X \sim G}[w(X)\tilde{\tau}(X)] \leq  T$, 
$\|\tilde{\tau}\|_\infty \leq r$, and $\ell_{\max}>r/T$, then \Cref{alg:reweighting2} modifies 
the weight function $w$ to $w'$ such that 
$(i)$ $ (1-\eps) \E_{X \sim G}[w(X) - w'(X)] < \eps \E_{X \sim B}[w(X) - w'(X)]$, 
and $(ii)$ upon termination we have $\E_{X \sim P}\left[ w'(X) \tilde{\tau}(X)   \right] \leq 54T$.
Furthermore, if the estimator of  \Cref{line:estimator_black_box} is set to be that of \Cref{cl:evaluate_cond},  
the algorithm terminates after $O(\log(\ell_{\max}))$ iterations, 
each of which uses $O((R^2\eps/\delta^2)\log(1/\failp))$ samples, 
takes $O(n d)$ time and memory $O(\log(1/\failp))$. 
\end{restatable}

\begin{algorithm}  
    
    \caption{Downweighting Filter using Approximate Oracle} 
      \label{alg:reweighting2}
    \begin{algorithmic}[1]  
      \Function{DownweightingFilter}{$P,w,\tilde{\tau}, R, T,\ell_{\max}$}
      \State $r \gets C  dR^{2+4\log d}$.  %
      \State Denote by $f(\ell)$ an estimator close to $\E_{X \sim P}[w(X)(1-\tilde{\tau}(X)/r)^{\ell} \tilde{\tau}(X)]$ (see \Cref{cl:evaluate_cond} for details). \label{line:estimator_black_box}
      \State $L \gets \{1,2\ldots,  \ell_{\max}\}$
      \While{$|L|>2$}   %
      \State Let $\ell$ be the element in the middle of $L$. \label{line:powerell}
           \If{$f(\ell)>9T$}
           \State Discard all elements smaller than $\ell$ from $L$.
           \Else
           \State Discard all elements greater than $\ell$ from $L$.
           \EndIf
        \EndWhile
        \State \textbf{return} any $\ell$ of $L$ satisfying $4T \leq f(\ell) \leq 36T$.
    \EndFunction  
    \end{algorithmic}  
  \end{algorithm}

The second adjustment is regarding the analog of \Cref{lem:AfterFiltering},  
i.e., $\eps \E_{X \sim B}[w_{t+1} \tau_t(X)]$ is small 
(the bound on $\eps \E_{X \sim B}[w_{t+1} \tilde{\tau}_t(X)]$ 
follows from the stopping condition as before).
Since the support is unbounded, we use an argument based on a fixed cover 
to show that the downweighting filter succeeds 
with the JL-sketch of size $\innerl$. The statement is given below.

\begin{lemma} \label{lem:remaining_thing} 
Under the deterministic \Cref{cond:streaming} and the context of \Cref{alg:streaming}, 
we have that  $\E_{X \sim B}[w_{t+1}(X) \tau_t(X)] \leq 5\E_{X \sim B}[w_{t+1}(X)\tilde{\tau}_t(X)] 
+ c (\lambda_t/\eps) \|\bM_t\|^2_F$, where $c$ is of the form $C/C_2$ 
with $C$ being a sufficiently large constant and $C_2$ being the constant used in \Cref{alg:streaming}.
\end{lemma}

The next section is dedicated to proving \Cref{lem:remaining_thing}. 
Here we just show that \Cref{lem:remaining_thing} 
suffices to prove the analog of \Cref{lem:AfterFiltering} below.

\begin{restatable}{lemma}{LemAfterFilt} \label{lem:after_filtering2}
Consider the context of \Cref{alg:streaming} and assume that \Cref{cond:streaming} holds. 
Then $\eps \cdot \E_{X \sim B}[w_{t+1}(X)\tau_t(X)] \leq c'\lambda_t \| \bM_t\|_F^2$, 
for some constant $c''$ of the form $C''/C_2$, where $C_2$ is the constant 
used in  \Cref{line:lambda1} and $C'$ is some absolute constant. 
\end{restatable}

\begin{proof}
Denoting by $c,c',c''>0$ constants that are all multiples of  $1/C_2$, we have the following:
\begin{align*}
\eps \E_{X \sim B}[w_{t+1}(X){\tau}_t(x)] &\leq 
5\eps \E_{X \sim B}[w_{t+1}(X)\tilde{\tau}_t(x)] + c\lambda_t\|\bM_t\|^2_F
\leq 5c'\lambda_t\|\bM_t\|^2_F + c\lambda_t\|\bM_t\|^2_F \leq c'' \lambda_t\|\bM_t\|^2_F\;,
\end{align*}
where the first inequality uses \Cref{lem:remaining_thing} 
and the  second inequality uses \Cref{lem:filterguarantee2}.
\end{proof}

Letting $\phi_t:=\tr(\bM_t^2)$ denote the potential function, 
the above result allows us to follow the same steps as in \Cref{sec:mainproof} 
to prove that  $\phi_{t+1} \leq 0.9999\phi_t$ exactly as in \Cref{sec:mainproof}. 
Thus, we get that after $\outerl$ iterations, we have that 
$\|\vec \Sigma_t\|_2 \lesssim \delta^2/\eps$. 
Under \Cref{as:mean_est} of \Cref{cond:streaming}, 
we have that the final estimate $\widehat{\mu}_t$ satisfies that $\|\widehat{\mu}_t - \mu\|_2 = O(\delta)$.
This completes the proof of correctness of \Cref{alg:streaming}.

\subsubsection{Proof of \Cref{lem:remaining_thing} via a Cover Argument} \label{sec:cover}

To outline the idea of proving \Cref{lem:remaining_thing}, 
recall the proof in the setting of \Cref{sec:near_linear_main_body}. 
There, we just required that $\tilde{g}_t(x)/g_t(x) \in [0.8,1.2]$ 
for all samples $x$ in our dataset, which can be translated to some 
relation between $\tilde{\tau}(x)$ and $\tau(x)$. Then, 
since $B$ was the empirical distribution on $\eps n$ of these points, 
the desired condition followed. However, in our case we cannot use 
pointwise relationships, since the distribution $B$ may be continuous 
and the Johnson-Lindenstrauss argument might not work for the entire $\R^d$ with $\polylog(d)$ vectors. 
The idea is first to relate $\E_{X \sim B}[w_{t+1}(X){\tau}_t(X)]$ to a discrete expectation over $N$ 
(not too many) points from a fixed set, then use the relationship between $\tilde{\tau}$ and $\tau$ 
for these points, and finally relate that discrete expectation back to 
$\E_{X \sim B}[w_{t+1}(X)\tilde{\tau}_t(X)]$. The existence of a cover of a small size is stated in the following.

\begin{restatable}{lemma}{CorRanCovSpecific} \label{cor:ran_cov_specific}
Consider the setting of \Cref{alg:streaming}, where $B$ is the distribution 
of outliers supported in a ball of radius $R$ around $\mu$. 
Let $r':= \left( C dR^2 + 1 + C_1\delta^2/\eps  \right)^{{C \log d}}$ for sufficiently large constant $C$. 
{Denote by $\eps$ the contamination rate and let an arbitrary $\eps' \in (0,1)$.}
There exists a set $S_{\text{cover}}$ of $N= \frac{1}{\eps'^3} d^4 \outerl^2 \innerl^2 (dR\eps/\delta^2)^{O(\log d)}$
points $x_1,\ldots, x_N$ lying in the ball of radius $R$ around $\mu$, such that for all $t \in [\outerl]$, 
{for all choices of the vectors $z_{t,j}$ of  \Cref{line:choose_z} of \Cref{alg:streaming} it holds}
\begin{align*}
\abs[\Big]{\E_{X \sim B}\left[\frac{1}{r'}w_{t+1}(X)\tilde{\tau}_t(X)\right] 
- \frac{1}{N}\sum_{i=1}^N \frac{1}{r'}w_{t+1}(x_i)\tilde{\tau}_t(x_i)} &\leq \eps'  \\
\text{and} \quad 
\abs[\Big]{\E_{X \sim B}\left[\frac{1}{r'}w_{t+1}(X)\tau_t(X)\right] 
- \frac{1}{N}\sum_{i=1}^N \frac{1}{r'}w_{t+1}(x_i)\tau_t(x_i)} &\leq \eps'  \;.
\end{align*}
\end{restatable}

We prove this result in \Cref{sec:omitted_cover}. 
Here we show how it implies the desired condition, following the  proof sketch from the start of this section.

\begin{proof}[Proof of \Cref{lem:remaining_thing}]
    Let $r':= \left( C dR^2 + 1 + C_1\delta^2/\eps  \right)^{{C \log d}}$ and $\eps' \in (0,1)$. Applying 
 \Cref{cor:ran_cov_specific}, let $S_{\text{cover}}$ be the corresponding cover of cardinality $N$.
From the guarantee of approximation of $\tilde{g}_t$ for every $x \in S_{\text{cover}}$, 
we get the following approximation for $\tilde{\tau}_t(x)$ for $x \in S_{\text{cover}}$ (proved in \Cref{sec:omitted_cover}).

\begin{restatable}{claim}{RelationClaim} \label{cl:RelationClaim}
Let $S$ be the cover of \Cref{cor:ran_cov_specific} with $r'$ and $\eps'$ as defined above.
Suppose that the deterministic condition \Cref{cond:streaming} holds.
If $x \in S_{\text{cover}}$, then $\tau_t(x) \leq 5\tilde{\tau}_t(x) + (18C_3 + 12/C_2)(\lambda_t/\eps)  \|\bM_t\|_F^2 $,
where $C_3$ and $C_2$ are the constants used in \Cref{alg:streaming}. 
\end{restatable}

Using \Cref{cl:RelationClaim} and  \Cref{cor:ran_cov_specific}, 
we obtain the following series of inequalities: 
\begin{align*}
&\E_{X \sim B}[w_{t+1}(X) \tau_t(X)] =  r' \E_{X \sim B}\left[\frac{1}{r'}w_{t+1}(X)\tau_t(X)\right] \\
&\leq \eps' r' +  r' \frac{1}{N} \sum_{i = 1}^N \frac{1}{r'}w_{t+1}(x_i)\tau_t(x_i)  \tag{using \Cref{cor:ran_cov_specific} for $\tau_t$} \\
&= \eps' r' +   \frac{1}{N} \sum_{i = 1}^N w_{t+1}(x_i)\tau_t(x_i)   \\
&\leq \eps'r' + (18 C_3+12/C_2)(\lambda_t/\eps)  \|\bM_t\|_F^2 
+ 5 \frac{1}{N} \sum_{i = 1}^N w_{t+1}(x_i)\tilde{\tau}_t(x_i) \tag{using \Cref{cl:RelationClaim} and $w_t \leq 1$}\\
&\leq 6\eps'r' + (18 C_3+12/C_2)(\lambda_t/\eps) \|\bM_t\|_F^2  
+ 5  \E_{X \sim B}[w_{t+1}(X)\tilde{\tau}_t(X)]   \tag{using \Cref{cor:ran_cov_specific}  for $\tilde{\tau}_t$ } \\
&= 5\E_{X \sim B}[w_{t+1}(X)\tilde{\tau}_t(X)] + (19C_3 + 12/C_2) (\lambda_t/\eps) \|\bM_t\|_F^2. \tag{using the definition of $\epsilon'$}
\end{align*}
{For the last line above, we want to choose $\eps'$ such that 
$\eps' \leq \frac{C_3\lambda_t}{\eps r'} \|\bM_t\|_F^2$. 
Since $\|\bM_t\|_F^2  \geq (C_2\delta^2/\eps)^{2\log d}$ (otherwise the algorithm has already terminated), 
it suffices to choose an $\epsilon'$ that satisfies $\eps' \gtrsim \frac{(C_2\delta^2/\eps)^{2\log d}}{\eps (CdR^2 + 1 + C_1 \delta^2/\eps)^{{C \log d}}}$.}
This gives an upper bound on the cardinality of the set $S_{\text{cover}}$, 
which gives the upper bound on the size of the JL-sketch, i.e., $\innerl$. 
We provide explicit calculations in  \Cref{remark:choice_of_L}. 
\end{proof}

 \subsection{Establishing  \Cref{cond:streaming}} \label{sec:concentration}

Throughout this section, we assume sample access to the distribution $P_t$. 
As mentioned earlier, \Cref{alg:reweighting} can simulate this 
by drawing a sample $x$ from $P$, calculating $w_t(x)$ 
(with poly-logarithmic cost in terms of running time and memory), 
and rejecting the sample with probability $1-w_t(x)$. 
With high probability, rejection sampling can increase the sample complexity 
by at most a constant factor because $\E[w_t(X)]\geq 1- O(\eps)$ (cf. \Cref{cl:invariant-weight}).

\subsubsection{\Cref{as:mean_est}} \label{sec:proof_of_estimator}
We establish \Cref{as:mean_est} in the following, which is proved in \Cref{sec:omitted_last}.

\begin{restatable}{lemma}{FinalStep} \label{lem:final_step}
In the setting of \Cref{alg:streaming}, there exist estimators $\widehat{\mu}_t$ such that, 
with probability at least $1-\failp$, for all $t \in [\outerl]$ we have that $\snorm{2}{\widehat{\mu}_t - \mu_t} \leq \delta/100$.
Furthermore, each $\widehat{\mu}_t$ can be computed  on a stream of
$n = O\left( \frac{R^2}{\delta^2/\eps} \log(\outerl/\failp) + \frac{d(1+\delta^2/\eps)}{\delta^2} \log(\outerl/\failp)\right)$ 
independent samples from $P_t$, in time $O(nd\log(\outerl/\failp))$ and using memory $O(d\log(\outerl/\failp))$.
\end{restatable}

\subsubsection{\Cref{as:lambda_samples,as:frob,as:scores} } \label{sec:itemsindetail}
Given that \Cref{as:mean_est} holds,
in this section, we show that \Cref{as:lambda_samples,as:frob,as:scores} of \Cref{cond:streaming}
hold with high probability  if sufficiently many samples from the underlying corrupted distribution $P$ are drawn. 
Let the  scores  $\widehat{g}_t(x):= \|\widehat{\bM}_t(x-\mu_t)\|_2^2$, 
where $\widehat{\bM}_t$ is the following sample-based estimator of $\bM_t$:
\begin{enumerate}
\item \label{step:1}Draw a batch $S_0$ of $\tilde{n}$ samples from $P$ 
and let the estimate $\widehat{W}_{t} = \E_{X \sim \cU(S_0)}[w_t(X)]$.
\item Let $P_t'$ be the distribution of the differences $(X-X')/\sqrt{2}$ for two independent $X,X' \sim P_t$.
\item Draw $\log d$ batches $S_1,\ldots, S_{\log d}$ of $\tilde{n}$ samples, each from $P_t'$. \label{it:draw_samples}
\item For $k \in [\log d]$, 
\begin{enumerate}
        \item Let $\widehat{\vec{\Sigma}}_{t,k} = \frac{1}{\tilde{n}}\sum_{x \in S_k} x x^T$. \label{step:sigma}
        \item Let $\widehat{\bB}_{t,k} = \widehat{W}_{t}^2 \widehat{\vec{\Sigma}}_{t,k} - (1-C_1\delta^2/\eps)\bI_d$.
\end{enumerate}
    \item \label{step:5} Return $\widehat{\bM}_t = \prod_{k=1}^{\log d} \widehat{\bB}_{t,k}$.
\end{enumerate}
{
\begin{remark}\label{remark:implementation}
\Cref{alg:streaming} does not need to calculate or store $\widehat{\bM}_t$ because it requires only that we can calculate products of $\widehat{\bM}_t$ with vectors $z$ as in \Cref{line:multiply}. This operation can be implemented in linear runtime and memory. Given the description of the estimator above, it suffices to show how to multiply $\widehat{\vec{\Sigma}}_{t,k}$ by a vector $z$ in linear time and memory. To this end, we observe that $\widehat{\vec{\Sigma}}_{t,k} z = \frac{1}{\tilde{n}}\sum_{x \in S_k} x (x^T z)$, thus by calculating the inner product $(x^T z)$ first, the result can be found in $O(nd)$ time in a streaming fashion.
\end{remark}}

We will show that \Cref{as:lambda_samples,as:frob,as:scores}  
of \Cref{cond:streaming} follow if 
$\|\widehat{\vec M}_t  - \bM_t\|_2 \leq 0.01 \min\left(\frac{\delta/\eps}{R},\frac{1}{\sqrt{d}}\right) \| \bM_t\|_F$ 
(cf. \Cref{lem:prod_concentration}) and 
$\|\widehat{\mu}_t - \mu_t\|_2 = O(\delta)$ (cf. \Cref{lem:final_step}). 

\begin{lemma} \label{lem:parts2and3}
Suppose that the estimators $\widehat{\bM}_t$ in \Cref{alg:streaming} 
are defined by the procedure as in Lines \ref{step:1} to \ref{step:5} above.
Let $C$ be a sufficiently large constant and assume that the dimension is 
$d = \Omega(1)$\footnote{This is without loss of generality as we could avoid the JL-sketch when $d = O(1)$.}.
{Further assume that $\|\widehat{\mu}_t - \mu\| \leq 0.01 \delta$.}
If $\tilde{n} \geq C {R^2 (\log d)^2  } \max\left(d, \frac{\eps^2d}{\delta^4},\frac{R^2\eps^2}{\delta^2}, \frac{R^2\eps^4}{\delta^6} \right)    \log\left(  \frac{ \outerl d \log d}{\failp}\right)$, 
then \Cref{as:lambda_samples,as:frob,as:scores} hold with probability at least $1-\failp$.
\end{lemma}
\begin{proof}
For now, we will assume that for all 
$t \in [\outerl]$,  $\|\widehat{\vec M}_t  - \bM_t\|_2 \leq 0.01  \min\left(\frac{\delta/\eps}{R},\frac{1}{\sqrt{d}}\right) \| \bM_t\|_F$ 
with the claimed sample complexity. This follows from  \Cref{lem:prod_concentration} and a union bound.
We will prove each of these conditions separately. 
\item
    \paragraph{Proof of \Cref{as:scores}:}
Denote $T:= 0.1 \delta/\eps$. Fix an iteration $t \in [\outerl]$. 
We will prove that the conditions hold in the $t$-th iteration 
with probability at least $1-\failp/\outerl$ and then a union bound will conclude the proof. 
    
{Define $\widehat{g}_t(x) := \|\vec M_t(x - \widehat{\mu}_t)\|_2^2$. 
Since $\| \widehat{\mu}_t - \mu_t\|_2 \leq 0.01\delta $, 
we have the following relation between $\widehat{g}_t$ and 
$g_t$:$  \|\vec M_t(x - \widehat{\mu}_t)\|_2^2 \geq 0.5 \|\vec M_t(x -\mu_t)\|_2^2 -  T^2\|\vec M_t\|_F^2$,
i.e., $\widehat{g}_t(x) \geq 0.5 g_t(x) - T^2 \| \vec M_t \|_F^2$.} 
We have that  $\|\widehat{\vec M}_t  - \bM_t\|_2 \leq 0.1 (T/R) \| \bM_t\|_F$ 
with probability at least $1-\failp/\outerl$ (see \Cref{lem:prod_concentration}).
This means that for any point $x$ with $\|x-\widehat{\mu}_t\|_2\leq 2R$, 
we have $\|\vec{M}_t (x-\widehat{\mu}_t)\|_2 \leq \|\widehat{\vec{M}}_t  (x-\widehat{\mu}_t)\|_2 + 0.2T \|\vec{M}_t\|_F$, 
which implies that 
$\|\vec{M}_t  (x-\widehat{\mu}_t)\|^2_2 \leq 2\|\widehat{\vec{M}}_t  (x-\widehat{\mu}_t)\|^2_2 + 0.08T^2\|\vec{M}_t\|_F^2$, 
or equivalently
\begin{align}
\|\widehat{\vec{M}}_t  (x-\widehat{\mu}_t)\|^2_2 \geq 0.5\widehat{g}_t(x) - 0.04T^2 \tr(\bM_t^2)\;. \label{eq:firststep}
\end{align}
The final step is taking the Johnson-Lindenstrauss sketch of $\widehat{\bM}_t$, 
which gives the matrix $\bU_t$ used in the definition of $\tilde{g}_t$. 
By repeating the proof of \Cref{cl:jl-cons} with $\widehat{\bM}_t$ in place of ${\bM}_t$, 
we get that if $\innerl = C \log((|S_{\text{cover}}|+d)\outerl/\failp)$, 
then $\tilde{g}_t(x) \geq 0.8 \|\widehat{\vec{M}}_t  (x-\widehat{\mu}_t)\|_2$ 
for all the points in the set $S_{\text{cover}}$ (the cover from \Cref{cor:ran_cov_specific}).
The value used for $\innerl$ in \Cref{alg:streaming} satisfies this condition (c.f. \Cref{remark:choice_of_L}). 
Combining this with \Cref{eq:firststep} and the relation between $\widehat{g}_t$ and $g_t$, 
we get that $\tilde{g}_t(x) \geq 0.4 \widehat{g}_t(x) - 0.04T^2 \| \vec M_t \|_F^2  \geq 0.2 g_t(x) - 0.44T^2 \|\vec M_t\|_F^2$.

\paragraph{Proof of \Cref{as:frob}:} 
Again, fix a $t \in [\outerl]$. We have that 
$\|\widehat{\vec M}_t  - \bM_t\|_2 \leq \frac{0.01}{\sqrt{d}}\| \bM_t\|_F$ 
with probability $1-\failp/\outerl$ using \Cref{lem:prod_concentration}. 
We thus have that  $ \|\vec M_t - \widehat{\vec M}_t \|_F \leq \sqrt{d} \|\vec M_t - \widehat{\vec M}_t \|_2 \leq 0.01 \|\vec M_t\|_F$, which implies
    \begin{align} \label{eq:abs_bound}
        \abs[\Big]{\|\vec M_t\|_F - \|\widehat{\vec M}_t \|_F} \leq 0.01 \|\vec M_t\|_F\;.
    \end{align}
It is easy to check that this is stronger than what we initially wanted. 
Indeed, squaring \Cref{eq:abs_bound} and using $\|\widehat{\vec M}_t \|_F \leq 1.01\|{\vec M}_t \|_F$ gives 
    \begin{align*}
        \|\widehat{\vec M}_t \|_F^2 \leq 2 \|\vec M_t\|_F \|\widehat{\vec M}_t \|_F - \|{\vec M}_t \|_F^2 
        + (0.01)^2 \|\vec M_t\|_F^2 <  1.1 \|\vec M_t\|_F^2\;,
    \end{align*}
which means that  $\|\widehat{\vec M}_t \|_F^2 - \|\vec M_t\|_F^2 \leq 0.1 \|\vec M_t\|_F^2$. 
For the other bound, \Cref{eq:abs_bound} implies
\begin{align*}
\|\vec M_t\|_F^2 &\leq 2 \|\vec M_t\|_F \|\widehat{\vec M}_t \|_F - \|\widehat{\vec M}_t \|_F^2+ (0.01)^2 \|\vec M_t\|_F^2 \\
&\leq 2 \|\widehat{\vec M}_t \|_F^2 + 2(0.01) \|\vec M_t\|_F \|\widehat{\vec M}_t \|_F - \|\widehat{\vec M}_t \|_F^2 + (0.01)^2\|\vec M_t\|_F^2 \\
&< 1.1 \|\vec M_t\|_F^2\;,
\end{align*}
which means that $\|\vec M_t\|_F^2 - \|\widehat{\vec M}_t \|_F^2  \leq 0.1 \|\vec M_t\|_F^2$. 
Therefore we obtain the following 
\begin{align}\label{eq:trace_approximation}
\abs{\|\vec M_t\|_F^2 - \|\widehat{\vec M}_t \|_F^2}  \leq 0.1 \|\vec M_t\|_F^2 \;.
\end{align}
Finally, the Johnson-Lindenstrauss step is exactly as described in the proof of  \Cref{as:scores}. 
    
\paragraph{Proof of  \Cref{as:lambda_samples}:}
Since we cannot access the same samples twice, 
the power-iteration algorithm now uses a different dataset in every step.
Let the matrix $\widehat{\bM}_t$ as in the beginning of \Cref{sec:concentration}. 
We have already shown in the previous paragraph that, 
with probability $1-\failp$, for all $t \in [\outerl]$, $\widehat{\bM}_t$ 
has Frobenius norm close to that of $\bM_t$ (\Cref{eq:trace_approximation}).
    For the rest of the proof, we condition on this event.    
    Consider the algorithm that calculates $v = \widehat{\bM}_t w$, 
    where $w \sim \cN(0,\bI_d)$ (this can be done in the streaming model by multiplying 
    with $\widehat{\vec B}_{t,k}$ iteratively; moreover this multiplication 
    can be implemented in time $O(\tilde{n}d)$). 
    We claim that the value $ \widehat{\lambda}_t = \|v\|_2^{1/\log d}$ 
    satisfies the desired relation. First, with at least constant non-zero probability 
    \begin{align} \label{eq:traceineq1}
        0.2 \tr(\widehat{\bM}_t^2) \leq  \|v\|_2^2 \leq 10\tr(\widehat{\bM}_t^2) \;,
    \end{align}
    where the one direction follows by Markov's inequality and the other one by \Cref{fact:var-quadratic}. 
    \Cref{eq:traceineq1,eq:trace_approximation} imply that ${0.1} \tr({\bM}_t^2) \leq \|v\|_2^2 \leq 11\tr({\bM}_t^2)$. 
    Furthermore, we have that
    \begin{align}\label{eq:traceineq3}
         \|v\|_2^{\frac{1}{\log d}} \leq (11 \tr({\bM}_t^2))^{\frac{1}{2\log d}} \leq \left(11  d \|\bB_t\|_2^{2 \log d} \right)^{\frac{1}{2\log d}} \leq 20 \|\bB_t\|_2  \;.
    \end{align}
    Similarly, for the lower bound, we have that
    \begin{align}\label{eq:traceineq4}
        \|v\|_2^{\frac{1}{\log d}} \geq ({0.1} \tr({\bM}_t^2))^{\frac{1}{2\log d}} \geq ({0.1})^{\frac{1}{2\log d}}  \|\bB_t\|_2 \geq ({0.1})  \|\bB_t\|_2\;,
    \end{align}
    where in the last inequality we assumed that the dimension is sufficiently large. Putting \Cref{eq:traceineq1,eq:trace_approximation,eq:traceineq3,eq:traceineq3} together, 
    with constant probability, we have that  $\new{0.1} \|\bB_t\|_2 \leq \|v\|_2^{1/\log d}\leq 20 \|\bB_t\|_2$.
    By repeating the procedure $O(\log(1/\tau'))$ times and taking the median, 
    we boost the probability of failure to $\tau'$. By union bound,  
    choosing $\tau'=\failp/\outerl$ makes the event hold for all iterations $t \in [\outerl]$ simultaneously with probability $1-\failp$.
\end{proof}

The remainder of this section is dedicated to showing that 
$\|\widehat{\vec M}_t  - \bM_t\|_2 \leq \min\left(\frac{\delta/\eps}{R},\frac{0.01}{\sqrt{d}}\right) \| \bM_t\|_F$.
We require the following lemma, which we prove in \Cref{sec:omitted_concentration}. 
We use $\prod_{i=1}^p \bB_i $ is to denote the matrix product $\bB_1 \bB_2 \cdots \bB_p$.

\begin{restatable}{lemma}{MatrixProp} \label{lem:matrix_prop} 
Let $\bA, \bB, \bB_1,\dots,\bB_p$ be symmetric $d\times d$ matrices and define $\bM = \bB^p, \bM_S = \prod_{i=1}^p \bB_i $. 
If $\|\bB_i - \bB\|_{2} \le \delta \|\bB\|_{2}$, then $\|\bM_S - \bB^p\|_{2} \leq p \delta(1 + \delta)^{p} \|\bB\|_{2}^p$. \label{prop:2}
\end{restatable}

We are now ready to prove our main technical result.
\begin{lemma} \label{lem:prod_concentration}
Assume that $\|\bB_t\|_2 \geq (C_1/2)\delta^2/\eps$ and $\E_{X \sim P}[w_t(X)]\geq 1-O(\eps)$ 
hold in the $t$-th iteration of \Cref{alg:streaming}. If $\widehat{W}$ and every $\widehat{\vec B}_{t,k}$ 
in the product $\widehat{\vec M}_t = \prod_{k=1}^{\log d}\widehat{\vec B}_{t,k}$ is calculated using 
\begin{align*}
\tilde{n} \geq C {R^2 (\log d)^2  } \max\left(d, \frac{\eps^2d}{\delta^4},\frac{R^2\eps^2}{\delta^2}, \frac{R^2\eps^4}{\delta^6} \right)   
 \log\left(  \frac{d \log d}{\failp}\right)
\end{align*}
samples, where $C$ is a sufficiently large constant,  we have that 
\begin{align*}
 \|\widehat{\vec M}_t  - \bM_t\|_2 \leq 0.01  \min \left(\frac{\delta/\eps}{R},\frac{1}{\sqrt{d}}\right) \| \bM_t\|_F\;,
\end{align*}
with probability at least $1-\failp$.
\end{lemma}
\begin{proof}
Let $T:=\delta/\eps$ and $p:=\log d$ for brevity. 
Using \Cref{lem:matrix_prop} we have that 
$ \|\widehat{\vec M}_t  - \bM_t \|_2 \leq p\gamma e^{\gamma p} \|\bM_t \|_2 $, 
where $\gamma>0$ is such that
    \begin{align} \label{eq:opnorm_conc}
        \| \widehat{\bB}_{t,k} - \bB_{t} \|_2 \leq \gamma \|\bB_{t}\|_2
    \end{align}
    for all $k \in [p]$. Therefore, for the lemma to hold, 
    it suffices that $p\gamma e^{\gamma p}  \leq 0.01  \min\left( \frac{T \|\bM_t \|_F}{R\|\bM_t \|_2},\frac{1}{\sqrt{d}}\right)$.
For that, it suffices to choose $\gamma = \frac{0.01}{3p}\min\left( \frac{1}{\sqrt{d}},\frac{T \|\bM_t \|_F}{R\|\bM_t \|_2} \right)$. 
At this point, we also assume two things: 
First, that the estimate $\widehat{\vec \Sigma}_{t,k}$ (defined in step \ref{step:sigma}) 
is such that $\|\widehat{\vec \Sigma}_{t,k} - \vec{\Sigma}_t\|_2\leq \eps'\|\vec{\Sigma}_t\|_2$ 
for some $\eps'$ to be specified later on. Second, that we have an estimate $\widehat{W}_{t}$ 
for $\E_{X\sim P}[w_t(X)]$ such that 
    \begin{align}\label{eq:hoef}
         \widehat{W}_{t}=\E_{X\sim P}[w_t(X)] + \eta\;,
    \end{align}
with $|\eta| \leq \xi$ for some $\xi\leq 1$ to be decided later. 
By Hoeffding's inequality, if we compute $\widehat{W}$ as shown in Step \ref{step:1}, 
then $\log(2/\failp)/\xi^2$ samples suffice to guarantee that \Cref{eq:hoef} holds 
with probability $1-\failp/2$. We now focus on \Cref{eq:opnorm_conc}. We note that
\begin{align*}
\| \widehat{\bB}_{t,k} - \bB_{t} \|_2 
&= \left\|\left(\E_{X\sim P_t}[w_t(X)]+\eta \right)^2 \widehat{\vec \Sigma}_{t,k} - \E_{X\sim P_t}[w_t(X)]^2 \vec{\Sigma}_{t} \right\|_2 \\
&\leq \E_{X\sim P_t}[w_t(X)]^2 \|\widehat{\vec \Sigma}_{t,k} - \vec{\Sigma}_t\|_2 
+ (\eta^2 + 2 \eta \E_{X\sim P_t}[w_t(X)] ) \|\widehat{\vec \Sigma}_{t,k}\|_2 \\
&\leq \|\widehat{\vec \Sigma}_{t,k} - \vec{\Sigma}_t\|_2  + 3\xi (\|\widehat{\vec \Sigma}_{t,k} - \vec{\Sigma}\|_2 + \| \vec{\Sigma}_t\|_2) \\
&= (1+3\xi) \|\widehat{\vec \Sigma}_{t,k} - \vec{\Sigma}_t\|_2  + 3\xi \| \vec{\Sigma}_t\|_2  \;.
\end{align*}
By choosing $\xi = \min(1,\eps'/3)$ and $\eps' = \frac{1}{5}\gamma \|\bB_t\|_2/\|\vec{\Sigma}_t\|_2$, 
the above implies that \Cref{eq:opnorm_conc} holds.
Thus, it suffices to show that  $\|\widehat{\vec \Sigma}_{t,k} - \vec{\Sigma}_t\|_2\leq \eps' \|\vec{\Sigma}_t \|_2$ 
for our choice of $\eps'$. Note that \Cref{fact:ver_cov} is not directly applicable to the distribution $P_t$ 
since it does not have zero mean. This is why we are working with samples of the form $(X-X')/\sqrt{2}$. 
By \Cref{fact:ver_cov} with $\eps'$ set as above and $\tau=\failp/(2p)$, 
we have the following upper bound on the sufficient number of samples: 
\begin{align}
\tilde{n} &= C \frac{R^2 }{{\eps'}^2 \| \vec \Sigma_t \|_2} \log\left(  \frac{2pd}{\failp}\right) \notag\\
&\lesssim \frac{R^2 }{\gamma^2 }  \frac{\| \vec \Sigma_{t} \|_2}{\| \bB_{t} \|_2^2}\log\left(  \frac{pd}{\failp}\right) \notag \\
&\lesssim \frac{R^2 \eps }{ \delta^2 \gamma^2 }  \frac{\| \vec \Sigma_{t} \|_2}{\| \bB_{t} \|_2}\log\left(  \frac{pd}{\failp}\right) \tag{using $\|\bB_t\|_2 \geq (C_1/2) \delta^2/\eps$} \\
&\lesssim  \frac{R^2 \eps }{\delta^2\gamma^2 }  \max\left(1, \frac{\eps}{\delta^2}  \right)  \log\left(  \frac{pd}{\failp}\right) \label{eq:fourth_line} \\
&\lesssim \frac{R^2p^2 \eps  }{\delta^2} \max\left(d, \frac{R^2\eps^2}{\delta^2}\frac{\|\bM_t\|^2_2}{\|\bM_t\|^2_F}  \right)    \max\left(1, \frac{\eps}{\delta^2}  \right)\log\left(  \frac{pd}{\failp}\right)  \notag\\
&\leq\frac{R^2p^2 \eps  }{\delta^2} \max\left(d, \frac{R^2\eps^2}{\delta^2} \right)    \max\left(1, \frac{\eps}{\delta^2}  \right)\log\left(  \frac{pd}{\failp}\right) \tag{using $\|\bM_t\|_2\leq \|\bM_t\|_F$}  \\
&\leq \frac{R^2p^2 \eps  }{\delta^2} \max\left(d, \frac{\eps d}{\delta^2},\frac{R^2\eps^2}{\delta^2}, \frac{R^2\eps^3}{\delta^4} \right)    \log\left(  \frac{pd}{\failp}\right) \label{eq:sample_compl}\;,
\end{align}

\Cref{eq:fourth_line} is derived as follows: 
First we note that $ \| \vec \Sigma_t \|_2 \leq \frac{\| \vec B_t \|_2 + 1}{\E_{X \sim P}[w_t(X)]^2} \lesssim \| \vec B_t \|_2 + 1$, 
where the last inequality uses our assumption that 
$\E_{X \sim P}[w_t(X)]\geq 1-O(\eps)$. 
We combine this  with $\| \vec B_t \|_2 \gtrsim \delta^2/\eps$ as follows:
\begin{align*}
\frac{\| \vec \Sigma_t \|_2^2}{\| \vec B_t \|_2^2 } 
\lesssim \frac{(\| \vec B_t \|_2 + 1 )^2}{\| \vec B_t \|_2^2} \leq 2 + \frac{2}{\| \vec B_t \|_2^2} 
\lesssim 1 + \frac{1}{(\delta^2/\eps)^2} \;.
\end{align*}
Regarding the samples required to achieve \Cref{eq:hoef}, a sufficient number is
\begin{align*}
\tilde{n} &= C \log\left(\frac{2}{\failp}\right) \frac{1}{\xi^2} \\
&\lesssim  \log\left(\frac{1}{\failp}\right) \max\left( 1,\frac{1}{\eps'^2}  \right) \\
&\lesssim   \log\left(\frac{1}{\failp}\right) \max\left( 1,\frac{1}{\gamma^2}\frac{\|\vec \Sigma_t\|^2_2}{\| \bB_t\|_2^2}  \right) \\
&\lesssim \log\left(\frac{1}{\failp}\right) \max\left( 1,\frac{1}{\gamma^2}\max\left( 1,\frac{\eps^2}{\delta^4} \right)  \right) \\
&\lesssim \log\left(\frac{1}{\failp}\right) \max\left( 1,p^2\max\left(d,\frac{R^2\eps^2}{\delta^2}  \right)\max\left( 1,\frac{\eps^2}{\delta^4} \right)  \right)\\
&\lesssim \log\left(\frac{1}{\failp}\right)\max\left( p^2d, \frac{p^2R^2\eps^2}{\delta^2}, \frac{p^2d\eps^2}{\delta^4}, \frac{p^2R^2\eps^4}{\delta^6} \right)\;,
\end{align*}
which is smaller compared to the right-hand side of \Cref{eq:sample_compl}.
\end{proof}

\subsubsection{\Cref{stopping_cond_as}} \label{sec:remainingparts}

The following lemma establishes that the estimator of  \Cref{stopping_cond_as} 
is accurate when called once. By using a union bound on the maximum number of times that it can be called, 
we get the sample complexity requirement of 
$n = O\left((R^2/(\delta^2/\eps)) \polylog\left( d, R, \frac{1}{\eps},\frac{1}{\failp} \right)\right)$.

\begin{lemma} \label{cl:evaluate_cond}
{Consider the context of \Cref{alg:streaming}} and denote $T_t:= c\widehat{\lambda}_t \|\bU_t\|_F^2$. 
Given a weight function $w: \R^d \to [0,1]$, there exists an estimator $f(w)$ on 
$n=O(\frac{R^2}{\delta^2/\eps}\log(1/\failp))$ samples such that, 
if $\E_{X \sim P}[w(X)\tilde{\tau}_t(X)] >  T_t$, then with probability at least $1-\failp$, $f > T_t/2$. 
Similarly, $\E_{X \sim P}[w(X)\tilde{\tau}_t(X)] <  T_t$ implies $f < (3/2)T_t$. 
{Moreover, the estimator uses $O(\log(1/\failp))$ memory and runs in $O(n d)$ time}.
\end{lemma}

\begin{proof}

We show the first direction; the other one has a symmetric proof. 
Suppose $\E_{X \sim P}[w(X)\tilde{\tau}_t(X)] >  c\widehat{\lambda}_t \|\bU_t\|_F^2$. 
It suffices to show that with probability at least $0.9$ we have that 
\begin{align} \label{eq:whatwewant}
\frac{1}{n}\sum_{i=1}^N w(X_i)\tilde{\tau}_t(X_i) >
 \frac{3}{4} \E_{X \sim P}[w(X)\tilde{\tau}_t(X)] - \frac{1}{4}c\widehat{\lambda}_t \|\bU_t\|_F^2 \;,
\end{align}
as we can repeat the procedure $O(\log(1/\failp))$ times and take the majority 
vote to boost the probability to $1-\failp$. 
By Chebyshev's inequality, we have that with probability $0.9$ it holds that  
\begin{align*}
\frac{1}{n}\sum_{i=1}^n w(X_i)\tilde{\tau}_t(X_i) 
> \E_{X \sim P}[w(X)\tilde{\tau}_t(X)] - \sqrt{\frac{10\Var_{X \sim P}(w(X) \tilde{\tau}(X))}{n}} \;.
\end{align*}
Therefore, it suffices to have 
$\sqrt{\frac{10\Var_{X \sim P}(w(X) \tilde{\tau}(X))}{n}} 
\leq \frac{1}{4} \E_{X \sim P}[w(X)\tilde{\tau}_t(X)] + \frac{1}{4}c\widehat{\lambda}_t \|\bU_t\|_F^2$, 
and thus we need $n$ to be a sufficiently  large multiple of 
$\Var_{X \sim P}(w(X) \tilde{\tau}_t(X))/(\E_{X \sim P}[w(X)\tilde{\tau}_t(X)] + c\widehat{\lambda}_t \|\bU_t\|_F^2)^2$. 
For that, it suffices to choose
\begin{align*}
n = \Theta\left( \frac{\Var_{X \sim P}(w(X) \tilde{\tau}_t(X))}{ \E_{X \sim P}[w(X)\tilde{\tau}_t(X)] c\widehat{\lambda}_t \|\bU_t\|_F^2} \right) \;.
\end{align*}
We now focus on bounding by above the right-hand side.
Let $T_t':= C_3\widehat{\lambda}_t \|\bU_t\|_F^2$ be the threshold 
used in the definition of $\tilde{\tau}_t(x)=\tilde{g}_t(x)\1\{ \tilde{g}_t(x) \geq T_t'\}$. 
For the variance we have that 
    \begin{align}
        \Var(w(X) \tilde{\tau}_t(X)) &\leq \E_{X \sim P}[((w(X) \tilde{\tau}_t(X))^2] \notag\\
        &\leq \E_{X \sim P}[w(X) \tilde{\tau}^2_t(X)]  \notag\\
        &\lesssim  R^2 \|\bU_t \|_F^2 \E_{X \sim P}[ w(X)\tilde{\tau}_t(X)] \;, \label{eq:var}
    \end{align}
     where the last inequality uses that 
     $\E_{X \sim P_t}[ \tilde{\tau}^2_t(X)] = \E_{X \sim P_t}[ \tilde{g}_t^2(X) \1\{\tilde{g}_t(X) \geq T_t'\}]$ 
     and bounds from above the one of the two factors of $\tilde{g}_t$ as follows:
    \begin{align}
    \tilde{g}_{t}(x) = \|\bU_t(x-\mu_t)\|_2^2 \leq \| \bU_t \|_F^2 R^2 \;,
    \end{align}
    where $\bU_t$ is the matrix used in  \Cref{line:scoresdef1} of the algorithm.
    Using \Cref{eq:var}, the number of samples that suffice  can  now be bounded as follows:
    \begin{align*}
        \frac{\Var_{X \sim P}(\tilde{\tau}_t(X))}{\widehat{\lambda}_t\E_{X \sim P}[w(X)\tilde{\tau}_t(X)]  \|\bU_t\|_F^2} 
        \lesssim  \frac{  R^2 \|{\bU}_t \|_F^2 }{\widehat{\lambda}_t  \|\bU_t\|_F^2}
        \lesssim \frac{ R^2  }{ \delta^2/\eps } \;,
    \end{align*}
    where we used  that $\widehat{\lambda}_t >C_2\delta^2/\eps$ from \Cref{line:lambda1} of our algorithm.
\end{proof}

\section{Applications: Beyond Robust Mean Estimation}%
\label{sec:applications-main}

In this section, we develop robust streaming algorithms with near-optimal space complexity
for more complex statistical tasks, specifically 
for robust covariance estimation and robust stochastic optimization. 
The main idea enabling these applications is that these
tasks can be effectively reduced to robust mean estimation.

\subsection{Robust Covariance Estimation}
\label{sec:robust_covariance_estimation}
In this subsection, we study the problem of estimating the covariance matrix $\vec \Sigma$ 
of a distribution $D$, having access to $\eps$-corrupted samples from $D$ 
in the sense of \Cref{def:oblivious}. Let $X \sim D$ and the Kronecker product 
$Y = X \otimes X$. Note that $\E[Y] = \vec \Sigma^\flat$, where $\flat$ denotes the flattening operation. 
Then, using any robust mean estimation algorithm on this $d^2$-dimensional distribution, 
one efficiently compute a vector close to $\vec \Sigma^\flat$ in $\ell_2$-norm, 
which translates to a Frobenius-norm guarantee for $\vec \Sigma$. 
Of course, our mean estimator works as long as the distribution of $Y$ is stable. 
If $\cov[Y]$ is bounded from above by a multiple of the identity matrix, 
then $Y$ is $(\eps,O(\sqrt{\eps}))$-stable with respect to $\vec \Sigma^\flat$, 
and thus we get the following as a corollary of \Cref{th:main}:

\begin{theorem}[Robust Covariance Estimation for Distributions with Bounded Moments] \label{thm:cov-est}
Let a distribution $D$ with $\cov_{X \sim D}[X \otimes X] \preceq  \bI_{d^2}$ 
and denote by $\vec \Sigma$ its covariance matrix. Let $d \in \Z_+$,  $0<\tau<1$ 
and $0<\eps<\eps_0$ for a sufficiently small constant $\eps_0$. 
There exists an algorithm that given $\eps,\tau$ and a set of 
$n=(d^4/\eps)\polylog(d,1/\eps,1/\tau)$ samples in the single-pass streaming 
model of \Cref{def:streaming} from a distribution $Q$ with $\dtv(D,Q)\leq \eps$, 
runs in time $n d^2 \polylog\left(d,1/\eps,1/\tau  \right)$, 
uses memory $d^2 \polylog\left(d,1/\eps,1/\tau  \right)$, 
and outputs a matrix $\widehat{\vec \Sigma}$ such that 
$\| \widehat{\vec \Sigma} - \vec \Sigma \|_F = O(\sqrt{\eps})$, 
with probability at least $1-\tau$.
\end{theorem}

For the special case when $D$ is Gaussian 
we have that the fourth moment tensor of $D$ is bounded:
\begin{fact}[see, e.g.,~\cite{cheng2019faster}]
\label{fact:fourth_moment_bound}
Let $X \sim \cN(0,\vec \Sigma)$ with $\vec \Sigma \preceq \bI_d$ and $Y = X \otimes X$. Then, $\cov[Y] \preceq 2\bI_{d^2}$.
\end{fact}
Using the above fact, we have that the guarantees of \Cref{thm:cov-est} 
hold in the Gaussian case, giving an algorithm for $O(\sqrt{\eps})$-approximation in Frobenius norm. 
However, the information-theoretic lower bound for covariance estimation 
of the Gaussian distribution is of the order of $\eps$. 
We can plug-in our streaming robust mean estimation algorithm 
to the covariance estimator given in \cite{cheng2019faster}, and achieve the nearly-optimal error of $O(\eps \log(1/\eps))$. 
This algorithm creates a series of estimates $\widehat{\vec \Sigma}_i$. 
At the $(i+1)$-th step, all samples are multiplied by ${\widehat{\vec \Sigma}_i}^{-1/2}$ thus, 
given that $\widehat{\vec \Sigma}_i$ is a good approximation for $\vec \Sigma$, 
this makes the distribution of the transformed samples closer to $\cN(0,\bI_d)$, 
which in turn allows us to produce a better approximation $\widehat{\vec \Sigma}_{i+1}$ of $\vec \Sigma$.
The resulting guarantees are summarized in the following theorem.

\begin{theorem}[Robust Gaussian Covariance Estimation] 
\label{thm:covariance_application_better_error}
Let $Q$ be a distribution on $\R^d$ with $\dtv(Q,\cN(0,\vec \Sigma))\leq \eps$ 
and assume that $ \frac{1}{\kappa}\bI_d \preceq \vec \Sigma \preceq \bI_d$, for some $\kappa>0$.
There is a single-pass streaming algorithm that uses 
$n= (d^4/\eps^2) \polylog(d,\kappa,1/\eps,1/\tau)$
samples from $Q$, runs in time $n d^2 \polylog\left(d,\kappa,1/\eps,1/\tau  \right)$, 
uses memory $d^2 \polylog\left(d,\kappa,1/\eps,1/\tau  \right)$, 
and outputs a matrix $\widehat{\vec \Sigma}$ such that 
$\| \vec \Sigma^{-1/2}\widehat{\vec \Sigma}\vec \Sigma^{-1/2} - \bI_{d} \|_F 
= O(\eps \log(1/\eps))$, with probability at least $1-\tau$. 
\end{theorem}

The reader is referred to \Cref{sec:appendix_applications} for more details 
on using \Cref{alg:streaming} to obtain \Cref{thm:covariance_application_better_error}.

\subsection{Stochastic Convex Optimization} 
\label{sec:app-optimization}

Here we explore the implications of \Cref{alg:streaming} in outlier-robust stochastic convex optimization. 
This subsection crucially leverages the prior works \cite{PSBR18,DKK+19-sever}, 
which apply robust mean estimation algorithms to perform robust stochastic optimization. 
In particular, we follow the framework of \cite{PSBR18}.

Concretely, we study the following generic optimization problem: 
Let a parameter space $\Theta$, sample space $\cZ$, 
and a loss function $f(\theta ; z ) : \Theta \times \cZ \to \R^+$. 
For an unknown distribution $D$ over $\cZ$, the goal is 
to minimize the associated \emph{risk} $\bar{f}(\theta) = \E_{z \sim D}[f(\theta ; z)]$, 
given sample access to the distribution $D$. We will occasionally just write $f(\theta)$ 
instead of $f(\theta;z)$ when no confusion arises. 
This setup is central in machine learning, 
since it captures a plethora of learning tasks. For example, 
$f$ can be a negative log-likelihood function for the learning problem of interest, 
e.g., square loss for linear regression and logistic loss for  logistic regression. 
In the robust version of the problem, the algorithm has access only 
to an $\eps$-corrupted version of $D$ in the sense of \Cref{def:oblivious}. 

We start by recalling a generic optimization algorithm 
that works whenever $\bar{f}$ is $\tau_\ell$-strongly convex 
and $\tau_u$-smooth, i.e., for all $\theta_1,\theta_2 \in \Theta$, we have that
\begin{align*}
    \frac{\tau_\ell}{2}\| \theta_1-\theta_2 \|_2^2 \leq \bar{f}(\theta_1)- \bar{f}(\theta_2)  - (\nabla\bar{f}(\theta_2))^T(\theta_1-\theta_2) \leq \frac{\tau_u}{2}\| \theta_1-\theta_2 \|_2^2 \;.
\end{align*}
We then give specific applications for robust linear regression and logistic regression.

The work of \cite{PSBR18} provides an analysis 
of projected gradient descent assuming oracle access to approximations of the gradient:
\begin{definition}[$(\alpha,\beta)$-gradient estimator] \label{def:grad_est}
A function $g(\theta)$ is an $(\alpha,\beta)$-gradient estimator for $\bar{f}$ 
if $\|g(\theta) - \nabla \bar{f}(\theta) \|_2 \leq \alpha \| \theta - \theta^*\|_2 + \beta$, for every $\theta \in \Theta$.
\end{definition}

Denoting by $\eta$ the step size of gradient descent, define the following parameter:
\begin{align} \label{eq:contraction}
\kappa := \sqrt{1- \frac{2 \eta \tau_\ell \tau_u}{\tau_\ell + \tau_u}} + \eta \alpha \;.
\end{align}

  \begin{algorithm}[tb]
    \caption{Robust Gradient Descent} 
    \label{alg:robustGD}
\begin{algorithmic}[1]
   \State {\bfseries Input:}  $g(\cdot)$,  $\tau$
   \For{$t=0$ {\bfseries to} $T-1$}
          \State $\theta^{t+1} = \arg\min_{\theta \in \Theta}  \left\| \theta^t - \eta g(\theta)\right\|_2^2 $
   \EndFor
\end{algorithmic}
\end{algorithm}

\begin{theorem}[\cite{PSBR18}]\label{thm:robust-gd} 
Let the domains $\Theta,\cZ \subset \R^d$, a distribution $D$ over $\cZ$, 
and a loss function $f:\Theta \times \cZ \to \R^+$ such that 
$\bar{f}(\theta):= \E_{z \sim D}[f(\theta;z)]$ is $\tau_\ell$-strongly convex and $\tau_u$-smooth. 
Let $g$ be an $(\alpha,\beta)$-gradient estimator with $\alpha<\tau_\ell$.
Let $\kappa$ from \Cref{eq:contraction} and $\theta^*$ be the minimizer of $\bar{f}$. 
Then \Cref{alg:robustGD}, initialized at $\theta^0$ with step size $\eta = 2/(\tau_\ell + \tau_u)$, after 
    \begin{align}
        T = \log_{\frac{1}{\kappa}} \left(\frac{(1-\kappa)\|\theta^0-\theta^*\|_2}{\beta} \right) \label{eq:T}
    \end{align}
iterations, returns a vector $\widehat{\theta}$ such that 
\begin{align}\label{eq:final_error}
  \|\widehat{\theta} - \theta^*\|_2 \leq \frac{2}{1-\kappa}\beta \;.
\end{align}    
\end{theorem}

If the distribution of the gradients has bounded covariance, 
then one can use the low-memory estimator of the previous sections in place of $g(\cdot)$. 
This bound on the covariance will not necessarily be known to the algorithm, 
thus we first need to strengthen the robust mean estimator 
so that it is adaptive to that unknown scale. 
This can be done using Lepski's method \cite{lepskii1991problem,birge2001alternative} 
(the details are deferred to \Cref{sec:lepski}). Having that version of the estimator at hand, 
we then obtain the following statement (see \Cref{sec:appendix_applications} for the proof):

\begin{restatable}{corollary}{CorRobustGD} \label{cor:robust-gd}  
In the setting of \Cref{thm:robust-gd}, suppose that the distribution 
of gradients satisfies $\cov[\nabla {f}(\theta)] \preceq \sigma^2 \bI_d$ 
with  $\sigma^2 = \alpha^2 \|\theta-\theta^*\|^2_2 + \beta^2$ for all $\theta \in \Theta$, 
where $\alpha \sqrt{\eps}<\tau_\ell$. Assume that the radius of the domain 
$\Theta$, $r:= \max_{\theta \in \Theta}\|\theta\|_2$ is finite.  
There exists a single-pass streaming algorithm that given 
$O(T (d^2/\eps)\log(1+\alpha r/\beta) \polylog(d,1/\eps,T/\tau,1+\alpha r/\beta ) )$ samples, 
runs in time $T n d \, \polylog (d,1/\eps,  T/\tau, 1+\alpha r/\beta )$, 
uses memory $d \, \polylog (d,1/\eps, T/\tau,1+\alpha r/\beta )$, 
and returns a vector $\widehat{\theta} \in \R^d$ such that 
$\| \widehat{\theta} - \theta^* \|_2 = O(\sqrt{\eps} \beta/(1-\kappa) )$ 
with probability at least $1-\tau$.
\end{restatable}

We now proceed to more specific applications, 
where we work out the parameters $\alpha,\beta$ for some distributions of interest.

\subsubsection{Linear Regression}

For linear regression, we assume the following generative model:
\begin{align}
    Y = X^T\theta^* + Z \;, \label{eq:lr}
\end{align}
where $\theta^* \in \R^d$ belongs in the ball $\|\theta^*\|_2 \leq r$, 
$X \sim D_x$,  $Z \sim D_Z$ independently, and $D_Z$ has zero mean. 
The loss function that we use in this case is $f(\theta) = \frac{1}{2}(Y-\theta^TX)^2$, 
and the risk function is 
\begin{align*}
\bar{f}(\theta) = \E_{(X,Y)}\left[f(\theta)\right] = 
\frac{1}{2}(\theta - \theta^*)^T \E_{X \sim D_x}[XX^T] (\theta - \theta^*) + \frac{1}{2} \Var(Z) \;.
\end{align*}
Letting $\lambda_{\max}(\E[XX^T])$ and $\lambda_{\min}(\E[XX^T])$ 
denote the largest and smallest eigenvalue of $\E[XX^T]$ respectively, 
it can be checked that for any $\tau_\ell \leq \lambda_{\min}(\E[XX^T])$ 
and $\tau_u \geq \lambda_{\max}(\E[XX^T])$, 
$\bar{f}$ is $\tau_\ell$-strongly convex and $\tau_\ell$-smooth.

Since we want the distribution of gradients to be stable, 
we impose the following sufficient conditions on the distributions $D_x$ and $D_Z$.

\begin{assumption} \label{as:cond-for-stable-grad}
    The random variables $X,Z$ are independent and satisfy the following conditions:
    \begin{enumerate}
        \item $\E_{Z \sim D_Z}[Z] = 0$
        \item $\Var_{Z \sim D_Z}[Z] \leq \xi^2$
        \item $\gamma \bI_d \preceq \E_{X \sim D_x}[X X^T] \preceq \sigma^2 \bI_d$.
        \item For some constant $C>0$, for every $v \in \cS^{d-1}$, $\E_{X \sim D_x}[(X^Tv)^4] \leq C \sigma^4$.
    \end{enumerate}
\end{assumption}

As shown below, these assumptions imply that the resulting distribution 
of the gradients has bounded covariance (and thus is stable with respect to its mean).
\begin{lemma}[see, e.g., \cite{DKK+19-sever}] \label{lem:bounded_cov_linear_reg}
For $D_x,D_Z$ satisfying \Cref{as:cond-for-stable-grad}, for every $\theta \in \Theta$, we have that 
$\cov[\nabla {{f}}(\theta)] \preceq (4 \sigma^2 \xi^2 + 4 C \sigma^4 \|\theta-\theta^*\|_2^2) \bI_d$.
\end{lemma}

Having \Cref{lem:bounded_cov_linear_reg} in hand, \Cref{cor:robust-gd} gives the following.

\begin{theorem}[Robust Linear Regression; full version of \Cref{thm:lr-simple}] \label{thm:lr}
Consider the linear regression model of \Cref{eq:lr} and suppose that 
\Cref{as:cond-for-stable-grad} holds. Let $0<\eps<\eps_0$ for a sufficiently small constant $\eps_0$. 
Assume that $C \sigma^2 \sqrt{\eps} < \gamma /2$. 
Let $\kappa,T$ as in \Cref{eq:contraction,eq:T} with $\tau_\ell = \gamma$, $\tau_u = \sigma^2$. 
There is an algorithm that uses $n = T\cdot (d^2/\eps)\log(1+r\sigma/\xi) \,\polylog\left(d,1/\eps, T/\tau,1+r\sigma/\xi \right)$ samples, 
runs in time $T n d \, \polylog (d,1/\eps, T/\tau,1+r\sigma/\xi )$, 
uses memory $d \, \polylog (d,1/\eps, T/\tau,1+r\sigma/\xi )$, 
and returns a vector $\widehat{\theta} \in \R^d$ such that 
$\| \widehat{\theta} - \theta^* \|_2 = O(\sigma \xi \sqrt{\eps}/(1-\kappa))$ 
with probability at least $1-\tau$.
\end{theorem}
\begin{proof}
In our case, we have that $\tau_\ell = \gamma$ and $\tau_u = \sigma^2$. 
Given the bound of \Cref{lem:bounded_cov_linear_reg}, 
we use \Cref{cor:robust-gd} with $\alpha = 2C \sigma^2$ and $\beta=2\sigma\xi$. 
The requirement from that corollary that $\alpha\sqrt{\eps}\leq \tau_\ell$ 
becomes $C \sigma^2 \sqrt{\eps} < \gamma /2$. Moreover, $\alpha r/\beta = O(r\sigma/\xi)$.
\end{proof}

\subsubsection{Logistic Regression}

We consider the joint distribution of $X \in \R^d,Y \in \{0,1\}$, 
where $X \sim D_x$ and $Y$ given $X$ is Bernoulli random variable:
\begin{align} \label{eq:log-reg}
    Y|X \sim \text{Bernoulli}(p), \quad \quad \text{with} \; p = \frac{1}{1+ e^{-x^T\theta^*}} \;.
\end{align}
The loss function we are minimizing in this case is the negative log-likelihood, 
which eventually can be written as 
$f(\theta) = -(\theta^T x)y + \Phi(\theta^T x)$, where $\Phi(t) := \log(1+e^{t})$. 
Regarding the strong convexity parameters, the Hessian of $\bar{f}$ can be shown to be
\begin{align}\label{eq:hessian}
    \nabla^2 \bar{f}(\theta)= \E_{X \sim D_x}\left[ \frac{e^{{\theta}^T X}}{(1+e^{{\theta}^T X})^2} XX^T\right] \;.
\end{align}
The parameter space $\Theta$ needs to be bounded in order for the eigenvalues 
of the Hessian to remain away from zero; 
we thus use $\Theta = \{\theta \in \R^d \; : \; \|\theta\|_2^2 \leq r  \}$ 
with $r>0$ being a universal constant. 
We also impose the following assumptions on the covariates.

\begin{assumption} \label{as:for-log-reg}
    We assume the following for the distribution of $X$: 
    \begin{enumerate}
        \item $\E[X] = 0$.
        \item (\emph{concentration}) For some constant $C>0$, $\E[XX^T] \preceq C^2 \bI_d $.
        \item (\emph{anti-concentration}) There exists constant $c_1>0$ and $c_2 \in (0,1/2)$ such that 
        for every unit vector $v$, $\pr_{X \sim D_x}[(v^T X)^2  > c_1 \|v\|_2^2] \geq c_2$.
    \end{enumerate}
\end{assumption}

Under these assumptions, we have the following:

\begin{lemma}[Lemma 4 in \cite{PSBR18}] \label{lem:bounded_cov_lemma}
Supposing that \Cref{as:for-log-reg} holds, for every $\theta \in \Theta$, 
we have that $\cov[\nabla {f}(\theta)] \preceq O(1) \bI_d $.
\end{lemma}

The above lemma shows that the distribution of $\nabla {f}(\theta)$ is $(\eps, O(\sqrt{\eps}))$-stable, 
and thus using our robust mean estimation algorithm one can get an $(\alpha,\beta)$-gradient estimator 
with $\alpha=0$ and $\beta=O(\sqrt{\eps})$. 
This proves the following (see \Cref{sec:appendix_applications} for a detailed proof):

\begin{restatable}[Robust Logistic Regression; full version of \Cref{thm:log-reg-simple}]{theorem}{thmlogistic} \label{thm:log-reg}
Consider the logistic regression model of \Cref{eq:log-reg} with the domain $\Theta$ 
of the unknown regressor being the ball of radius $r$, for some universal constant $r>0$, 
and suppose that \Cref{as:for-log-reg} holds. 
Assume that $0< \eps < \eps_0$ for a sufficiently small constant $\eps_0$. 
There is a single-pass streaming algorithm that uses 
$n =  (d^2/\eps) \,\polylog\left(d,1/\eps, 1/\tau \right)$ samples, 
runs in time $ n d \, \polylog (d,1/\eps, 1/\tau )$, 
uses memory $d \, \polylog (d,1/\eps, 1/\tau )$, 
and returns a vector $\widehat{\theta} \in \R^d$ such that 
$\| \widehat{\theta} - \theta^* \|_2 = O( \sqrt{\eps})$ 
with probability at least $1-\tau$. 
\end{restatable}

\subsection{Byzantine Adversary and Second-order Optimal Point}\label{sec:non_convex}

We now describe the application of our algorithm 
to the setting of robust distributed non-convex optimization.
As before, for a parameter space $\Theta \subset \R^d$, 
a loss function $f: \Theta \times \cZ \to \R^+$, 
and a distribution $D$ over $\cZ$, 
the goal is to approximately minimize $\bar{f}(\theta) = \E_{z \sim D}[f(\theta; z)]$. 
In this section, we consider the case when $D$ is a uniform distribution 
over $mn$ points $\{z_{i,j}: i \in [m], j \in [n]\}$ 
that are distributed over $m$ machines (workers), 
with each machine having access to $n$ samples.
Furthermore, we do not impose convexity constraints on $f$, 
and thus would restrict ourselves to finding a second-order stationary point, 
i.e., a stationary point $\widehat{\theta}$ 
such that the Hessian on $\widehat{\theta}$ is not too negative in any direction.

We now explain the distributed setup in more detail. 
There are $m$ \emph{workers} who have their own private samples, 
and a single \emph{master} machine which is responsible 
for collecting gradient estimates from the workers and updating the candidate vector iteratively. 
Concretely, the $i$-th worker has $n$ samples $\{z_{ij}\}_{j=1}^n$.
The master machine queries all workers with a parameter $\theta \in \Theta$, 
and each $i$-th worker responds with $g_i(\theta)$, 
where $g_i: \R^d \to \R^d$ is defined as follows: 
(i) if the $i$-th worker is honest, then $g_i(\theta)$ is the average of the gradients of $f$ at $\theta$ of their samples, 
i.e., $g_i(\theta):= (1/n)\sum_{j=1}^n \nabla f(\theta; z_{ij})$, and 
(ii) if the $i$-th worker is dishonest, then $g_i(\cdot)$ is an arbitrary function.
In our results, we require only that $(1 - \epsilon)$-fraction of workers are honest.
Recent work of~\cite{yin2019defending} provided an algorithm 
that uses a robust mean estimation algorithm on the gradients as a black-box procedure. 
In particular, the algorithm of \cite{yin2019defending}  requires only an access to the following oracle:
\begin{definition}[$\Delta$-inexact gradient]
We call the vector $v(\theta)$ a $\Delta$-inexact gradient of $\bar{f}$ 
at the point $\theta$ if $\|v(\theta) - \nabla \bar{f}(\theta) \|_2 \leq \Delta$.
\end{definition}
We assume that each worker machine has access to its own samples 
throughout the optimization process, and our goal is to reduce the memory requirement 
of the master machine.
Thus, we will use the  algorithm from \Cref{thm:polylog_passes} to calculate $\Delta$-inexact gradient 
for the master machine, which requires only an oracle access to the gradient estimates $\{g_i(\theta): i \in [m]\}$. 

\begin{assumption} \label{as:for-byzantine-algo} 
Let $\cI \subseteq [m]$ be the set of honest workers with $|\cI|\geq (1- \eps)m$.
\begin{enumerate}
    \item There exists $\delta$ with $0 \leq \eps \leq \delta\leq\delta_0$, for some sufficiently small $\delta_0$, 
    such that for every $\theta \in \Theta$, the set $\{g_i(\theta) \; | \; i \in \cI\}$ is $(C\eps,\delta)$-stable 
    with respect to $\nabla \bar{f}(\theta)$ for a large enough constant $C$.
    \item We assume that $\bar{f}$ is $L$-smooth and $\rho$-Hessian Lipschitz on $\Theta$, 
    i.e., for every $\theta_1,\theta_2 \in \Theta$ we have that 
    $\| \nabla \bar{f}(\theta_1) -  \nabla \bar{f}(\theta_2)\|_2 \leq L \| \theta_1 - \theta_2\|_2$ and 
     $\| \nabla^2 \bar{f}(\theta_1) -  \nabla^2 \bar{f}(\theta_2)\|_2 \leq \rho \| \theta_1 - \theta_2\|_2 \;.$
\end{enumerate}
\end{assumption}

We note that if the samples of honest workers are sampled 
i.i.d.\ from a distribution $P$, then the set $\{g_i(\theta): i \in \cI\}$ 
for a fixed $\theta \in \Theta$ will be stable with respect to $\nabla \bar{f}(\theta)$ 
with high probability, 
provided that the distribution of $\nabla f(\theta;Z)$ satisfies 
mild concentration under $Z\sim P$ and $m$ is sufficiently large. 
Using a standard cover argument with the smoothness properties of $f$, 
this can be extended to all $\theta \in \Theta$. 
We thus obtain the following theorem, under \Cref{as:for-byzantine-algo}.

\begin{theorem}
Suppose that \Cref{as:for-byzantine-algo} holds. Let $m$ denote the number of workers. 
Assume $0<\tau<1$, $\Delta := C' \delta<1$, for $C'$ a sufficiently large constant and define
\begin{align*}
Q := 2 \log \left( \frac{\rho(\bar{f}(\theta_0) - \inf_{\theta \in \R^d}\bar{f}(\theta))}{48 L\tau(\Delta^{6/5}d^{3/5}+ \Delta^{7/5}d^{7/10})   }  \right)\; , 
\quad 
        T_{th} := \frac{L}{ 384 (\rho^{1/2} + L (\Delta^{2/5}d^{1/5}+ \Delta^{3/5}d^{3/10})  } \;.
\end{align*}
There is an algorithm where the master, 
if initialized at $\theta_0$, does $T = \frac{2(\bar{f}(\theta_0) - \inf_{\theta \in \R^d}\bar{f}(\theta))}{3\Delta^2} Q \, T_{th}$ iterations, 
each running in $m d \,\polylog(d,1/\eps,T/\tau) $ time, 
uses $d \, \polylog(d,1/\eps,T/\tau)$ memory, 
and outputs a  vector $\widehat{\theta}$ such that, with probability $1-\tau$, 
$\|\nabla \bar{f}(\widehat{\theta}) \|_2 \leq 4 \Delta $ and 
$\lambda_{min}(\nabla^2 \bar{f}(\widehat{\theta})) \geq - \Delta^{2/5} d^{1/5}$.
\end{theorem}

\medskip

\section{Discussion} \label{sec:conc}
In this work, we gave the first efficient streaming algorithm 
with near-optimal space complexity for outlier-robust high-dimensional mean estimation.
As an application, we also obtained low-space streaming algorithms for a range of 
other robust estimation tasks. Our work is a first step towards understanding the space
complexity of high-dimensional robust statistics in the streaming setting.

Our work suggests a number of open problems.
First, the sample complexity of our mean estimation
algorithm is $\tilde{O}(d^2/\eps^2)$, 
while the information-theoretic optimum (without space constraints!) 
is $\tilde{O}(d/\eps^2)$. 
What is the {\em optimal} sample-space tradeoff? 
A similar question can be asked for the broader tasks of 
covariance estimation and stochastic optimization. 
A more general goal is to characterize 
the tradeoff between space complexity, number of passes, 
and sample size/runtime for other robust high-dimensional statistics tasks, 
e.g., clustering and learning of mixture models.

Finally, another research direction concerns the considered contamination model. 
Throughout this paper, we focused on the TV-contamination model. 
One can consider an even stronger contamination model with an adaptive adversary, 
where the outliers can be completely arbitrary (i.e., not follow any distribution),
and the adversary can additionally control the order in which the points are presented
in the stream.
Is it possible to obtain $\tilde{O}_{\eps}(d)$-space single-pass streaming algorithms
for robust mean estimation in the presence of such an adversary?
While our algorithms can be shown to work in this model
with a poly-logarithmic number of passes,
it is not clear whether a single-pass algorithm
with sub-quadratic space complexity exists in this setting.

\newpage

\bibliographystyle{alpha}
\bibliography{allrefs}

\newpage
\appendix

\section{Omitted Proofs from \Cref{sec:prelim}: Technical Details Regarding Stability} 
\label{sec:omitted_prelims}

\label{sec:technical_details_regarding_stability}

\LemCert*
\begin{proof}
    Let $\dtv(P,G)=\alpha$. By \Cref{fact:dtv-decompose} we can write $P = (1-\alpha)G_0 + \alpha B$. We may assume without loss of generality $\alpha=\eps$, since we can always treat a part of the inliers as outliers. Denoting by $\mu_P, \vec{\Sigma}_P$ the mean and covariance of $P$, and using $\mu_{G_0}, \mu_B, \vec{\Sigma}_{G_0}, \vec{\Sigma}_{B}$ for the corresponding quantities of the other two distributions, we have that
    \begin{align*}
        \vec{\Sigma}_P = (1-\eps) \vec{\Sigma}_{G_0} + \eps \vec{\Sigma}_{B} + \eps(1-\eps) (\mu_{G_0} - \mu_B)(\mu_{G_0} - \mu_B)^T \;.
    \end{align*}
    Letting $v$ be the unit vector in the direction of $\mu_{G_0} - \mu_B$, we have that 
    \begin{align} \label{eq:variancedec}
        1+ \lambda \geq v^T \vec{\Sigma}_{P} v = (1-\eps) v^T \vec{\Sigma}_{G_0} v + \eps v^T \vec{\Sigma}_{B} v + \eps(1-\eps) (v^T(\mu_{G_0} - \mu_B))^2 \;.
    \end{align}
    The second term of the left-hand side is nonzero and the third one is just $\eps(1 - \eps)$ $\| \mu_{G_0} - \mu_B \|_2^2$. We now focus on the first term, which by adding and subtracting $\mu$ (the vector realizing the definition of stability for $G$) can be written as
    \begin{align} \label{eq:varianceterm}
        (1-\eps) \E_{X \sim G_0}[(v^T(X-\mu_{G_0}))^2]  = (1-\eps)\left(\E_{X \sim G_0}[(v^T(X-\mu))^2] - (v^T(\mu-\mu_{G_0}))^2  \right)\;.
    \end{align}
    We note that in the decomposition of \Cref{fact:dtv-decompose}, we can write $G_0(x)=w_0(x)G(x)$ with 
    \begin{align*}
        w_0(x) = \frac{1}{1-\eps}
        \begin{cases}
            P(x)/G(x) \;, &\text{if $G(x)> P(x)$}\\
            1 \;, &\text{otherwise} \;.
        \end{cases}
    \end{align*}
    Letting $h(x):=(1-\eps)w_0(x)$ we have that $h(x) \leq 1$ for all $x$ and $\E_{X \sim G}[h(X)]=1-\eps$, thus $G_0(x)=h(x)G(x)/(\int h(x)G(x)\d x)=:G_h(x)$. Returning to \Cref{eq:varianceterm}, this means that 
    \begin{align} \label{eq:stabilityappl}
        \E_{X \sim G_0}[(v^T(X-\mu))^2] = \E_{X \sim G_h}[(v^T(X-\mu))^2] = v^T\overline{\vec \Sigma}_{h,G} v \geq 1-\frac{\delta^2}{\eps} \;,
    \end{align}
    by applying stability. Similarly, the other term in \Cref{eq:varianceterm} is $(v^T(\mu-\mu_{G_0}))^2 \leq \delta^2$. Putting everything together, \Cref{eq:variancedec} becomes
    \begin{align*}
        1+\lambda &\geq (1-\eps)(1-\delta^2/\eps - \delta^2) + \eps(1-\eps)\| \mu_{G_0} - \mu_B \|_2^2 \\
            &\geq 1-3\delta^2/\eps + (\eps/2)\| \mu_{G_0} - \mu_B \|_2^2 \;,
    \end{align*}
    which yields $\| \mu_{G_0} - \mu_B \|_2 \lesssim \sqrt{\lambda/\eps} + \delta/\eps$. Then, writing $\mu_P = (1-\eps) \mu_{G_0} + \eps \mu_B$ and using stability follows that $\| \mu_P - \mu\|_2 \lesssim \delta + \sqrt{\lambda \eps}$.
\end{proof}

\LemDTVStab*
\begin{proof}
By \Cref{fact:dtv-decompose}, we have the decomposition 
$P = (1-\eps)G_0 + \eps B$, where $G_0(x)=\min\{G(x), P(x) \}/(1-\eps)$. 
We can write $G_0(x)  = w_0(x) G(x)$, where
    \begin{align*}
        w_0(x) = \frac{1}{1-\eps}
        \begin{cases}
            P(x)/G(x) \;, &\text{if $G(x)> P(x)$}\\
            1 \;, &\text{otherwise} \;.
        \end{cases}
    \end{align*}
    To see why the final claim is true, we consider a weight function $w:\R^d \to [0,1]$ 
    such that $\E_{X \sim G_0}[w(X)] \geq 1-\eps$ and examine the adjusted distribution ${G_0}_{w}$. We have that
    \begin{align*}
       {G_0}_{w}(x) = \frac{w(x)G_0(x)}{\int_{\R^d}w(x)G_0(x)\d x  } = \frac{(1-\eps)w(x) w_0(x)G(x)}{\int_{\R^d} (1-\eps)w(x)w_0(x)G(x) \d x} = \frac{h(x)G(x)}{\E_{X \sim G}[h(X)]} = G_h(x)\;,
    \end{align*}
    where we let $h(x):= (1-\eps)w(x)w_0(x)$. We have that $h(x) \leq 1$ point-wise and $\int_{\R^d} h(x)G(x) \d x = (1-\eps) \E_{X \sim G_0}[w(X)] \geq (1-\eps)^2 \geq 1- 2\eps$. Recalling that $G$ is $(2\eps,\delta)$-stable, the conclusion follows.
   \end{proof}

\ClTriangle*
\begin{proof}
We can write
\begin{align} \label{eq:twoterms}
\E_{X \sim G_w}\left[  \|\vec U (X-b)\|_2^2\right] 
&= \E_{X \sim G_w}[\|\vec U (X-\mu)\|_2^2 + \|\vec U (\mu-b)\|_2^2 + 2(X-\mu)^T\vec U^T \vec U(\mu - b)] \notag \\
&= \E_{X \sim G_w}\left[\|\vec U (X-\mu)\|_2^2 \right] + 
\|\vec U (\mu-b)\|_2^2 + 2(\mu_{w,G}-\mu)^T\vec U^T \vec U(\mu - b) \;.
\end{align}
We now focus on the first term. 
Let the spectral decomposition $\vec U^T \vec U = \tr( \vec U^T \vec U ) \sum_{i=1}^d \alpha_i v_i v_i^T$, 
where $\sum_{i=1}^d \alpha_i =1$ and $\alpha_i \geq 0$. We have that 
    \begin{align*}
        \E_{X \sim G_w}\left[\|\vec U (X-\mu)\|_2^2 \right] &= \tr\left(\vec U^T \vec U  \E_{X \sim G_w}[(X-\mu)(X-\mu)^T]\right) 
        = \tr(\vec U^T \vec U ) \sum_{i=1}^d \alpha_i \tr(v_i v_i^T \overline{\vec \Sigma}_{w,G}) \\
        &= \tr(\vec U^T \vec U ) \sum_{i=1}^d \alpha_i v_i^T \overline{\vec \Sigma}_{w,G} v_i
        = \tr(\vec U^T \vec U ) (1 \pm \delta^2/\eps) = \|\bU\|_F^2(1 \pm \delta^2/\eps) \;,
    \end{align*}
    where the second from the end relation is due to stability. 
    Regarding the last term of \Cref{eq:twoterms}, we have that
    \begin{align*}
        |(\mu_{w,G}-\mu)^T \vec U^T \vec U (\mu - b)| 
        &= |\tr(\vec U^T \vec U (\mu - b)(\mu_{w,G}-\mu)^T)|\leq \tr(\vec U^T \vec U ) \|(\mu - b)(\mu_{w,G}-\mu)^T\|_2 \\
        &= \|\bU\|_F^2\|\mu - b\|_2\|\mu_{w,G}-\mu \|_2 \leq \|\bU\|_F^2 \delta \|\mu - b\|_2 \;,
    \end{align*}
    where the last inequality uses stability condition for the mean.
\end{proof}

\CorShift*
\begin{proof}
    Beginning with the upper bound, we have the following inequalities:
    \begin{align*}
        \E_{X \sim G}[w(X)\tilde{g}(X)] &= \E_{X \sim G}[w(X)] \E_{X \sim G_w}[\tilde{g}(X)] \\
        &\leq \E_{X \sim G_w}[\tilde{g}(X)] \tag*{(Using $\tilde{g}(x) \geq 0$ and $w(x) \leq 1$)}\\ 
        &\leq \|\bU\|_F^2 (1+\delta^2/\eps) + \|\vec U (\mu - b)\|_2^2 + 2\delta \|\bU\|_F^2 \|b - \mu \|_2  \\
        &\leq \|\bU\|_F^2 \left(1+\delta^2/\eps + \|b - \mu \|^2_2 + 2\delta\|b - \mu \|_2\right) \;,
    \end{align*}
    where the second inequality from the end uses \Cref{cl:triangle}.    The lower bound is derived similarly:
    \begin{align*}
        \E_{X \sim G}[w(X)\tilde{g}(X)] &= \E_{X \sim G}[w(X)] \E_{X \sim G_w}[\tilde{g}(X)] \\
        &\geq (1-\eps) \E_{X \sim G_w}[\tilde{g}(X)] \\
        &\geq (1-\eps) \|\bU\|_F^2(1-\delta^2/\eps -2 \delta\|b - \mu \|_2)  \;,
    \end{align*}
    where we applied \Cref{cl:triangle} in the last step.
\end{proof}

\section{Omitted Proofs from~\Cref{sec:near_linear_main_body}} \label{sec:omitted_near_linear}

\subsection{Johnson-Lindenstrauss Sketch} \label{sec:omitted_JL}

\clJLCons*
\begin{proof}
    We show that the claim holds for a fixed iteration $t$ with probability $\failp/\outerl$. 
    Recall that $\tilde{g}_t(x)$ from \Cref{alg:near-linear-only}, can be written as 
    \begin{align*}
        \tilde{g}_t(x) = \| \bU_t(x-\mu_t)\|_2^2 = \frac{1}{\innerl}\sum_{j \in [\innerl]} (v_{t,j}^T(x-\mu_t))^2 = \frac{1}{\innerl}\sum_{j \in [\innerl]}(z_{t,j}^T{\vec M}_t(x_i-\mu_t))^2\;.
    \end{align*}
    Applying \Cref{fact:jl} with $\gamma=\failp/\outerl$ and $u_i = {\vec M}_t(x_i-\mu_t)$, 
    gives that choosing $\innerl = C \log(n \outerl/\failp)$ 
    suffices to guarantee that $\tilde{g_t}(x_i)/{g}_t(x_i) \in [0.8,1.2]$ 
    for every $x_i$ with probability $1-\failp$.

    We now show the second claim. 
    Again, fix a $t \in [\outerl]$. 
    Consider the orthonormal base $\{e_i\}_{i=1}^d$ of $\R^d$. 
    We apply \Cref{fact:jl} with $\gamma=\failp/\outerl$ and $u_i = {\vec M}_t e_i$, $i\in [d]$. 
    This yields that choosing $\innerl = C\log(d \outerl/\failp)$, we get that for all $i \in [d]$:
    \begin{align*}
        \frac{1}{\innerl} \sum_{j=1}^{\innerl} \tr(z_{t,j} z_{t,j}^T {\vec M}_t e_i e_i^T {\vec M}_t^T)  = \frac{1}{\innerl} \sum_{j=1}^{\innerl}(z_{t,j}^T u_i)^2
        =[0.8,1.2]\frac{1}{\innerl} \sum_{j=1}^{\innerl}\|u_i\|^2
        =[0.8,1.2] \tr({\vec M}_t^T {\vec M}_t e_i e_i^T )\;,
    \end{align*}
    with probability $1-\failp/\outerl$. Summing all these inequalities for $i=1,\ldots, d$ and noting that $\sum_{i=1}^d e_i e_i^T = \vec I_d$ gives that
    \begin{align*}
        \frac{1}{\innerl} \sum_{j=1}^{\innerl} \tr(z_{t,j} z_{t,j}^T {\vec M}_t^T {\vec M}_t) = [0.8,1.2] \tr({\vec M}_t^T {\vec M}_t) \;,
    \end{align*}
    which precisely means that $\frac{1}{\innerl}\sum_{j=1}^{\innerl}\|v_{t,j}\|^2 = [0.8,1.2] \| \bM_t\|_F^2$. To have both claims hold simultaneously, we can just apply \Cref{fact:jl} for all $n+d$ points, giving the result.
\end{proof}

\subsection{Proof of \Cref{lem:filterguarantee}} \label{sec:omitted_filter}

It is more useful to think of the algorithm in the following equivalent form.

\begin{algorithm}[h]  
    
    \caption{Downweighting Filter} 
      \label{alg:reweighting_equiv}
    \begin{algorithmic}[1]  
      \Function{DownweightingFilter}{$P,w,\tilde{\tau}, R, T,\ell_{\max}$}
      \State $r \gets C  dR^{2+4\log d}$.  %
      \State $w' \gets w, \ell \gets 1$.
      \While{$\E_{X \sim P}\left[ w'(X) \tilde{\tau}(X)   \right]> 2T$ and $\ell\leq \ell_{\max}$}   \label{line:loopcondition} %
            \State $\ell \gets \ell + 1$            \label{line:ell_update}    
            \State $w'(x) \gets w(x)(1 - \tilde{\tau}(x)/r)$\label{line:update} 
        \EndWhile
        \label{line:end_while_line} 
        \State \textbf{return} $w'$.
    \EndFunction  
    \end{algorithmic}  
  \end{algorithm}
  
\Lemfilterguarantee*
\begin{proof}
    We show correctness of \Cref{alg:reweighting_equiv}. We denote by $w_\ell$ the weight function at the $\ell$-th iteration of the filter, which is of the form $w_\ell(x) = w(x)(1-\tilde{\tau}(x)/r)^\ell$ for every $x \in \R^d$. To show the first claim, we fix an iteration $\ell$ for which the algorithm has not stopped yet and examine the loss in weight between that iteration and the $(\ell+1)$-th iteration. From the update rule $w_{\ell+1}(x) = w_\ell(x) (1-\tilde{\tau}(x)/r)$ we get that $w_{\ell}(x) - w_{\ell + 1}(x) =w_\ell(x)\tilde{\tau}(x)/r$. Thus, the weight removed in that iteration from the good  distribution is
    \begin{align*}
        (1-\eps)\E_{X \sim G}[w_{\ell}(X) - w_{\ell + 1}(X)] &= \frac{1-\eps}{r} \E_{X \sim G}[w_\ell(X)\tilde{\tau}(X)] \leq \frac{1}{r}T  \;, 
    \end{align*}
    while the weight removed from the bad distribution is
    \begin{align*}
        \eps \E_{X \sim B}[w_{\ell}(X) - w_{\ell + 1}(X)] &= \frac{1}{r} \eps \E_{X \sim B}[w_\ell(X)\tilde{\tau}(X)] \\
         &= \frac{1}{r} \left( \E_{X \sim P}[w_\ell(X)\tilde{\tau}(X)] -  (1-\eps)\E_{X \sim G}[w_\ell(X)\tilde{\tau}(X)]  \right) \\
         &> \frac{1}{r}T \;,
    \end{align*}
    where the last inequality uses that $(1-\eps)\E_{X \sim G}[w_\ell(X)\tilde{\tau}(X)] \leq (1-\eps)\E_{X \sim G}[w(X)\tilde{\tau}(X)] \leq T$ and the fact that since the algorithm has not terminated in the $\ell$-th iteration it must be true that $\E_{X \sim P}[w_\ell(X)\tilde{\tau}(X)] > 2T$. This completes the proof of the first claim.
    
    Regarding runtime, it suffices to show that for any $\ell>\frac{r}{e T}$, $\E_{X \sim P}\left[ w_\ell(X) \tilde{\tau}(X)   \right] \leq  2T$. This follows from the inequalities 
    \begin{align*}
        \E_{X \sim P}\left[ w_\ell(X) \tilde{\tau}(X)   \right] &= \E_{X \sim P}\left[ w(X)(1-\tilde{\tau}(X)/r)^\ell \tilde{\tau}(X)   \right] \\
        &\leq \E_{X \sim P}\left[ w(X)\exp(-\ell\tilde{\tau}(X)/r) \tau(X)   \right] \\
        &\leq \frac{r}{e \cdot \ell}\E_{X \sim P}\left[ w(X)   \right] \leq \frac{r}{e \cdot\ell} \leq T \;,
    \end{align*}
    where we used the fact that $xe^{-\alpha x} \leq 1/(e\cdot \alpha)$ for all $x \in \R$.
    By noting that $w(x)(1-\tilde{\tau}(x)/r)^\ell$ is monotonically decreasing as $\ell$ grows, we can improve the running time  by using a binary search implementation. This gives the logarithmic guarantee of our statement.
\end{proof}

\subsection{Proof of \Cref{lem:goodpartsmall}}  \label{sec:omitted_proof_of_lemma}
We state and prove a more general version of \Cref{lem:goodpartsmall} so that it can be also used in \Cref{sec:low-memory-main}. The difference is that we allow the scores to center points using a vector different from the true mean $\mu_t$ of $P_t$, so long as this vector is $O(\delta)$-close to $\mu_t$ in Euclidean norm. \Cref{lem:goodpartsmall} is obtained by using \Cref{lem:goodpartsmall_general} below with $\widehat{\mu}_t = \mu_t$.
\begin{restatable}{lemma}{lemGoodpartsmall_general} \label{lem:goodpartsmall_general}
    Consider the setting of \Cref{alg:near-linear-only} and the deterministic \Cref{cond:near-linear}. Moreover, let $\hat{\mu}_t$ be any vector in $\R^d$ with $\|\hat{\mu}_t-\mu_t\| = O(\delta)$ and define the functions 
    \begin{align} \label{eq:score_functions}
        &f_t(x) := \|\bM_t(x-\hat{\mu}_t)\|_2^2 ,  &&\tilde{f}_t(x) := \|\bU_t(x-\hat{\mu}_t)\|_2^2\\
        &h_t(x) := f_t(x) \1\{f_t(x) > C_3 \| \bM_t \|_F^2 \lambda_t/\eps \} , &&\tilde{h}_t(x) := \tilde{f}_t(x) \1\{\tilde{f}_t(x) > C_3 \| \bU_t \|_F^2 \widehat{\lambda}_t/\eps \} \;. \notag
    \end{align}
    We have that $\E_{X \sim G}[w_{t}(X)h_t(X)]$ and $ \E_{X \sim G}[w_{t}(X)\tilde{h}_t(X)]$ are bounded from above by $c \lambda_t \|\vec M_t\|_F^2 $     for some constant $c$ of the form $c=C/C_2$, where $C_2$  is the constant used in \Cref{line:lambda} and $C$ is some absolute constant.
\end{restatable}

We  prove the result by using the facts from \Cref{sec:more_stabilityfacts}. 
For brevity, we will prove the results for $\tilde{h}_t$ by using the functions $\tilde{f}_t$ and the matrix $\bU_t$; the results for $h_t$ would follow by replacing $\tilde{f}_t$ and $\bU_t$ with $f_t$ and $\bM_t$ respectively and using that $\|\bU_t\|_F$ is close to $\|\bM_t\|_F$ (\Cref{as:frobenious_near_linear} of the deterministic condition).
We begin with \Cref{lem:smallsetweight}, which is a generalization of the following implication of stability: The $(\epsilon,\delta)$-stability of a distribution $G$ implies that $\E_{X \sim G}[(v^T(X-\mu))^2 \1\{X \in L\}] \leq 3\delta^2/\eps$ for any set $L$ with mass $\pr_{X \sim G}[X \in L] \leq \eps$ (see, for example, Proposition C.3 of \cite{pensia2020robust}).
The following lemma generalizes this to having a matrix in place of $v$.

\begin{restatable}{lemma}{LemSmallSetWeight} \label{lem:smallsetweight}
    Under the setting of \Cref{alg:near-linear-only}, the deterministic \Cref{cond:near-linear}, and using the notation of \Cref{eq:score_functions}, if $L_t \subseteq \R^d$ is a set with $\E_{X \sim G}[w_t(X)\1\{ X \in L_t\}] \leq \eps$, then we have that 
    \begin{align*}
        \E_{X \sim G}[w_t(X)\tilde{f}_t(X)\1\{ X \in L\}] \leq c \lambda_t \|\vec U_t\|_F^2 \;,
    \end{align*}
    for some constant $c$ of the form $c=C'/C_2$, where $C'$ is a sufficiently large constant. 
\end{restatable}

\begin{proof}
    Define the  new weight function $w_t'(x) = w_t(x) \1\{ x \not\in L_t \}$. We have assumed that the distribution $G$ is $(C''\eps,\delta)$-stable. Let $\eps' := C''\eps$ for brevity.  We have the following inequalities, which we explain below.
    \begin{align*}
        \E_{X \sim G}[w_t(X)\tilde{f}_t(X)\1\{ X \in L_t\}] &= \E_{X \sim G}[w_t(X)\tilde{f}_t(X)] - \E_{X \sim G}[w_t(X)\tilde{f}_t(X)\1\{ X \not\in L_t\}] \\
        &= \E_{X \sim G}[w_t(X)\tilde{f}_t(X)] - \E_{X \sim G}[w_t'(X)\tilde{f}_t(X)] \\
        &\leq \|\vec U_t \|_F^2 \left(1+\frac{\delta^2}{\eps'} + \|\widehat{\mu}_t - \mu\|^2_2 + 2\delta \|\widehat{\mu}_t - \mu\|_2 \right) \\
        &- (1-2\eps) \|\vec U_t \|_F^2 \left(1-\frac{\delta^2}{\eps'} - 2 \delta\|\widehat{\mu}_t - \mu\|_2 \right) \\
        &\leq \|\vec U_t \|_F^2\left(3 \frac{\delta^2}{\eps'}  + 4\delta\|\widehat{\mu}_t - \mu\|_2 + \|\widehat{\mu}_t - \mu\|^2_2 \right) \\
        &\leq \|\vec U_t \|_F^2\left(3 \frac{\delta^2}{\eps'}  + 4\delta\left(   \|\widehat{\mu}_t - \mu_t\|_2  + \|\mu_t - \mu\|_2  \right) + 4\|\widehat{\mu}_t - \mu_t\|^2_2 + 4\|\mu_t - \mu  \|_2^2 \right) \\
        &\leq \|\vec U_t \|_F^2 \left(3\frac{\delta^2}{\eps'}  + \delta O(\delta + \sqrt{\eps' \lambda_t}) + O(\delta^2 + \eps' \lambda_t)\right) \\
        &\leq c \|\vec U_t \|_F^2 \lambda_t \;.
    \end{align*}
We note that the third line above follows by applying \Cref{cor:shift} with $\bU= \bU_t$ 
and $b=\mu_{t}$ on both terms of the previous line. The fifth line uses the triangle inequality. 
The sixth line uses the assumption that $\| \widehat{\mu}_t - \mu_t \|_2 = O(\delta)$ 
as well as the certificate lemma (\Cref{lem:certificate}). Note that the required assumption 
$\dtv(P_t,P)=O(\eps)$ from that lemma is satisfied because 
$\E_{X \sim G}[w_t'(X)] \geq \E_{X \sim G}[w_t(X)]-\eps \geq 1- O(\eps)$ 
(see \Cref{cl:invariant-weight,cl:total-var-small}). 
    
    Regarding that last line, we recall \Cref{line:lambda} of \Cref{alg:near-linear-only}, which implies that $\lambda_t \gtrsim C_2\delta^2/\eps$. Thus the terms $\delta^2/\eps'$ can be bounded as $\delta^2/\eps' \lesssim 1/C_2$. Using that, it can be seen that we can choose $c = C'/C_2$ for some constant $C'>0$.
\end{proof}

We are now ready to prove our result.
\begin{proof}[Proof of \Cref{lem:goodpartsmall_general}]
    We first show the following.
    \begin{claim} \label{cl:weight_of_nonzero}
        Consider the setting of \Cref{alg:near-linear-only}, the notation of \Cref{eq:score_functions}, and assume that the deterministic \Cref{cond:near-linear} holds. Let $L_t = \{ x : \tilde{f}_t(x) > C_3 \| \bU_t \|_F^2 \widehat{\lambda}_t/\eps\}$. We have that $\E_{X \sim G}[w_t(X)\1\{X \in L_t \}] \leq \eps$.
    \end{claim}
    \begin{proof} 
        Let $u^* := \arg\max \{ u : \E_{X \sim G}[w_t(X)\1\{\tilde{f}_t(X) > u \}] \geq \eps \}$ and the set $L_t^* = \{ x : \tilde{f}_t(x) > u^* \}$. It suffices to show that $u^* \leq C_3 \| \bU_t \|_F^2 \hat{\lambda}_t/\eps$ (because this would mean that $L_t \subseteq L_t^*$ ). 
        
        By \Cref{lem:smallsetweight}, we have that $\E_{X \sim G}[w_t(X)\tilde{f}_t(X)\1\{X \in L_t^* \}] \leq (C'/C_2) \lambda_t \| \bU_t \|_F^2$.
         If we define the new weights $w_t'(x) = w_t(x)\1\{x \in L_t^* \}$ and use them to normalize the distribution, we get that
        \begin{align*}
            \E_{X \sim {G_{w_t'}}}[\tilde{f}_t(X)] = \frac{1}{\E_{X \sim G}[w_t(X)\1\{X \in L_t^* \}]} \E_{X \sim G}[w_t(X)\tilde{f}_t(X)\1\{X \in L_t^* \}]
                \leq (C'/C_2) \| \bU_t \|_F^2\frac{\hat{\lambda}_t}{\eps} \;,
        \end{align*}
        where we used that the denominator is $\eps$. The fact that $\E_{X \sim {G_{w_t'}}}[\tilde{f}_t(X)] \leq (C'/C_2)  \| \bU_t \|_F^2 \hat{\lambda}_t/\eps$ means that at least one point in $L_t^*$ has $\tilde{f}_t(X)\leq (C'/C_2) \| \bU_t \|_F^2\hat{\lambda}_t/\eps$, which shows that   $u^* \leq (C'/C_2) \| \bU_t \|_F^2\hat{\lambda}_t/\eps$. Since the algorithm uses the value $C_3=(C'/C_2)$, the proof is completed.
    \end{proof}

    \Cref{lem:goodpartsmall_general} now follows by combining \Cref{cl:weight_of_nonzero} with \Cref{lem:smallsetweight}:
    \begin{align*}
        \E_{X \sim G}[w_t(X)\tilde{h}_t(X)] &=  \E_{X \sim G}\left[w_t(X)\tilde{f}_t(X)\1\left\{\tilde{f}_t(X) >  C_3\| \bU_t \|_F^2 \frac{\hat{\lambda}_t}{\eps}\right\} \right] \\
        &\leq (C'/C_2) \lambda_t \| \bU_t \|_F^2  \leq 2(C'/C_2) \lambda_t \| \bM_t \|_F^2 \;,
    \end{align*}
    where the last inequality is due to \Cref{as:frobenious_near_linear} of \Cref{cond:near-linear}. Letting $C=2C'$ we have that $ \E_{X \sim G}[w_t(X)\tilde{h}_t(X)] \leq C/C_2$ as claimed.
    As mentioned earlier, the same techniques lead to a similar bound on $\E_{X \sim G}[w_t(X)h_t(X)]$.
\end{proof}

\subsection{Proof of \Cref{claim:naivepruning}} \label{sec:omitted_naiveprune}

\NAIVEPRUNE*

We can use the following well-known fact (see, e.g.,  \cite{dong2019quantum}, for a proof):
\begin{restatable}{fact}{NAIVEPRUNE2}\label{fact:naive_prune}
    There is an algorithm $\mathrm{NaivePrune}$ with the following guarantees. Let $\eps' < 1/2$, and $\failp>0$. Let $S \subset \R^d$ be a set of $n$ points so that  there exists a $\mu \in \R^d$ and a subset $S' \subseteq S$ so that $|S'| \geq (1-\eps')n$, and $\|x - \mu\|_2 \leq R$ for all $x \in S'$. Then $\mathrm{NaivePrune}(S,R,\failp)$ runs in time $O(nd \log(1/\failp))$, uses memory $O(nd)$, and with probability $1-\failp$ outputs a set of points $T \subset S$ so that $S' \subseteq T$, and $\|x - \mu\|_2 \leq 4R$  for all $x \in T$.
\end{restatable}
\begin{proof}[Proof of \Cref{claim:naivepruning}]
The estimator draws a set $S=\{X_1,\ldots, X_k\}$ of $k$ samples from the distribution $P$. Letting $Y_i:= \1\{ \|X_i-\mu\|_2 > R \}$ we have that $\E[Y_i]\leq 2\eps \leq 1/5$. Thus, using Hoeffding bound,
\begin{align*}
    \pr\left[\frac{1}{k}\sum_{i=1}^k Y_i > 1/4 \right] \leq \pr\left[\frac{1}{k}\sum_{i=1}^k Y_i < \E\left[ \frac{1}{k}\sum_{i=1}^k Y_i\right] + 0.05 \right]  \leq 2 e^{-2k(0.05)^2} \;.
\end{align*}
Choosing $k=200\log(2/\failp)$ makes this probability at most $\failp$. Conditioning on that event, the fraction of points outside the ball is at most $\eps':=1/4$, thus running  $\mathrm{NaivePruning}(S,R,\failp)$ algorithm of \Cref{fact:naive_prune} and outputting any point from the returned set yields the desired estimator.
    
\end{proof}

\subsection{Omitted Proofs from \Cref{sec:correctness}} \label{sec:omitted_correctness}

\clinvariantweight*
\begin{proof}
    For every iteration, we denote by $\Delta w_t = w_{t} - w_{t+1}$, that is for every point $x \in \R^d$, $\Delta w_t(x) = w_{t}(x) - w_{t+1}(x)$ is the difference between the weights for the two consecutive iterations $t$ and $t+1$. 
    \begin{align*}
        \E_{X \sim G}[w_t(X)] &= \E_{X \sim G}[w_1(X)] - \sum_{i=1}^{t-1} \E_{X \sim G}[\Delta w_i(X)] \\
        &\geq 1- \eps - \sum_{i=1}^{t-1} \E_{X \sim G}[\Delta w_i(X)]   \tag{\Cref{cl:radius}}\\
        &\geq 1- \eps - \frac{\eps}{1-\eps} \sum_{i=1}^{t-1} \E_{X \sim B}[\Delta w_i(X)]  \tag{\Cref{lem:filterguarantee}}\\
        &\geq 1- \eps - \frac{\eps}{1-\eps} \left(\E_{X \sim B}[ w_1(X)]  - \E_{X \sim B}[ w_t(X)] \right)\\
        &\geq 1- \eps - \frac{\eps}{1-\eps}  \\
        &\geq 1- 3\eps \;,
    \end{align*}
    where the last line uses that $\eps\leq1/2$.
The proof of the second conclusion follows from \Cref{cl:total-var-small2} (stated  below).
\end{proof}

\begin{restatable}{claim}{clTotalVarSmall} \label{cl:total-var-small2}
    Let $\eps \leq 1/8$. If $\E_{X \sim G}[w_t(X)] \geq 1-3\eps$, then $\dtv(P_t,P) \leq 9\eps$. 
\end{restatable}

Although we work with discrete distributions (the empirical distributions on the samples) in \Cref{sec:near_linear_main_body}, we prove the claim for continuous distributions because it will be useful in \Cref{sec:low-memory-main}.
\begin{proof}
    By definition $P_t(x) = w_t(x)P(x) / \E_{X \sim P}[w_t(X)]$. Letting $L:= \int_{\R^d} w_t(x)  P(x) \d x = \E_{X \sim P}[w_t(x)]$, we have that
    \begin{align*}
        \int_{\R^d} |P_t(x) - P(x)| \d x = (1-\eps)\int_{\R^d} G(x)  \abs[\Big]{\frac{w_t(x)-L}{L}} \d x  + \eps \int_{\R^d} B(x) \abs[\Big]{\frac{w_t(x)-L}{L}} \d x \;.
    \end{align*}
    We note that $1\geq L \geq (1-\eps)\E_{X \sim G}[w_t(x)] \geq 1-4\eps$ using \Cref{cl:invariant-weight}.  The second term can be bounded as
    \begin{align*}
        \eps \int_{\R^d} B(x) \abs[\Big]{\frac{w_t(x)-L}{L}} \d x \leq \frac{\eps}{L} \int_{\R^d} B(x) (w_t(x) + L) \leq \frac{2\eps}{1-4\eps} \leq 4\eps \;.
    \end{align*}
    For the first term, we have that
    \begin{align*}
        (1-\eps)\int_{\R^d} G(x)  \abs[\Big]{\frac{w_t(x)-L}{L}} \d x &\leq \frac{1-\eps}{L}  \int_{\R^d} G(x) (|1-w_t(x)| + |1-L|) \d x \\
        &\leq \frac{1}{1-4\eps} \left( 1- \E_{X \sim G}[w_t(X)] + 4\eps \right) \leq 14 \eps \;.
    \end{align*}
\end{proof}

\ClPSD*
\begin{proof}
    Using the simple fact that for random variables $X$, $Y$ it holds  $\Var(Y) \geq \E_{X}[\Var(Y|X)]$, we get that 
    \begin{align*}
        \vec \Sigma_t &\succeq (1-\eps) \E_{X \sim G_{w_t}}\left[(X-\mu_{G_{w_t}})(X- \mu_{G_{w_t}})^T\right]\\
                &= (1-\eps) \left( \E_{X \sim G_{w_t}}\left[(X-\mu)(X- \mu)^T\right]  - (\mu_{G_{w_t}}-\mu)(\mu_{G_{w_t}}-\mu)^T \right) \\
                &\succeq (1-\eps)\left( (1-\delta^2/\eps)\vec{I}_d - \delta^2 \vec{I}_d\right)  \tag{by stability  and \Cref{cl:total-var-small}}\\
                &\succeq (1-\eps) (1-2\delta^2/\eps)\vec{I}_d \\
                &\succeq (1-3\delta^2/\eps)\vec{I}_d \;,
    \end{align*}
    where we used that $\eps \leq \delta$. We recall the definition $\vec{B}_t = ( \E_{X \sim P_t}[w_t(X)] )^2 \vec{\Sigma}_t - (1-C_1 \delta^2/\eps) \vec{I}_d$ and bound the first term as follows
    \begin{align*}
         \left(\E_{X \sim P_t}[w_t(X)] \right)^2  \vec{\Sigma}_t  &\succeq (1-3\eps)^2 (1-3\delta^2/\eps)\vec{I}_d  \tag{\Cref{cl:total-var-small}}\\
        &\succeq (1-4\delta^2/\eps - 6\eps - 27 \delta^2 \eps)\vec{I}_d \\
        &\succeq (1-11\delta^2/\eps)\vec I_d\;,
    \end{align*}
    where the last line uses $\eps<1/6$ and $\eps \leq \delta$. Therefore, if we choose $C_1 \geq 22$, we get that $\bB_t \succeq (0.5C_1\delta^2/\eps)\vec{I}_d$.
\end{proof}

\ClRelationTauFirst*

\begin{proof}

    By \Cref{cond:near-linear} we have that for all the $n$ samples, $\tilde{g_t}(x)\geq 0.8 {g}_t(x)$.
    Recall the definitions $\tilde{\tau}_t(x) = \tilde{g}(x) \1\{\tilde{g}(x)>C_3\|\bU_t \|_F^2 \hat{\lambda}_t/\eps \}$ and $ {\tau}_t(x) = {g}(x) \1\{{g}(x)>C_3\|\bM_t \|_F^2{\lambda}_t/\eps \}$. We split into cases based on whether each of $g_t,\tilde{g}_t$ has been zeroed by its thresholding operation:
    \begin{itemize}
        \item If $\tau_t(x)$ has been zeroed, (i.e., $g_t(x)<C_3\|\bM_t \|_F^2{\lambda}_t/\eps$), the claim trivially holds since the left-hand side is zero.
        \item If  none of $\tilde{\tau}_t(x),{\tau}_t(x)$ has been zeroed, then  $\tilde{\tau}_t(x) = \tilde{g}_t(x)$ and ${\tau}_t(x) = {g}_t(x)$, thus the claim holds by the aforementioned fact that $\tilde{g_t}(x)\geq 0.8 {g}_t(x)$.
        \item If $\tilde{\tau}_t(x)$ has been zeroed but ${\tau}_t(x)$ has not, then the worst case  is $g_t(x)=(1/0.8)\tilde{g}_t(x)$. This means that in this case:
        \begin{align*}
            \tau_t(x) \leq \frac{1}{0.8}C_3 \frac{\hat{\lambda}_t}{\eps}|\bU_t \|_F^2 < 3 C_3 \frac{{\lambda}_t}{\eps}|\bM_t \|_F^2\;,
        \end{align*}
        where we used that $\hat{\lambda} \leq 1.2 \lambda_t$ and $\|\bU_t \|_F^2 \leq 1.2\|\bM_t \|_F^2$, due to \Cref{cond:near-linear}.
    \end{itemize}
\end{proof}

\section{Omitted Proofs from \Cref{sec:low-memory-main}}
\label{sec:omitted_low_memory}

\subsection{Omitted Proofs from \Cref{sec:correctness_streaming}}  \label{sec:appendixnewDownweighting}

\Lemfilterguaranteenew*

\begin{proof}
    Let the set $L^* = \{ \ell \in [\ell_{\max}] \; : \; 6T\leq \E_{X \sim P}[w(X)(1-\tilde{\tau}(X)/r)^{\ell}\tilde{\tau}(X)] \leq 18T\}$. The invariant is that throughout \Cref{alg:reweighting2}, the set $L$ maintained has non-zero intersection with $L^*$. This can be seen by examining cases about  $\ell$ in \Cref{line:powerell}. If $\ell \in L^*$, then $\ell$ is kept in $L$. If $\ell \not\in L^*$, then all elements discarded are not members of $L^*$ (for example if $\E_{X \sim P}[w(X)(1-\tilde{\tau}(X)/r)^{\ell}\tilde{\tau}(X)] > 18 T$, then by \Cref{cl:evaluate_cond} $f(\ell)>9T$ and we discard the lower half of $L$). Thus, at the end, $L$ has at most two elements with at least one of them belonging in $L^*$. This element would satisfy $3T \leq f(\ell) \leq 27T$. Thus the algorithm will definitely return some element. On the other hand, any element returned will satisfy $2T<\E_{X \sim P}[w(X)(1-\tilde{\tau}(X)/r)^{\ell}\tilde{\tau}(X)]<54T$. This has already shown part $(ii)$ of the lemma. For part $(i)$, it is more convenient to imagine let $\ell$ increased by one at each step until it reaches the value finally returned by the algorithm and consider the loss in weight between that and the next iteration, exactly as in the proof of \Cref{lem:filterguarantee}. That proof was using only the facts that for all $\ell'\leq \ell$, $2T<\E_{X \sim P}[w(X)(1-\tilde{\tau}(X)/r)^{\ell'}\tilde{\tau}(X)]$ (which we just showed above) and $\E_{X \sim G}[w(X)(1-\tilde{\tau}(X)/r)^{\ell'}\tilde{\tau}(X)]<T$ (which is true by assumption). The reason why $\ell_{\max} = r/T$ suffices is also shown identically to the proof of \Cref{lem:filterguarantee}.
\end{proof}

\subsection{Omitted Proofs from \Cref{sec:cover}} \label{sec:omitted_cover}
 We now focus on showing \Cref{cor:ran_cov_specific}. In order to avoid confusion with the fraction of outliers $\eps$, we use $\eps'$ for our accuracy parameter. We will use a uniform convergence result from \cite{AntBar99} combined with a powerful VC-dimension bound from \cite{goldberg1995bounding} for the class of functions that are computable by a small number of arithmetic operations. \cite{goldberg1995bounding} considers the class of concepts parameterized by $k$ real numbers, $\mathcal{F} = \{ h_a \}_{a\in \R^k}$, for which there exists an algorithm $\cA$ for calculating $h_a(x)$ that takes as input $x,a$ and each line of $\cA$ is one of the following:
 \begin{itemize}
 \item an arithmetic operation $+,-,\times$, and $/$ on two inputs or previously computed values,
 \item a jump to a different line of the algorithm conditioned on whether an input or a previously calculated value is greater than or equal to zero,
 \item output zero or one.
\end{itemize}
The parameters $a$ and the inputs $x$ consist of real numbers, and the model of computation assumed allows for arithmetic operations and the comparisons between reals to be done in constant time.
We refer the reader to \cite[Section 2]{goldberg1995bounding} for more details and relation with algebraic decision trees with bounded depths.
 The result from \cite{goldberg1995bounding} is that $\mathrm{VCdim}(\cF) =O(m k)$ where $k$ is the size of the parameterization and $m$ is the runtime of the algorithm $\cA$. 
 Using the bound on the VC dimension, we have the following result for the uniform convergence:

\begin{restatable}[\cite{goldberg1995bounding, AntBar99}]{proposition}{UniformConv} \label{prop:uniform_theorem}
    Let $\cF$ be a class of functions of the form $\cF = \{ h_a:\R^d \to [0,1] \; | \; a \in \R^k \}$, where for any $(a,x) \in \R^k \times \R^d$, $h_a(x)$ can be computed by an algorithm $\cA$ with runtime $m$ that takes as input $a,x$ and  is allowed to perform conditional jumps (conditioned on equality and inequality of real values) and execute the standard arithmetic operations on real numbers ($+,-,\times, /$) in constant time. For any distribution $D$ on $\R^d$ and any $\eps' \in (0,1)$,  there exist  $N = O\left(\frac{1}{(\eps')^2}(\log(km) + km\log(1/\eps')) \right)$ points $x_1,\ldots, x_N$ in $\R^d$ such that 
    \begin{align*}
        \sup_{h \in \cF}\abs[\Big]{\E_{X \sim D}[h(X)] - \frac{1}{N}\sum_{i=1}^Nh(x_i)} \leq \eps' \;.
    \end{align*}
\end{restatable}
For completeness, we show at the end of this section how this is derived from the statements of~\cite{AntBar99} and \cite{goldberg1995bounding}. We now apply this result to our case. We need to specify a family $\cF$ of functions broad enough to capture every $w_{t+1}\tilde{\tau}_t$ and $w_{t+1}\tau_t$ that could be encountered during the execution of \Cref{alg:streaming}. The factor $r$ used in the statement below is a normalization factor to make sure that the functions are in $[0,1]$.

\begin{lemma} \label{cl:family}
    Consider the setting of \Cref{alg:streaming}. Let $r':= \left( C dR^2 + 1 + C_1\delta^2/\eps  \right)^{C \log d}$. There exists a family $\cF$ of functions from $\R^d$ to $[0,1]$ such that: 
    \begin{enumerate}
        \item For every iteration $t$ of \Cref{alg:streaming}, we have that $\frac{1}{r'}w_{t+1}\tau_t\in \cF$ and $\frac{1}{r'}w_{t+1}\tilde{\tau}_t \in \cF$.
        \item Functions of $\cF$ are parameterized by at most $k=O(d \outerl \max(\innerl,d))$ real numbers, that is, $\cF$ has the form $\cF = \{ h_a : \R^d \to [0,1] \; | \; a \in \R^k \}$.
        \item For every $h_a \in \cF$ and $x \in \R^d$, $h_a(x)$ can be in $m = d \outerl \max(\innerl,d)(dR\eps/\delta^2)^{O(\log d)}$ steps in the model that allows conditional jumps and standard arithmetic operations on real numbers.
    \end{enumerate}
\end{lemma}
\begin{proof}
Let $L'_2:= \max(\innerl,d)$. Every function in our family will be parameterized by $2\outerl+1$ scalars, $\{u_t \in \R \; : \; t \in [\outerl+1]\} \cup \{\ell_t \in \R : t\in [\outerl] \}$, and $(\outerl+1)(L'_2+1) $ vectors in $\R^d$, $\{a_t: t \in [\outerl+1]\}\cup \{v_{t,j}: t \in [\outerl+1], j \in [L'_2]\}$. For brevity, we denote by $\bV$ the tensor in $\R^{(\outerl+1) \times L'_2 \times d}$ having all the vectors $\bV_{t,j}=v_{t,j}$ and by $\bA$ the tensor in $\R^{(\outerl+1) \times d}$ with $\bA_{t}=a_t$, $t\in [\outerl+1]$. Similarly, denote by $u$ the vector $(u_1,\ldots, u_{\outerl+1})$ and let $\ell=(\ell_1,\ldots \ell_{\outerl})$. We define our class to be
\begin{align*}
    \cF = \left\lbrace h_{\ell,u,\bV,\bA} : \R^d \to [0,1]  \; : \;  u \in \R^{\outerl+1}, \ell \in \R^{\outerl}, \bV \in \R^{(\outerl+1) \times L'_2 \times d}, \bA \in \R^{(\outerl+1)\times d}    \right \rbrace \;,
\end{align*}
 which includes all functions of the form $h_{\ell,u,\bV,\bA}(x)= \tilde{h}_{\ell,u,\bV,\bA} \1\{ \tilde{h}_{\ell,u,\bV,\bA}(x) \in (0,1)\}$, where
 \begin{align}
    \tilde{h}_{\ell,u,\bV,\bA}(x) = &\1\left\{ \|x- \hat{\mu}\|_2 \leq 5R \right\} \frac{1}{r}
    \cdot \frac{\sum_{j=1}^{L'_2}(v_{\outerl+1,j}^T(x-a_{\outerl+1}))^2}{L'_2} \1 \left\lbrace \frac{\sum_{j=1}^{L'_2}(v_{\outerl+1,j}^T(x-a_{\outerl+1}))^2}{L'_2} > u_{\outerl+1} \right\rbrace \notag\\
    &\cdot \prod_{t =1}^{\outerl} \left(1-\frac{1}{r}\frac{\sum_{j=1}^{L'_2}(v_{t,j}^T(x-a_{t}))^2}{L'_2} \1\left\lbrace\frac{\sum_{j=1}^{L'_2}(v_{t,j}^T(x-a_t)))^2}{L'_2} > u_t \right\rbrace  \right)^{\ell_t}  \;. \label{eq:second_line}
\end{align}

We note that the radius $r':= \left( C dR^2 + 1 + C_1\delta^2/\eps  \right)^{C \log d}$ is an upper bound on the value that the functions $w_{t+1}\tau_t$ and $w_{t+1}\tilde{\tau}_t$ in \Cref{alg:streaming} can take: For $\tau_t(x)$ we have
\begin{align*}
    \tau_t(x) \leq \|\vec M_t(x-\mu_t)\|_2^2 \leq \|\vec M_t \|_2^2 \|x-\mu_t\|_2^2 \leq \|\vec \Sigma_t\|_2^{2\log d} R^2
    =O(   d R^{2 + 4 \log d})\;, 
\end{align*}
while for $\tilde{\tau}(x)$ we have the bounds
\begin{align*}
    \tilde{\tau}_t(x) &\leq \tilde{g}_t(x) \leq \| \bU_t(x-\mu_t)\|_2^2
    \lesssim R^2 \|\bU_t\|_2^2 \lesssim R^2 \|\bU_t\|_F^2 \\
    &\leq R^2 \frac{1}{\innerl}\sum_{j=1}^{\innerl} \| \bM_t z_{t,j}\|_2^2
    \leq d R^2 \|\bM_t\|_2^2 
    \leq d R^2 \|\bB_t\|_2^{2\log d}  \\
    &\leq d R^2 \left( \|\vec \Sigma_t\|_2 + 1 + C_1\delta^2/\eps  \right)^{2\log d} 
    \leq d  R^2 \left( C R^2 + 1 + C_1\delta^2/\eps  \right)^{2\log d} \\
    &\leq \left( C dR^2 + 1 + C_1\delta^2/\eps  \right)^{O(\log d)} \;.
\end{align*}

We check that $\cF$ can indeed implement the functions $w_{t+1}\tilde{\tau}_t$ used in \Cref{alg:streaming} for any $t \in [\outerl]$: Note that the scores $\tilde{g}_t$ used in the algorithm are  means of the form $\frac{1}{\innerl}\sum_{j=1}^{\innerl}(v_{t,j}^T(x-a_t))^2$. Thus, the first line of \Cref{eq:second_line} implements $\1\{\|x-\hat{\mu}\|_2\leq 5R \}\frac{1}{r} \tilde{\tau}_t$.
The purpose of the second line in \Cref{eq:second_line} is to match the operation of the Downweighting filter, which, in the $t$-th round multiplies $w_t$ with $(1-\tilde{\tau}_t(x)/r)^{\ell_t}$ for some power $\ell_t$. Finally, we note that $w_{t+1}\tau_t$ are implemented in $\cF$ by taking $v_{t,j}$ to be the rows of the matrix $\bM_t$ (this is why we need the sums to be on $\innerl' = \max(\innerl,d)$ terms in \Cref{eq:second_line}).

We need to specify the arithmetic complexity $m$ and the dimension of the parameterization $k$ of our family $\cF$. For the first, we have that for any $h\in \cF$ and $x \in \R^d$, the value $h(x)$ can be computed using $O(\outerl  d \innerl' \ell_{\max})$
standard arithmetic operations and jumps, where $\ell_{\max}$ is the maximum exponent that $\ell_t$ can have and is set to be $\ell_{\max}:=\left( \frac{dR}{\delta^2/\eps} \right)^{C \log d} $ in \Cref{line:set_ell_max} of \Cref{alg:streaming}. The $d\innerl'$ comes from the computation of the means $(1/\innerl')\sum_{j=1}^{\innerl'}(v_{t,j}^T(x-a_t))^2$ and the $\outerl$ comes from the fact that we have $\outerl$ factors in the expression of $h$.

Regarding the other parameter $k$, we have that every $h \in \cF$ is parameterized by $O(\outerl)$ scalars and $O(\outerl \innerl')$ $d$-dimensional vectors. Thus, $k=O(\outerl \innerl' d)$.
\end{proof}

We are now ready to prove \Cref{cor:ran_cov_specific}.
\CorRanCovSpecific*
\begin{proof}[Proof of \Cref{cor:ran_cov_specific}] 
    We use \Cref{prop:uniform_theorem} for the family $\cF$ of \Cref{cl:family} and plug  the upper bounds for the arithmetic complexity $m$ and the dimension of the parameters $k$. \Cref{prop:uniform_theorem} 
    states that $N$ can be chosen to be a multiple of  
    \begin{align*}
        \frac{1}{\eps'^2}(\log(km) + km\log(1/\eps')) \;.
    \end{align*}
Taking the much looser bound $N = \Theta(\frac{km}{\eps'^3})$ suffices for our purposes. 
Plugging in $k=O(d \outerl \max(\innerl,d))$, $m = d \outerl \max(\innerl,d)(dR\eps/\delta^2)^{O(\log d)}$ 
from \Cref{cl:family}, we get 
$km =d^2\outerl^2\max(d^2,\innerl^2) (dR\eps/\delta^2)^{O(\log d)} \lesssim d^4\outerl^2\innerl^2(dR\eps/\delta^2)^{O(\log d)}$.
\end{proof}

For completeness, we provide the proof of \Cref{prop:uniform_theorem}.
\begin{proof}[Proof of \Cref{prop:uniform_theorem}]
    We derive the result from the statements of~\cite{AntBar99} without explaining all of the definitions of the notions involved. Please see \cite{AntBar99} for more details.  Applying Theorem 17.7 \cite{AntBar99} with the loss function $\ell_h(x,y) = h(x)$ we obtain 
    \begin{align}\label{eq:uniform}
        \pr\left[ \sup_{h \in \cF}\abs[\Big]{\E_{X \sim D}[h(X)] - \frac{1}{N}\sum_{i=1}^Nh(X_i)} > \eps' \right] \leq 4 \cN_1(\eps'/8, \cF,2N)\exp(-\eps'^2N/32) \;,
    \end{align}
    where the probability is taken over a set of $N$ i.i.d.~points $X_1,\ldots, X_N$ drawn from $D$.

    To bound from above the covering number $\cN_1(\eps'/8, \cF,2N)$, we use Theorem 18.4 from \cite{AntBar99} which gives that $\cN_1(\eps'/8, \cF,2N) \leq e(d'+1)(16e/\eps')^{d'}$ where $d' = \mathrm{Pdim}(\cF)$ is the pseudo-dimension of $\cF$. From that, we conclude that choosing any 
    \begin{align*}
        N>\frac{32}{\eps'^2}\log(4e(d'+1))+\frac{32d'}{\eps'^2}\log(16e/\eps')
    \end{align*}
    makes the probability in \Cref{eq:uniform} less than 1.
    
    It remains to bound $d'$ from above, which can be done as follows. Define the \emph{subgraph class} associated to the family $\cF$
    \begin{align*}
        \cB_{\cF} :=  \left\lbrace B_h \; | \; h \in \cF \right\rbrace \;,
    \end{align*}
    where for any $h \in \cF$, $B_h:\R^{d+1} \to \{0,1\}$ is defined as $B_h(x,y) = \1\{ h(x) \geq y \}$.
    The pseudo-dimension is defined to be $\mathrm{Pdim}(\cF) = \mathrm{VCdim}(\cB_{\cF})$ (see Section 11.2 in \cite{AntBar99}). By Theorem 2.3 in \cite{goldberg1995bounding}, we have that $\mathrm{VCdim}(\cB_{\cF}) = O(km)$ since it $\cB_\cF$ functions that are parametrized by vectors of $\R^k$ (same as for family $\cF$) and the functions of $B_h(x,y)$ can be computed using at most $m+2$ operations ($m$ to compute $h$ and two to do the comparison with $y$ and threshold). Putting everything together, it suffices to choose   
    \begin{align*}
        N = C \frac{1}{\eps'^2}(\log(km) + km\log(1/\eps'))
    \end{align*}
in order to make the probability in \Cref{eq:uniform} less than 1. In that case, by probabilistic argument, there exists at least one set of $N$ points satisfying the desired event.
\end{proof}

\RelationClaim*
\begin{proof}
By \Cref{cond:streaming} we have that for all the $N$ samples of the cover, $\tilde{g_t}(x)\geq 0.2 {g}_t(x) - 0.8 (\delta^2/\eps^2) \|\bM_t \|_F^2$.
    Recall the definitions $\tilde{\tau}_t(x) = \tilde{g}(x) \1\{\tilde{g}(x)>C_3\|\bU_t \|_F^2\hat{\lambda}_t/\eps \}, {\tau}_t(x) = {g}(x) \1\{{g}(x)>C_3\|\bM_t \|_F^2{\lambda}_t/\eps \}$. We split into cases based on whether each of $g_t,\tilde{g}_t$ has been zeroed by their thresholding operation:
    \begin{itemize}
        \item If $\tau_t(x)$ has been zeroed, (i.e., $g_t(x)<C_3\|\bU_t \|_F^2{\lambda}_t/\eps$), the claim trivially holds since the left-hand side is zero.
        \item If  none of $\tilde{\tau}_t(x),{\tau}_t(x)$ has been zeroed, then  $\tilde{\tau}_t(x) = \tilde{g}_t(x)$ and ${\tau}_t(x) = {g}_t(x)$, thus the claim holds by the aforementioned fact that $\tilde{g_t}(x)\geq 0.2 {g}_t(x) -  0.8(\delta^2/\eps^2)\|\bM_t \|_F^2$.
        \item If $\tilde{\tau}_t(x)$ has been zeroed but ${\tau}_t(x)$ has not, then the worst case  is $\tilde{g_t}(x) = 0.2 {g}_t(x) -  0.8(\delta^2/\eps^2)\|\bM_t \|_F^2$. This means that in this case:
        \begin{align*}
            \tau_t(x) &\leq \frac{1}{0.2}C_3 \frac{\hat{\lambda}_t}{\eps}\|\bU_t \|_F^2  + 4\frac{\delta^2}{\eps^2}\|\bM_t \|_F^2 \\
            &< 18 C_3 \frac{{\lambda}_t}{\eps}\|\bM_t \|_F^2 + 4\frac{\delta^2}{\eps^2}\|\bM_t \|_F^2\\
            &\leq 18 C_3 \frac{{\lambda}_t}{\eps}\|\bM_t \|_F^2 + \frac{12}{C_2}\frac{\lambda_t}{\eps}\|\bM_t \|_F^2 \;,
        \end{align*}
        where in the second inequality we used that $ \hat{\lambda}_t \leq 3 \lambda_t$ and $\|\bU_t \|_F^2 \leq 1.2 \|\bM_t \|_F^2$ due to \Cref{cond:streaming}, and in the last inequality we used that $\delta^2/\eps < \hat{\lambda}_t/C_2$ and $\hat{\lambda}_t \leq 3 \lambda_t$ (\Cref{cond:streaming} again).
    \end{itemize}
\end{proof}

\begin{remark}[On the choice of $\outerl$ and $\innerl$] \label{remark:choice_of_L}
We comment on how the values for the number of iterations $\outerl$ and $\innerl$ that are used in \Cref{alg:streaming} are derived. First, the derivation of ${\outerl= C \log d \log(dR/(\delta^2/\eps))}$  for large enough constant $C$ is  identical to that of \Cref{sec:correctness} (see \Cref{eq:choose\outerl}).
We will thus focus on $\innerl$. We note that in the proof of \Cref{lem:remaining_thing} we use \Cref{cor:ran_cov_specific} with $\eps' \gtrsim \frac{(\delta^2/\eps)^{{2 \log d}}}{\eps (CdR^2 + 1 + C_1 \delta^2/\eps)^{{C \log d}}}$. This means that the cover $S_{cover}$ of that lemma has size bounded by
\begin{align*}
    |S_{cover}| \leq \frac{1}{\eps'^3}d^4\outerl^2 \innerl^2\left( \frac{dR}{\delta^2/\eps} \right)^{O(\log d)} 
    \lesssim \frac{ (CdR^2 + 1 + C_1 \delta^2/\eps)^{O(\log d)}}{(\delta^2/\eps)^{O(\log d)}} \innerl^2 \;.
\end{align*}
The analog of \Cref{cl:jl-cons}
thus  requires that $\innerl$ is multiple of $\log\left(\frac{|S_{cover}|+d}{\failp}  \right)$, where $\failp$ is the desired probability of failure. Note that we have the following (rough) bounds
\begin{align*}
    \log\left(\frac{|S_{cover}|+d}{\failp}  \right) &\lesssim \log \left( \frac{ (\innerl(CdR^2 + 1 + C_1 \delta^2/\eps)^{O(\log d)}}{\failp (\delta^2/\eps)^{O(\log d)}}  \right)\\
    &\lesssim \log^2(d) \log(CdR^2+1+C_1\delta^2/\eps)\log\left( \frac{1}{\failp \eps} \right) \log(\innerl) \;.
\end{align*}
Thus, we want to choose $\innerl$  such that it holds $\innerl \geq C \log^2(d) \log(CdR^2+1+C_1\delta^2/\eps)\log\left( \frac{1}{\failp \eps} \right)  \log \innerl$. Using the basic fact that for any $a>0$, $x \geq 2a \log a \Rightarrow x \geq a \log x$ with $a=C \log^2(d) \log(CdR^2+1+C_1\delta^2/\eps)\log\left( \frac{1}{\failp \eps} \right) )$, it suffices to choose any $L$ satisfying the following
\begin{align*}
    \innerl \geq { C }\log^2(d) \log\left(CdR^2+1+C_1\frac{\delta^2}{\eps} \right) \log\left( \frac{1}{\failp \eps} \right)\log\left(\log^2(d) \log\left(CdR^2+1+C_1\frac{\delta^2}{\eps} \right)\log\left( \frac{1}{\failp \eps} \right)   \right)   \;.
\end{align*}
We see that the choice in \Cref{alg:streaming} satisfies this condition.
\end{remark}

\subsection{Omitted Proofs from \Cref{sec:concentration}}  \label{sec:omitted_concentration}

\MatrixProp*
\begin{proof}

We have the following:
\begin{align*}
    \bB^p - \prod_{i=1}^p\bB_i  &= \sum_{i=0}^{p-1} \left( \left(\prod_{j=1}^i\bB_j\right)  \bB^{p-i}   -   \left(\prod_{j=1}^{i+1}\bB_j\right)  \bB^{p-i-1}\right) = \sum_{i=0}^{p-1} \left( \left(\prod_{j=1}^i\bB_j\right)  \left(\bB - \bB_{i+1}  \right)\bB^{p-i-1} \right) \;.
 \end{align*}
Using that $\|\bB_j\|_2 \leq (1 + \delta) \|\bB\|_2 $, we obtain the following bound:
\begin{align*}
\left\|\bB^p - \prod_{i=1}^p\bB_i\right\|_2 &\leq \sum_{i=0}^{p-1} \left\| \left(\prod_{j=1}^i\bB_j\right)  \left(\bB - \bB_{i+1}  \right)\bB^{p-i-1}  \right\|_2 \\
&\leq \sum_{i=0}^{p-1} \left( \left(\prod_{j=1}^i\|\bB_j\|\right)  \|\bB - \bB_{i+1}  \|_2 \|\bB\|_2^{p-i-1} \right) \\
&\leq \sum_{i=0}^{p-1} \|\bB\|_2^p (1 + \delta)^i \delta \leq p \delta (1 + \delta)^p \|\bB\|_2^p.
\end{align*}
\end{proof}

 \subsection{Omitted Proofs from \Cref{sec:proof_of_estimator}}  \label{sec:omitted_last}

 \begin{lemma} \label{lem:basic_estimator}
     For any $\delta,\failp \in (0,1)$ and any distribution  $D$ on $\R^d$ with mean $\mu$ and covariance matrix $\vec \Sigma$, there exists an estimator $\widehat{\mu}$ on $n=O\left( ({\tr(\vec \Sigma)}/{\delta^2}) \log(1/\failp)\right)$ i.i.d.~samples from $D$, such that $\| \widehat{\mu} - \mu \|_2 = O(\delta)$. Moreover, this $\widehat{\mu}$ can be computed in time $O(n d\log(1/\failp))$ and using memory $O(d \log(1/\failp))$.
 \end{lemma}
 \begin{proof}
 Let $X_1,\ldots, X_m$ be independent samples from $D$. We first show that the empirical mean $Y:=(1/m)\sum_{i=1}^m X_i$, is $\delta$-accurate with constant probability. 
 \begin{align*}
     \E[ \snorm{2}{ Y - \mu }^2] = \sum_{j=1}^d \E[(Y_j - \mu_j)^2]
     = \frac{1}{m}  \sum_{j=1}^d \vec  \Sigma_{jj} = \frac{\tr(\vec \Sigma)}{m} \;.
 \end{align*}
 By Markov's inequality, we get that
 \begin{align}\label{eq:emp_mean_is_accurate}
     \pr[\|Y-\mu\|_2^2 > \delta^2] \leq \frac{\tr(\vec \Sigma)}{ m \delta^2} \leq \frac{1}{20}\;,
 \end{align}
 where the last inequality is true if we choose  $m={20\tr(\vec \Sigma)}/{\delta^2}$. Having \Cref{eq:emp_mean_is_accurate} at hand, the probability of success of the above estimator can be boosted to $1-\failp$ by using \Cref{claim:naivepruning}. We use that claim with $G$ being the distribution of $Y$, $B=G$, $\eps=1/20$ and $R=\delta$. This completes the proof.
 \end{proof}

 As a corollary of the above, we obtain the estimators $\widehat{\mu}_t$ of \Cref{alg:streaming}.
\FinalStep*
\begin{proof}
We use the estimator of \Cref{lem:basic_estimator} with $\failp/\outerl$ in place of $\failp$. It remains to bound $\tr(\vec \Sigma_t)$. We have that $\dtv(P_t,G)=1-O(\eps)$, thus, by \Cref{fact:dtv-decompose} we can write $P_t = (1-\alpha)G_0 + \alpha B$, with $\alpha=O(\eps)$ and $G_0(x)=h(x)G(x)/(\int h(x)G(x)\d x)$ some weighted version of the inlier's distribution with $\E_{X \sim G}[h(X)]=1-\alpha$ (same argument that we have used before in the proof of \Cref{lem:certificate}). We have that
\begin{align*}
        \vec{\Sigma}_t = (1-\alpha) \vec{\Sigma}_{G_{0}} + \alpha \vec{\Sigma}_{B} + \alpha(1-\alpha) (\mu_{G_0} - \mu_{B})(\mu_{G_0} - \mu_{B})^T \;.
\end{align*}
Due to stability, the first term has $\vec{\Sigma}_{G_{0}} \preceq (1 + \delta^2/\eps)\bI_d$. For the second term we use that
\begin{align*}
    \tr(\vec{\Sigma}_{B}) = \E_{X \sim B}[\tr((X-\mu_B)(X-\mu_B)^T)] = \E_{X \sim B}[\snorm{2}{X-\mu_B}^2] \leq O(R^2) \;.
\end{align*}
We also bound the trace of the last term by $ O(\eps R^2)$. Therefore, we obtain that $\tr(\vec \Sigma_t) \lesssim d(1+\delta^2/\eps) + \eps R^2$.

\end{proof}

\section{Adaptive Choice of Upper Bound on Covariance}
\label{sec:lepski}

 In this section, we show that a simple procedure can be used 
 to make the algorithm adaptive to the scale of covariance 
 (such a procedure is useful for some of our applications in \Cref{sec:applications-main}).

As noted earlier, the definition of stability that we have used so far (\Cref{def:stability1,def:stability2}) 
was designed for distributions with covariance matrix comparable to $\bI_d$. 
In particular, if inliers satisfy $\cov[X] \preceq \bI_d$, 
then our algorithms result in error $O(\sqrt{\eps})$. 
In many practical cases, some of which are encountered in Section~\ref{sec:applications-main}, 
the inliers are much better concentrated, 
satisfying $\cov[X] \preceq \sigma \bI_d$, with $\sigma$ much smaller than $1$. 
In that case, the optimal asymptotic error is  $\Theta(\sigma \sqrt{\eps})$). 
If $\sigma$ is known beforehand, then a simple preprocessing step 
allows our algorithms to obtain the error $O(\sigma \sqrt{\eps})$.
We now describe a procedure using Lepski's method~\cite{lepskii1991problem,birge2001alternative} 
that can adapt to the setting when $\sigma$ is unknown. 
Concretely,  we consider the task of robustly estimating the mean $\mu$ 
of a distribution where inliers have bounded covariance, $\cov[X] \preceq \sigma^2 \bI_d$, 
but $\sigma$ is unknown to the algorithm.

Let $\mathrm{RobustMean}(\tilde{\sigma},\gamma)$ be any black-box robust mean estimation algorithm, 
where $\tilde{\sigma}$ is a guess for an upper bound on the covariance of inliers 
(ideally, we would like to use $\tilde{\sigma}=\sigma$) and $\gamma$ is the probability of failure. 
The procedure below tries different values for $\tilde{\sigma}$ 
in order to find a vector that is as good as the output of $\mathrm{RobustMean}$ 
when run with the best choice of $\tilde{\sigma}=\sigma$. 
The assumption made here is that $\sigma$ belongs in some known interval $[A,B]$. 

As a small note, a more explicit notation would be $\mathrm{RobustMean}(S,\tilde{\sigma},\gamma)$, 
where $S$ is the dataset used, but we omit $S$ because this depends on the data-access model: 
If a streaming model is assumed, then $S$ necessarily has to be different in each call of the algorithm, 
otherwise there is no need for using different datasets.

\begin{algorithm}[h!]  
    \caption{Adaptive search for $\sigma$} 
    \label{alg:lepski}
    \begin{algorithmic}[1] 
    \State \textbf{input:} $A,B,\gamma,r(\cdot)$  
    \State{Denote $\tilde{\sigma}_j := B/2^j$ for $j=0,1,\ldots,\log(B/A)$ and set $\gamma':=\gamma/\log(B/A)$.   }
      \State $J \gets 0$
      \State $\widehat{\mu}^{(0)} \gets \mathrm{RobustMean}(\tilde{\sigma}_0,  \gamma')$
      \While{ $\tilde{\sigma}_j \geq A$ and $\| \widehat{\mu}^{(J)} - \widehat{\mu}^{(j)} \|_2 \leq r(\tilde{\sigma}_J) + r(\tilde{\sigma}_j)$ for all $j =0,1,\ldots,J-1$}  \label{line:while}
      \State $J \gets J+1$.
      \State $\widehat{\mu}^{(J)} \gets \mathrm{RobustMean}(\tilde{\sigma}_J,  \gamma')$
      \EndWhile
      \State $\widehat{J} \gets J-1$\\   
      \textbf{return} $\widehat{\mu}^{(\widehat{J})}$
    \end{algorithmic}  
  \end{algorithm}

\begin{theorem}\label{thm:lepski}
Let $\mu \in \R^d$ $A,B>0$, $\sigma\in [A,B]$, and a non-decreasing function $r:\R^+ \to \R^+$. 
Suppose that $\mathrm{RobustMean}(\tilde{\sigma},\gamma)$ 
is a black-box algorithm which is guaranteed to return a vector $\widehat{\mu}$ 
such that $\|\widehat{\mu}-\mu\|_2\leq r(\tilde{\sigma})$ with probability $1-\gamma$, 
whenever $\tilde{\sigma} \geq \sigma$. Then, \Cref{alg:lepski} returns $\widehat{\mu}^{(\widehat{J})}$ such that, 
with probability at least $1-\gamma$, we have that $\|\widehat{\mu}^{(\widehat{J})} - \mu  \|_2 \leq 3 r( 2\sigma)$. 
Moreover, \Cref{alg:lepski} calls $\mathrm{RobustMean}$ $O(\log(B/A))$ times 
with desired failure probability set to $\gamma/\log(B/A)$ 
and using at most $ O(d \log(B/A))$ additional memory.
\end{theorem}
\begin{proof}
For $j=0,1,\ldots, \log(B/A)$, denote by $\cE_j$ the event that $\| \widehat{\mu}^{(j)} - \mu \|_2 \leq r( \tilde{\sigma}_j)$. 
Let  $J$ be the index corresponding to the value of the unknown parameter $\sigma$, 
i.e., $\tilde{\sigma}_{J+1} \leq \sigma \leq \tilde{\sigma}_{J}$.
    Conditioned on the event $\cap_{j=0}^J \cE_j$, we have that 
    $\| \widehat{\mu}^{(j)} - \mu \|_2 \leq r(\tilde{\sigma}_j)$ for all $j=0,1,\ldots, J$. 
    Using the triangle inequality, this gives that 
    $\|\widehat{\mu}^{(J)} - \widehat{\mu}^{(j)} \|_2 \leq r( \tilde{\sigma}_J ) + r(\tilde{\sigma}_j)$. 
    This means that the stopping condition of  \Cref{line:while} is satisfied on round $J$ and thus, 
    if $\widehat{\mu}^{(\widehat{J})}$ denotes the vector returned by the algorithm, we have that
    \begin{align*}
        \|\widehat{\mu}^{(\widehat{J})} - \widehat{\mu}^{(J)} \|_2 
        \leq r( \tilde{\sigma}_{\widehat{J}}) + r(\tilde{\sigma}_J) \leq 2r(\tilde{\sigma}_J) \leq 2 r(2\sigma)\;,
    \end{align*}
    where the first inequality uses that the condition of  \Cref{line:while}, 
    the second uses that $r$ is non-decreasing and $ \tilde{\sigma}_{\widehat{J}} \leq \tilde{\sigma}_{J}$, 
    and the last one uses that $J$ was defined to be such that $\tilde{\sigma}_{J+1} \leq \sigma \leq \tilde{\sigma}_{J}$ 
    so multiplying $\sigma$ by $2$ makes it greater than $\tilde{\sigma}_{J}$. 
    Using the triangle inequality once more, we get $\|\widehat{\mu}^{(\widehat{J})} - \mu  \|_2 \leq 3 r( 2\sigma)$.
    Finally, by union bound on the events $\cE_j$, the probability of error is upper bounded by 
    $\sum_{j=0}^{J} \gamma' \leq \gamma$. 
    The additional memory requirement of this algorithm is to store $\{\widehat{\mu}_j: j \in \{0,\dots,\log(B/A)\} \}$.
\end{proof}

We now state the implications that \Cref{thm:lepski} has for \Cref{alg:near-linear-only,alg:streaming}, given in \Cref{sec:near_linear_main_body,sec:low-memory-main}:
\begin{corollary}
Let $A,B>0$. In the setting of \Cref{thm:polylog_passes}, 
let $\sigma>0$ be such that the  scaled version $S' = \{x/\sigma\; : \; x \in S\}$ 
of the dataset $S$ is $(C\eps,\delta)$-stable with respect to $\mu/\sigma$. 
Assuming that $\sigma \in [A,B]$, there exists an algorithm that given $S,\eps,\delta,\failp,A,B$ (but not $\sigma$), 
accesses each point of $S$ at most $\polylog\left(d,1/\eps,1/\failp,B/A \right)$ times, 
runs in time $n d  \,\polylog\left(d,1/\eps,1/\failp,B/A  \right)$, 
uses additional memory $d\, \polylog\left(d,1/\eps,1/\failp,B/A  \right)$, 
and outputs a vector $\widehat{\mu}$ such that, with probability at least $1-\failp$, 
it holds $\|\mu - \widehat{\mu}\|_2 = O(\sigma\delta)$.
\end{corollary}
\begin{proof}
In order to use the search method of \Cref{alg:lepski}, 
we define the procedure $\mathrm{RobustMean}(\tilde{\sigma},  \gamma)$ to be the following:
\begin{itemize}
\item Let $\tilde{S}= \{ x/\tilde{\sigma} \; : \; x \in S\}$.
\item Let $\tilde{\mu}$ be the vector found by the estimator of \Cref{thm:polylog_passes} 
on $\tilde{S}$ using $\gamma$ for the desired probability of failure.
\item Return  $\tilde{\sigma} \tilde{\mu}$.
\end{itemize}
\Cref{thm:lepski} with $r(\tilde{\sigma})=C'\sigma \delta$, for a sufficiently large $C'>0$, 
implies the correctness. In terms of resources used, \Cref{alg:lepski} calls the robust mean 
estimation algorithm at most $\log(B/A)$ times, 
and thus the running time gets multiplied by $\log(B/A)$. 
We also need to store one vector for each call, thus $d \log(B/A)$ additional memory suffices.
\end{proof}

\begin{corollary} \label{cor:low_memory_adaptive}
Let $A,B>0$. In the setting of \Cref{th:main},  
let $\sigma>0$ be such that the distribution $D'$ of the points $X/\sigma$, 
$X \sim D$  is $(C\eps,\delta)$-stable with respect to $\mu$. 
Assuming that $\sigma \in [A,B]$, there exists an algorithm that given 
\begin{align}\label{eq:sample_complex}
n =O \left( R^2 \max\left(d,\frac{\eps}{\delta^2}, \frac{(1+\delta^2/\eps)d}{\delta^2 R^2}, \frac{\eps^2d}{\delta^4},\frac{R^2\eps^2}{\delta^2}, \frac{R^2\eps^4}{\delta^6} \right)   \polylog\left(d,\frac{1}{\eps},\frac{1}{\failp},R,\frac{B}{A} \right)   \right)  
\end{align}
samples in a stream according to the model of \Cref{def:streaming}, 
and given  the parameters $\eps,\delta,\failp,A,B$ (but not $\sigma$), 
runs in time  $n d  \,\polylog\left(d,1/\eps,1/\failp,R,B/A  \right)$, 
uses additional memory $d\, \polylog\left(d,1/\eps,1/\failp,R,B/A  \right)$, 
and returns a vector $\widehat{\mu}$ such that, 
with probability at least $1-\failp$, it holds $\|\mu - \widehat{\mu}\|_2 = O(\sigma\delta)$.
\end{corollary}

Finally, we note that a similar search procedure can be used 
for designing algorithms that are adaptive to the parameter $\delta$ 
when $\sigma$ is known. However, we will not need this generalization for our applications.

\section{Omitted Details from \Cref{sec:applications-main}} \label{sec:appendix_applications}

\subsection{Proof Sketch of \Cref{thm:covariance_application_better_error}}

We describe how   \Cref{alg:streaming} can be plugged in the algorithm of  \cite{cheng2019faster}. We outline the analysis and describe in more detail only the parts from \cite{cheng2019faster} that need to be changed. The algorithm is Algorithm 1 from \cite{cheng2019faster}, which remains unchanged. This uses Algorithm 2 as a subroutine, which we replace by our estimator of \Cref{alg:streaming}.

Regarding the analysis, the proof in \cite{cheng2019faster} uses two claims that state correctness of the black-box robust mean estimator: Lemma 3.4 and Lemma 3.5. For our case, Lemma 3.4 is replaced by our \Cref{th:main} specialized to bounded covariance distributions (also see part 2 of \Cref{main-thm-intro} which says that the sample complexity of \Cref{alg:streaming} for that case is $\tilde{O}(d^2/\eps)$). 

Lemma 3.5 in \cite{cheng2019faster} also holds when using our estimator. We restate this as a claim below and provide a proof:
\begin{claim} \label{cl:translate_lemma35}  
Let $D$ be a distribution supported on $\R^d$ with unknown mean $\mu^*$ and covariance $\vec \Sigma$. Let $0<\gamma<1$, $0<\eps<\eps_0$ for some universal constant $\eps_0$ and $\delta=O(\sqrt{\tau \eps} + \eps \log(1/\eps))$ for some $\tau = O(\sqrt{\eps})$. Suppose that $D$ has exponentially decaying tails and $\vec \Sigma$ is close to the identity matrix $\| \vec \Sigma - \bI_d\|_2 \leq \tau$. Denote $R:=\sqrt{(d/\eps)(1+\delta^2/\eps)}$. \Cref{alg:streaming} uses
    \begin{align} \label{eq:samplecompl2}
         n =  \tilde{O}\left( R^2 \max\left(d,\frac{\eps}{\delta^2}, \frac{(1+\delta^2/\eps)d}{\delta^2 R^2}, \frac{\eps^2d}{\delta^4},\frac{R^2\eps^2}{\delta^2}, \frac{R^2\eps^4}{\delta^6} \right)  \right)      
     \end{align}
samples drawn from $D$ and outputs  a hypothesis vector $\hat{\mu}$ such that  $\| \hat{\mu} - \mu^* \|_2 = O(\delta)$, with probability $1-\gamma$. Moreover, this is done in $nd\,\polylog(d,1/\eps,1/\gamma)$
 time and $d\,\polylog(d,1/\eps,1/\gamma)$ space.
\end{claim}
\begin{proof}
    Since $D$ has exponentially decaying tails, we know that $D$ is stable with respect to its mean $\mu^*$ and covariance $\vec \Sigma \preceq O(1)\bI_d$ with parameter $\delta=O(\eps \log(1/\eps))$ (this follows from the tails of the distribution and \Cref{def:stability1}). That is, for any weight function $w:\R^d \to [0,1]$ with $\E_{X \sim D}[w(X)] \geq 1-\eps$ we have that 
    \begin{align*}
        \snorm{2}{ \mu_{w,D} -  \mu} \leq \delta \quad \text{and} \quad
        \snorm{2}{\overline{\vec \Sigma}_{w,D} - \vec{\Sigma}} \leq \frac{\delta^2}{\eps} \;.
    \end{align*}
    We claim that $D$ is $(\eps,O(\sqrt{\tau \eps} + \eps \log(1/\eps)))$-stable in the sense of \Cref{def:stability1} (the difference from what written above is that \Cref{def:stability1} uses identity matrix in place of $\vec \Sigma$). This can be seen by using triangle inequality:
    \begin{align}
        \snorm{2}{\overline{\vec \Sigma}_{w,D} - \vec{I}_d } 
        \leq \snorm{2}{\overline{\vec \Sigma}_{w,D} - \vec{\Sigma}} + \snorm{2}{ \vec{\Sigma} - \vec{I}_d} \leq \frac{1}{\eps}\left( \delta + \sqrt{\eps \tau} \right)^2 \;.
    \end{align}
    The proof is concluded by recalling the guarantee of \Cref{alg:streaming} for $(\eps,O(\sqrt{\tau \eps} + \eps \log(1/\eps)))$-stable distributions and using \Cref{cl:radius} for the value of $R$.
\end{proof}
We also note that \cite{cheng2019faster} uses a fast matrix inversion and multiplication procedure for calculating the rotated versions $Y={\hat{\vec \Sigma}_i}^{-1/2}X$ of the samples $X$. In our case, the run-time of our robust mean-estimation procedure exceeds that of these methods, thus we do not need to use them. We can instead use any numerically stable method that has running time up to $\tilde{O}(d^6)$ and approximates the result within error $\poly(\eps \kappa/d)$ (see, e.g., \cite{banks2020pseudospectral}).
Finally, since \Cref{cl:translate_lemma35} is essentially used for the $d^2$-dimensional distributions of the points $Y \otimes Y$, we get the $d^4$ factor in the final sample complexity, as well as the $d^2$ factors in the time and space complexity.

\subsection{Omitted Proofs from Section~\ref{sec:app-optimization}}

\CorRobustGD*

\begin{proof}
This follows by using the estimator of \Cref{cor:low_memory_adaptive} 
in place of $g(\cdot)$ in \Cref{alg:robustGD}. 
The known bounds for $\sigma$, $A \leq\sigma\leq B$ are $A=\beta$ and $B=2\alpha r + \beta$, 
thus $B/A \leq 1+2\alpha r/\beta$. The distribution of the scaled gradients $\frac{1}{\sigma}\nabla {f}(\theta)$ 
is $(C\eps,O(\sqrt{\eps}))$-stable. For these parameters, $n$ from \Cref{eq:sample_complex} 
gives $n=(d^2/\eps)\polylog(d,1/\eps,\tau', 1+\alpha r/\beta)$, 
where $\tau'$ is the desired probability of failure for each call of the estimator. 
Setting $\tau'=\tau/T$ ensures that the estimates of all rounds 
are successful with probability $1-\tau$. Successful estimates of the gradients 
are within $O(\sigma \delta) = O((\alpha \|\theta-\theta^*\|_2 + \beta)\sqrt{\eps})$ 
from the true one in Euclidean norm, thus in every round 
we have an $(\sqrt{\eps}\alpha,\sqrt{\eps}\beta)$-gradient estimation (in the sense of \Cref{def:grad_est}). 
Finally, \Cref{thm:robust-gd} requires the condition $\alpha\sqrt{\eps}<\tau_\ell$. 
Assuming that this is true, that theorem concludes the proof.
\end{proof}

\thmlogistic*
\begin{proof}
The algorithm is that of \Cref{thm:robust-gd} using the estimator 
of \Cref{alg:streaming} in place of $g(\cdot)$ in \Cref{alg:robustGD}. 
The distribution of the gradients is $(C\eps,O(\sqrt{\eps}))$-stable 
because of \Cref{lem:bounded_cov_lemma}. For these stability parameters, 
a sufficient number of samples is  $(d^2/\eps) \,\polylog\left(d,1/\eps, T/\tau \right)$ 
(see \Cref{eq:sample_complex_in_theorem} with $\delta=O(\sqrt{\eps})$ and $R=O(\sqrt{d})$), 
where $T$ is the number of iterations over which take a  union bound. 
It thus remains to specify the parameters $\tau_\ell,\tau_u,k,T$.

Using \Cref{as:for-log-reg}, we can calculate bounds on the parameters $\tau_{\ell},\tau_u$. 
For $\tau_{\ell}$, let $v$ be a unit vector from $\R^d$. 
Let the event $\cE_{v,\theta} := \{(v^T X)^2  \geq c_1 \text{ and } |\theta^T X| \leq  2rC^2/c_2 \}$, 
where $c_1,c_2,C$ are the constants from \Cref{as:for-log-reg}. 
The probability of the complement of this event is
\begin{align*}
    \pr[\cE_{v,\theta}^c] \leq \pr[(v^T X)^2   < c_1] + \pr[|\theta^T X| > 2rC^2/c_2]
    \leq 1-c_2 + c_2/2 \leq 1-c_2/2 \;,
\end{align*}
where the first term is bounded using the anti-concentration property 
and the second is bounded using the concentration property along with Markov's inequality. 
Thus, using the formula of \Cref{eq:hessian} for the Hessian, we have that
\begin{align*}
    v^T \nabla^2 \bar{f}(\theta) v \geq \pr[\cE_{v,\theta}] \E_{X \sim D_x}\left[ \frac{e^{{\theta}^T X}}{(1+e^{{\theta}^T X})^2} (v^TX)^2 \; \middle| \; \cE_{v,\theta} \right] 
    \geq 0.5c_2 \frac{e^{2rC^2/c_2}}{(1+e^{2rC^2/c_2})^2} c_1 \;.
\end{align*}
Regarding the upper bound $\tau_u$, 
using the bounded covariance property 
we get that $ v^T \nabla^2 \bar{f}(\theta) v \leq C^2 \sup_{a \in \R} e^a/(1+e^a)^2= C^2/4$. 
Therefore, we can choose the values
\begin{align*}  
\tau_{\ell} = 0.5c_2 \frac{e^{2rC^2/c_2}}{(1+e^{2rC^2/c_2})^2} c_1 \quad \text{and} \quad \tau_{u}= C^2/4 \;,
\end{align*}
for \Cref{alg:robustGD}.
    The guarantees of our mean estimator imply that 
    we have an $(0,O(\sqrt{\eps}))$-gradient estimator (in the sense of \Cref{def:grad_est}). 
    Regarding the value of $\kappa$, we use \Cref{eq:contraction} with $a =0$. 
    Since that $\tau_{\ell}, \tau_u$ are positive constants, 
    this means that $\kappa$ is bounded away from $1$. 
    Therefore, we have that the factor $1/(1-\kappa)$ appearing 
    in the final error (\Cref{eq:final_error}) is $O(1)$ and the number of iterations from \Cref{eq:T} 
    are upper bounded by $T \lesssim \log_2(\|\theta_0 - \theta^*\|_2/\sqrt{\eps}) \lesssim \log(1/\eps)$, 
    where we used that the radius of the domain $\Theta$ is $r=O(1)$.
\end{proof}

\section{Bit Complexity of \Cref{alg:streaming}}
\label{app:bit-complexity}

Until this point, we have assumed that our algorithms 
could save real numbers exactly in a single memory cell 
and perform calculations involving reals in $O(1)$ time. 
Thus, by saying that \Cref{alg:streaming} uses extra memory 
at most $d \, \polylog(d,R,1/\eps,1/\tau)$, 
we meant that it needs to store only that many real numbers. 
We now describe how the algorithm would work in the most realistic word RAM model, 
where finite precision numbers can be stored in registers of predetermined word size 
and operations like addition, subtraction and multiplication are performed in $O(1)$ time.  
We now show that the previous bound of $d \, \polylog(d,R,1/\eps,1/\tau)$, 
 worsened only by another poly-logarithmic factor, holds for the total number 
of bits that need to be stored.
We begin by clarifying how the input is given to the algorithm.

\begin{definition}[Single-Pass Streaming Model with Oracle Access for Real Inputs]
Let $S$ be a fixed set of points in $\R^d$.
The elements of $S$ are revealed one at a time to the algorithm as follows: 
For each point of $S$ that is about to be revealed, 
the algorithm is allowed to query as many bits as it wants 
from that point with whatever order it wants. 
The process then continues with the next point in the stream. 
Each point of $S$ is presented only once to the algorithm in the aforementioned way.
\end{definition}

In the reminder of this section, we use the same notation as in \Cref{th:main}. 
We assume $R\leq M$ and that $\| \mu \|_2 \leq M$, 
for some $M=(d/\eps)^{\polylog(d/\eps)}$ 
(otherwise, the estimation of the mean with extra memory of the order $d \polylog(d/\eps)$ becomes impossible).
The modified algorithm for this model is the following: 
Every input point $X$ is ignored if found to have norm greater than $2M$. 
Otherwise, it is deterministically rounded to an $X'$ 
so that their difference $X-X':=\eta(X)$ has norm at most $\eta$, 
for some $\eta \leq O(\min\{\delta, \frac{\delta^2}{\eps R}, R  \})$ 
(see below for more on this choice of $\eta$). 
The exact same algorithm as \Cref{alg:streaming} is run on these rounded points. 

\paragraph{Correctness} First, we note the rejection step removes less 
than an $\eps$-fraction of the input, thus the resulting distribution 
has not changed more than $\eps$ in total variation distance from the original one. 
Moreover, the distribution of the rounded points has essentially 
the same stability property required by our theorem. 
Concretely, if we choose the rounding error $\eta$ to be 
$\eta = O(\min\{\delta, \frac{\delta^2}{\eps M},M  \})$, 
then it can be seen (\Cref{lem:stability_rounded} below) 
that the distribution of the rounded points is $(\eps,O(\delta))$-stable 
and $\pr_{X'}[\|X'-\mu\|_2 = O(R)]\geq 1-\eps$, 
which are the only assumptions needed for \Cref{alg:streaming} 
to provide an accurate estimate up to $O(\delta)$ error.

\begin{lemma} \label{lem:stability_rounded}
Fix $0<\eps<1/2$ and $\delta\geq \eps$. 
Let $G$ be an $(\eps,\delta)$-stable distribution with respect to some vector $\mu \in \R^d$ 
and assume $G$ is a distribution such that $\|X-\mu\|_2 \leq M$ almost surely for some $M>0$. 
For any deterministic function $\eta: \R^d \to \R^d$ with $\|\eta(x)\|_2 \leq \eta$ 
for all $x$ in the support of $G$, if $G'$ denotes the distribution of the points $X'=X+\eta(X)$, 
where $X \sim G$, then $G'$ is $(\eps,O(\delta + \eta + \sqrt{\eps \eta M}))$-stable 
with respect to $\mu$.
\end{lemma}
\begin{proof}
We check the two conditions for stability. 
Let a weight function $w:\R^d \to [0,1]$ with $\E_{X \sim G}[w(X)] \geq 1-\eps$ 
and let $\delta'=O(\delta + \eta + \sqrt{\eps \eta M})$. We have that
\begin{align*}
    \| \mu_{w,G'} - \mu \|_2 &\leq \| \mu_{w,G'} - \mu_{w,G}\|_2 + \| \mu_{w,G} - \mu \|_2 \\
    &\leq \snorm{2}{\int_{\R^d} (x + \eta(x)) \frac{w(x)G(x)}{\E_{X \sim G}[w(X)]} \d x  - \mu_{w,G} } + \delta \\
    &\leq \snorm{2}{\int_{\R^d} \eta(x) \frac{w(x)G(x)}{\E_{X \sim G}[w(X)]} \d x } + \delta \\
    &\leq \eta + \delta \leq \delta' \;,
\end{align*}
where first inequality uses the triangle inequality and the second one uses the stability of $G$.
Regarding the second stability condition, we have the following:
\begin{align*}
    \snorm{2}{\overline{\vec \Sigma}_{w,G'} - \bI_d } &\leq  \snorm{2}{\overline{\vec \Sigma}_{w,G'} - \overline{\vec \Sigma}_{w,G}} + \snorm{2}{\overline{\vec \Sigma}_{w,G} - \bI_d }\\
    &\leq \snorm{2}{\int_{\R^d} (x-\mu+\eta(x))(x-\mu+\eta(x))^T \frac{w(x)G(x)}{\E_{X \sim G}[w(X)]} \d x} + \frac{\delta^2}{\eps} \\
    &\leq \snorm{2}{\int_{\R^d} (x-\mu)\eta(x)^T \frac{w(x)G(x)}{\E_{X \sim G}[w(X)]} \d x} + \snorm{2}{\int_{\R^d} (x-\mu)^T\eta(x) \frac{w(x)G(x)}{\E_{X \sim G}[w(X)]} \d x} \\
    &+ \snorm{2}{\int_{\R^d} \eta(x)\eta(x)^T \frac{w(x)G(x)}{\E_{X \sim G}[w(X)]} \d x} +  \frac{\delta^2}{\eps}  \\
    &\leq 2 M \eta + \eta^2 +  \frac{\delta^2}{\eps} \leq \frac{\delta'^2}{\eps}\;,
\end{align*}
where  we used stability of $G$, triangle inequality and the bounds $\|x-\mu\|_2 \leq M$, $\|\eta(x)\|_2 \leq \eta$.
\end{proof}

\paragraph{Total Bits of Memory Used} 
In order for the differences $X-X':=\eta(X)$ to have $\|\eta(X)\|_2 \leq \eta$ for all $X$, 
it is sufficient to round every coordinate to absolute error $O(\eta/\sqrt{d})$. 
Recall that by our assumption on the a priori bound on the norm of the true mean 
and the way that we reject input points of large norm, 
we know that all points surviving will have norm at most $2M$. 
Thus, each coordinate of such points  can be stored in a word of $O(\log(M d/\eta))$ bits 
after being rounded to accuracy $\eta$. Therefore, each $d$-dimensional point 
that the algorithm will need to manipulate can be stored using $d$ registers of size $O( \log(Md/\eta))$. 
However, we need to show that the results of all intermediate calculations can be calculated in low memory.
We show the following result to this end:
\begin{claim}
In the context of \Cref{th:main}, given a stream of $d$-dimensional points, where each coordinate has bit complexity $B$,
\Cref{alg:streaming} can be implemented in a word RAM machine 
using $d \, \polylog(d,1/\eps,1/\failp,R)$ many of registers 
of size $B \, \polylog(d,1/\eps,1/\failp,R)$.
\end{claim}
\begin{proof}[Proof Sketch]
Multiplying two numbers of bit complexity $B_1$ and $B_2$ 
may result in bit complexity $B_1+B_2$. Adding $k$ numbers of bit complexity $B$, 
may make the resulting bit complexity $B+\log(k)$. 
We need to check that every step of the algorithm 
performs calculations that cannot cause the bit complexity to grow by more than poly-log factors.

\Cref{line:naive_est_2} performs only comparisons and counting. 
Regarding \Cref{line:multiplication}: As pointed out at the beginning of \Cref{sec:low-memory-main}, 
the vector $v_{t,j} \gets \widehat{\vec{M}}_t z_{t,j}$ is calculated by multiplying $z_{t,j}$ by $\widehat{\bB}_{t,k}$ 
for $k=1,\ldots, \log d$ iteratively. Consider a single iteration, say the first one. 
Performing $\widehat{\bB}_{t,k} z_{t,j}$ involves calculating $\frac{1}{n} \left(\sum_{x} x x^T\right)z_{t,j}$ 
(see \Cref{sec:itemsindetail}), which can be done as $\frac{1}{n} \sum_{x} x (x^Tz_{t,j})$, 
i.e., calculating the inner products $x^Tz_{t,j}$ first). A single inner product of that form 
is just a sum of $d$ numbers of bit complexity $B$ with appropriate signs, 
thus the bit complexity increases only by $O(\log d)$. 
Finally, multiplying by $x$ and taking the mean over for all of the $x$'s 
can  add only another $O(B + \log(n))$. Since the number of iterations 
of such calculations is $\log d$, the final result has the claimed bit complexity.

Regarding the Downweighting filter (\Cref{alg:reweighting2}), 
there are a couple of places where the weights $w_t$ are involved in calculations. 
We note that \Cref{alg:reweighting}  stores only the counts $\ell_t$, 
which fit in registers of size $\log(\ell_{\max}) = \polylog(d,1/\eps,R)$.
These counts are used to calculate $w_t(x)$ as 
$w_t(x) = \prod_{t' \leq t} (1-\tilde{\tau}_{t'}(x)/r)^{\ell_{t'}}$ whenever there is such a need. 
An exact calculation would require operations of the form $x^y$, 
for some $x \in [0,1]$ and $y \in [\ell_{\max}]$, i.e., exponentiation of a real number.
In fact, as we will show later on, instead of calculating $w_t(x)$ with perfect accuracy, 
it suffices to use an approximate value of $w_t(x) \pm \eta$ for some error $|\eta| < \poly(1/d,1/R,\eps,\failp)^{\log d}$. 
This will allow us to calculate a good enough approximation in 
$\polylog(d,R,1/\eps,1/\failp)$ bits as follows: 
we can use exponentiation by squaring algorithm for calculating $w_t(x)$'s 
and round the result in each step to make it fit into  our registers.
We first explain this in more detail below. 
\begin{claim} \label{claim:exponentiation}
Let $x \in [0,1]$, $y \in \Z_+$, and assume that both $x,y$ have bit complexity $B$. 
The power $x^y$ can be calculated up to a rounding error of $2^{-\Omega(B)}$ 
in the word RAM model that uses registers of size $2B$. 
Furthermore, this can be done in $O(B)$ standard arithmetic operations.
\end{claim}
\begin{proof}
We can use exponentiation by squaring: 
This consists of writing $x$ in binary as $b_k \cdots b_0$ 
for $k= \log y$ and calculating the sequence 
$r_{k+1},\ldots, r_0$ as $r_{k+1}=1$, $r_i=r^{2}_{k+1} x^{b_i}$ for $i=k,\ldots, 0$. 
We assume every $r_i$ gets rounded to $2B$ bits.
Because of the rounding, we incur error $2^{-2B}$ in each round. 
However, the error of the previous rounds gets amplified, 
since the result of that round (true value plus error) gets squared.
We consider one such iteration to see how that sequence of errors grows: 
In the $t$-th iteration, let $\mathrm{res}_{t-1}$ denote the true result (before rounding) 
from the previous round and $\eta_t$ the rounding error of that round 
(i.e., $r_t = \mathrm{res}_t + \eta_t)$. Then, we have that
\begin{align*}
    \mathrm{res}_{t} + \eta_{t} := (\mathrm{res}_{t-1} + \eta_{t-1})^2 + 2^{-2B}
    \leq \mathrm{res}_{t-1}^2 + \eta_{t-1}^2 + 2  \eta_{t-1} + 2^{-2B}\;,
\end{align*}
where w.l.o.g. we assume that $\mathrm{res}_t \leq 1$ always. 
Thus, the rounding error grows as 
$\eta_t \leq \eta_{t-1}^2 + 2\eta_{t-1} + 2^{-2B} \leq 3 \eta_{t-1} + 2^{-2B}$. 
In the first round, we start with rounding error of $2^{-2B}$. 
Thus, after $k = \log(y) = B$ rounds, 
the final error is $\eta_k \leq  2^{-\Omega(B)}$.
\end{proof}

We continue with examining how fine approximations for $w_t$ are needed. 
First, in \Cref{sec:itemsindetail}, we use the estimator $\widehat{W}_{t} = \E_{X \sim \cU(S_0)}[w_t(X)]$, 
which we require to be $\eta$-close to $\E_{X \sim P}[w_t(X)]$ (\Cref{eq:hoef}) 
for some $\eta > \poly(1/d,1/R,\eps,\failp)$. Therefore, when calculating $w_t$, 
it suffices to round the intermediate results to error $\eta$. 
This would mean using \Cref{claim:exponentiation} with $B = O(\log((1/d,1/R,\eps,\failp))$.

Second, the weights $w_t$ are also used in evaluating the stopping condition of the Downweighting filter. \Cref{line:estimator_black_box} of that filter is implemented using the estimator of \Cref{cl:evaluate_cond}. 
As it can be seen in \Cref{eq:whatwewant}, it suffices to use rounded versions 
of $w_t$ in $(1/n)\sum_{i=1}^N w_t(X_i)\tilde{\tau}_t(X_i)$, 
as long as it does not change the result by an additive factor 
of $ c \widehat{\lambda_t} \| \bU_t \|_F^2$, for a small constant $c$. 
Since  $\tilde{\tau}_t(X_i)=O(d R^{2+4 \log d} )$ (\Cref{eq:bound_on_tau}) 
and $\widehat{\lambda_t} \| \bU_t \|_F^2 > (\delta^2/\eps)^{\Theta(\log d)}$ 
(otherwise, the algorithm has terminated), we can again round $w_t$ 
up to error $\poly(1/d,1/R,\eps)^{\log(d)}$. This means that the results 
of these calculations fit into $\polylog(d,R,1/\eps,R)$ bits. 

Finally, there are two places in \Cref{alg:streaming} 
where we need to simulate samples from the weighted distribution 
$P_{w_t}$: (i) \Cref{line:low-memB1}, whose implementation 
is outlined in \Cref{sec:concentration} and (ii) \Cref{final:estimate}. 
We focus on the first one since the argument for the other case is identical.
To simulate $P_{w_t}$, we use rejection sampling, 
as described at the beginning of \Cref{sec:concentration}, 
with the only difference that we use the rounded versions of the weights $w_t$. 
We are thus essentially simulating samples from a slightly different distribution $P_{\widehat{w}_t}$. 
However, this is close to $P_{w_t}$ in total variation distance, as shown below.

\begin{claim} \label{claim:small_tvd}
Let $P$ be a distribution on $\R^d$ and let $P_w$ denote the weighted version of $P$ 
according to the function $w: \R^d \to [0,1]$, i.e., $P_w(x) = w(x)P(x)/\int_{\R^d} w(x)P(x) \d x$. 
For any $w,\widehat{w}: \R^d \to [0,1]$ such that $\int_{\R^d} w(x)P(x) \d x \geq 1/2$ 
and $\sup_{x \in \R^d}| \widehat{w}(x) - w(x) | \leq \xi$ with $\xi\leq 1/8$, 
it holds that $\dtv(P_{\widehat{w}},P_{w}) \leq 8\xi$.  
\end{claim}
\begin{proof}
First, letting the normalization factors $\widehat{C}:= \int_{\R^d} \widehat{w}(x)P(x) \d x$ 
and $C:= \int_{\R^d} {w}(x)P(x) \d x$, we note that 
$|C-\widehat{C}| \leq \xi$. Letting $\Delta w(x):= \widehat{w}(x) - w(x)$ and 
$\Delta C:= \widehat{C}- C$, we have that
\begin{align*}
    \dtv(P_{\widehat{w}},P_{w}) &= \frac{1}{2} \int_{\R^d} \left| P_{\widehat{w}}(x) - P_{w}(x) \right| \d x 
    = \frac{1}{2} \int_{\R^d} \left| \frac{w(x) + \Delta w(x)}{C + \Delta C} - \frac{w(x)}{C} \right| P(x) \d x \\
    &\leq \frac{1}{2} \int_{\R^d}  \left| \frac{C w(x) + C \Delta w(x) -C w(x) - \Delta C w(x)}{C^2 + C \Delta C}  \right| P(x) \d x \\
    &\leq 4 \int_{\R^d}  \left|  C \Delta w(x) - \Delta C w(x) \right| P(x) \d x 
    \leq 8   \xi \;,
\end{align*}
where in the last line, we first use that $C^2 + C\Delta C \geq 1/4 - \xi \geq 1/8$ 
(since $1/2\leq C \leq 1$ and $0\leq \xi \leq 1/8$)
and then we use that 
$\left|  C \Delta w(x) - \Delta C w(x) \right| \leq |\Delta w(x)| + |\Delta C| \leq 2 \xi$.
\end{proof}

As $N$ (more than one) samples are drawn from $P_t$ in the $t$-th iteration (see \Cref{sec:concentration}), 
we require the joint distribution of these $N$ samples from $P_{w_t}$ and $P_{\hat{w}_t}$ 
to be within total variation $\failp$ (the probability under which the conclusion of \Cref{lem:parts2and3} holds true). 
This bound on the total variation distance implies that \Cref{lem:parts2and3} 
continues to hold for $P_{\hat{w}_t}$ with an additional failure probability of $\failp$. 
To do this, we use \Cref{claim:small_tvd} with $\xi = \Theta(\failp/N)$, 
which means that these rounded weights have bit complexity 
$\Theta(\log \xi) = \polylog(d,R,1/\eps,1/\failp)$. 
\end{proof}

\end{document}